\documentclass{article}
\usepackage[a4paper, margin=1.0in]{geometry}
\usepackage[T1]{fontenc}
\usepackage{dsfont}
\usepackage[normalem]{ulem}
\usepackage{lmodern}
\usepackage{microtype}
\usepackage{graphicx}
\usepackage{standalone}
\usepackage{tikz}
\usetikzlibrary{calc,cd,patterns}
\usepackage[dvipsnames,table]{xcolor}
\usepackage{float}
\usepackage{amsmath}
\usepackage{amsfonts,amssymb,amsthm}
\usepackage{mathtools}
\usepackage{enumerate}
\usepackage{lipsum}
\usepackage[linesnumbered,ruled,vlined]{algorithm2e}
\usepackage{csquotes}
\usepackage[backend=bibtex,style=numeric,natbib,maxbibnames=99]{biblatex}
\usepackage{hyperref}
\usepackage{aliascnt}
\usepackage{aliascnt}

\hypersetup{
  colorlinks=true,
  linkcolor=darkblue,   % internal: sections, theorems, equations
  citecolor=darkred,    % citations
  urlcolor=darkblue,    % URLs (same as links is fine)
}
   
\usepackage{titlesec}

\titleformat{\subsection}
  {\normalfont\large\bfseries\centering}
  {\thesubsection}{1em}{}

\titlespacing*{\subsection}
  {0pt}{3.25ex}{0.6ex}

\newtheorem{theorem}{Theorem}[section]

\newtheorem{theoremA}{Theorem}

\theoremstyle{plain}
\newaliascnt{lemma}{theorem}
\newtheorem{lemma}[lemma]{Lemma}
\aliascntresetthe{lemma}

\newaliascnt{conjecture}{theorem}
\newtheorem{conjecture}[conjecture]{Conjecture}
\aliascntresetthe{conjecture}

\newaliascnt{proposition}{theorem}
\newtheorem{proposition}[proposition]{Proposition}
\aliascntresetthe{proposition}

\newaliascnt{corollary}{theorem}
\newtheorem{corollary}[corollary]{Corollary}
\aliascntresetthe{corollary}

\newaliascnt{problem}{theorem}

\aliascntresetthe{problem}

\theoremstyle{definition}
\newaliascnt{example}{theorem}
\newtheorem{example}[example]{Example}
\aliascntresetthe{example}

\newaliascnt{definition}{theorem}
\newtheorem{definition}[definition]{Definition}
\aliascntresetthe{definition}

\newaliascnt{observation}{theorem}
\newtheorem{observation}[observation]{Observation}
\aliascntresetthe{observation}

\newaliascnt{remark}{theorem}
\newtheorem{remark}[remark]{Remark}
\aliascntresetthe{remark}

\newaliascnt{notation}{theorem}

\aliascntresetthe{notation}

\newtheoremstyle{algorithmstyle}
  {10pt} % Space above
  {10pt} % Space below
  {\normalfont} % Body font
  {} % Indent
  {\bfseries} % Head font
  {} % No Punctuation after head
  {1em} % Space after head
  {\thmname{#1}\thmnumber{ #2}\thmnote{ (#3)}\newline} % Head spec + forced line break
\theoremstyle{algorithmstyle}

\makeatletter
\let\c@myalgorithm\c@algocf
\makeatother

\SetKw{KwRet}{return}
\SetKw{KwTo}{to}
\SetKw{KwDownto}{downto}
\SetKw{KwAnd}{and}
\SetKw{KwBreak}{break}
\SetKw{KwContinue}{continue}
\SetKw{KwTrue}{true}
\SetKw{KwFalse}{false}

\newcommand{\bbN}{\mathbb{N}}
\newcommand{\bbZ}{\mathbb{Z}}
\newcommand{\bbR}{\mathbb{R}}
\newcommand{\bbQ}{\mathbb{Q}}
\newcommand{\bbC}{\mathbb{C}}
\newcommand{\bbF}{\mathbb{F}}
\newcommand{\K}{\mathbf k} % Which letter do we want for the base field

\newcommand{\vect}{\textbf{Vect}}
\newcommand{\simp}{\textbf{SCpx}}
\newcommand{\op}{\text{op}}
\newcommand{\Pers}[1]{\mathrm{Pers}(#1)} %
\newcommand{\Persfp}[1]{\mathrm{Pers}_\text{fp}(#1)} % category of finitely presentable persistence modules
\newcommand{\Persfpb}[1]{\mathrm{Pers}_\text{fp}^\text{b}(#1)} % category of finitely presentable persistence modules

\newcommand{\oto}[1]{\xrightarrow{#1}}

\global\long\def\into{\hookrightarrow}%
\global\long\def\onto{\twoheadrightarrow}%
\newcommand{\iso}{\xrightarrow{\,\smash{\raisebox{-0.5ex}{\ensuremath{\scriptstyle\sim}}}\,}}
\global\long\def\into{\hookrightarrow}%
\global\long\def\onto{\twoheadrightarrow}%

\DeclareMathOperator{\Img}{Im}
\DeclareMathOperator{\rank}{rank}
\DeclareMathOperator\coker{coker}

\DeclareMathOperator\Id{Id}
\newcommand\up[1]{\langle #1 \rangle} %upper set?
\DeclareMathOperator{\fun}{Fun}

\DeclareMathOperator*{\colim}{colim}

\DeclareMathOperator\Hom{Hom}

\DeclareMathOperator{\End}{End}

\DeclareMathOperator\supp{supp}

\DeclareMathOperator{\Rips}{VR}
\DeclareMathOperator{\Alpha}{D}
\DeclareMathOperator{\multicover}{MC}
\newcommand{\bH}{\mathbf{H}}
\newcommand{\Int}{\mathcal{I}\text{nt}}

\newcommand{\module}{V}% V 

\newcommand{\mathsc}[1]{\textup{\textsc{#1}}}
\newcommand{\calA}{\mathcal A}

\DeclareMathOperator{\grid}{Grid}
\newcommand{\gri}{\mathcal{G}}
\DeclareMathOperator{\diam}{diam}
\DeclareMathOperator{\Wong}{Wong}

\newcommand{\bigslant}[2]{%
  \mathchoice
    {\raisebox{.2em}{$\displaystyle #1$}\!\left/\raisebox{-.2em}{$\displaystyle #2$}\right.}
    {\raisebox{.2em}{$\textstyle #1$}\!\left/\raisebox{-.2em}{$\textstyle #2$}\right.}
    {\raisebox{.1em}{$\scriptstyle #1$}\!/\raisebox{-.1em}{$\scriptstyle #2$}}
    {\raisebox{.1em}{$\scriptscriptstyle #1$}\!/\raisebox{-.1em}{$\scriptscriptstyle #2$}}%
}

\DeclareMathOperator*{\argmax}{arg\,max}

\newcommand{\Oc}{\mathcal{O}}

\DeclareMathOperator{\udim}{\underline\dim}
\DeclareMathOperator{\intdim}{\int \udim}
\newcommand{\set}[1]{{\left\{#1\right\}}}
\definecolor{lightred}{rgb}{1, 0.5, 0.5}
\definecolor{lightblue}{rgb}{0.4, 0.5, 1}
\definecolor{darkblue}{rgb}{0.3, 0.3, 0.8}
\definecolor{darkred}{rgb}{0.8, 0.3, 0.4}

\pgfdeclarepatternformonly{thick nw lines}
  {\pgfpoint{-1pt}{-1pt}}{\pgfpoint{8pt}{8pt}}
  {\pgfpoint{7pt}{7pt}}
{
  \pgfsetlinewidth{0.8pt}
  \pgfpathmoveto{\pgfpoint{0pt}{7pt}}
  \pgfpathlineto{\pgfpoint{7pt}{0pt}}
  \pgfpathmoveto{\pgfpoint{-1pt}{1pt}}
  \pgfpathlineto{\pgfpoint{1pt}{-1pt}}
  \pgfpathmoveto{\pgfpoint{6pt}{8pt}}
  \pgfpathlineto{\pgfpoint{8pt}{6pt}}
  \pgfusepath{stroke}
}
\newcommand{\smat}[1]{{\scriptstyle\setlength{\arraycolsep}{2pt}\renewcommand{\arraystretch}{0.8}#1}}

\newcommand{\ignore}[1]{}

\newcommand{\thick}[1]{\mathfrak{t}(#1)}

\title{Computing the Skyscraper Invariant}
\author{Marc Fersztand and Jan Jendrysiak}
\date{March 2026}

\begin{document}

\maketitle

\begin{abstract}

We develop the first algorithms for computing the Skyscraper Invariant [FJNT24]. This is a filtration of the classical rank invariant for multiparameter persistence modules defined by the Harder-Narasimhan filtrations along every central charge supported at a single parameter value.

Cheng's algorithm [Cheng24] can be used to compute HN filtrations of arbitrary acyclic quiver representations in polynomial time in the total dimension, but in practice, the large dimension of persistence modules makes this direct approach infeasible.
 We show that by exploiting the additivity of the HN filtration and the special central charges, one can get away with a brute-force approach. For $d$-parameter modules, this produces an FPT $\varepsilon$-approximate algorithm with runtime dominated by $\Oc( 1/\varepsilon^d \cdot T_{\mathsc{dec}})$, where $T_{\mathsc{dec}}$ is the time for decomposition, which we compute with \textsc{aida} [DJK25]. 

We show that the wall-and-chamber structure of the module can be computed via lower envelopes of degree $d - 1$ polynomials. This allows for an exact computation of the Skyscraper Invariant whose runtime is roughly $\Oc(n^d \cdot T_{\mathsc{dec}})$ for $n$ the size of the presentation of the modules and enables a faster hybrid algorithm to compute an approximation.

For 2-parameter modules, we have implemented not only our algorithms but also, for the first time, Cheng's algorithm.
We compare all algorithms and, as a proof of concept for data analysis, compute a filtered version of the Multiparameter Landscape for biomedical data.
\end{abstract}

\begin{figure}[H]
\centering
\includegraphics[scale = 0.75]{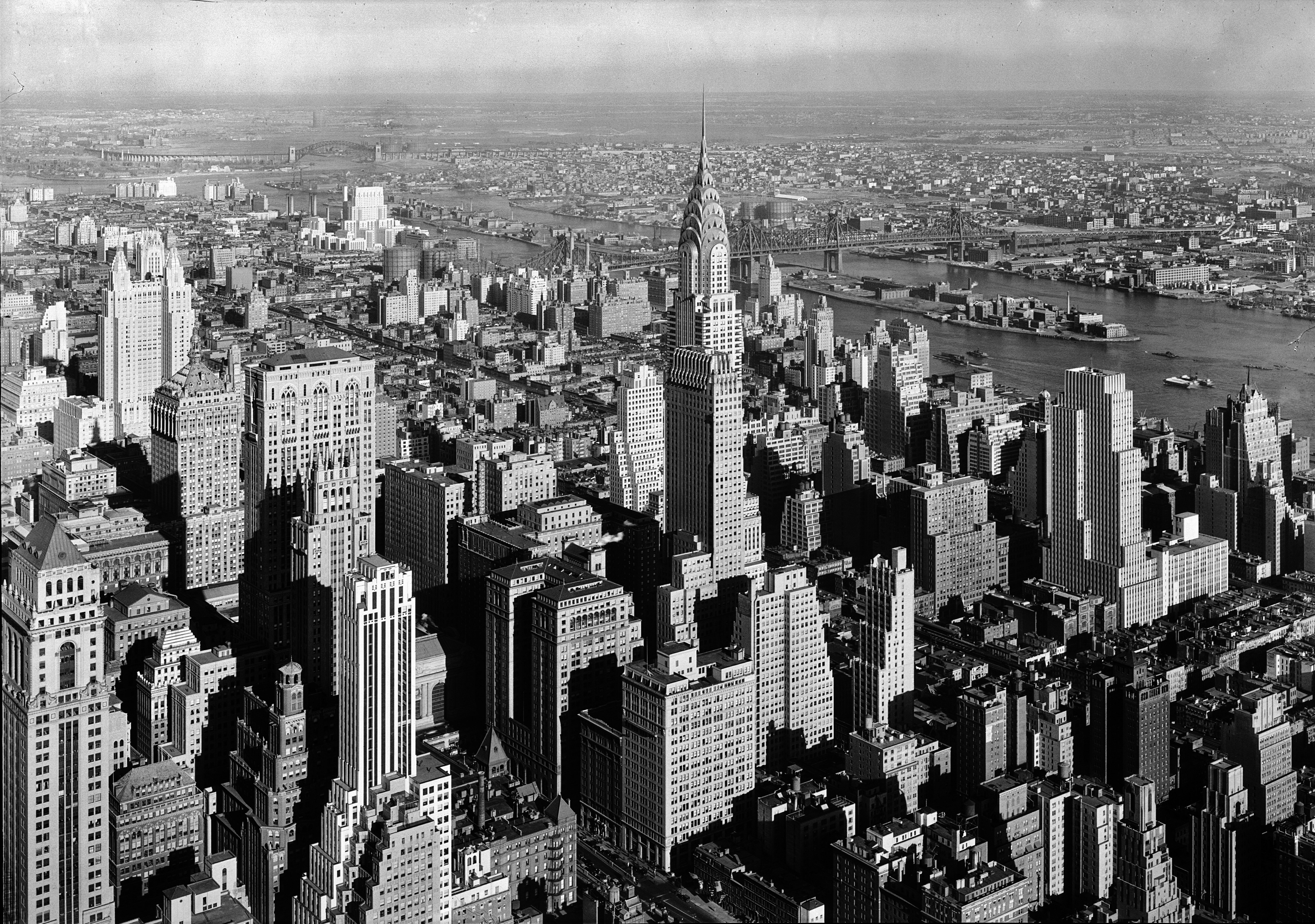}
\caption{View from Empire State Building, Gotscho-Schleisner,  1932}
\label{fig:front_picture}
\end{figure}

\paragraph{Introduction.}

Multiparameter Persistence Modules (MPM) provide a way to capture topological features of multifiltered data. Unlike the one parameter case, the decomposition of an MPM is, in general, neither easy to interpret nor robust to noise \cite{BauerScoccola}. To use MPM in practice, one needs invariants that are informative, stable, computable, and easy either to interpret or to integrate into a statistical or Machine Learning pipeline. A first example is given by the \emph{rank invariant} \cite{CZMP,CZMP-socg} and its vectorisations \cite{corbet2019kernel, vipond2018multiparameter}, which have been used for applications in biology \cite{Vipond21,Kate24}. 
In \cite{HN_discr, fersztand2024harder} Fersztand, Jacquard, Tillmann, and Nanda introduced the \emph{Skyscraper Invariant} (\autoref{fig:front_picture}), a \emph{filtration} of the rank invariant.

 Given a $d$-parameter module $\module$, $\theta \in \bbR$ and parameters $\alpha\leq \beta \in \bbR^d$,
  the value $s^\theta_V(\alpha, \beta)$ of the Skyscraper  invariant is the rank of the structure map $V_\alpha \to V_\beta$ after -- informally -- restricting it to elements that persist over a parameter region whose volume is at most $1/\theta$.
 It is defined by the \emph{Harder-Narasimhan} filtration for a central charge concentrated at $\alpha$. It is strictly stronger than the rank invariant and carries a more discriminating erosion-distance, which makes it stable with respect to the interleaving distance \cite{fersztand2024harder}. The Skyscraper Invariant induces filtrations of all the aforementioned vectorisations of the rank invariant. 

 We will describe multiple algorithms to compute the Skyscraper Invariant and also touch on the computation of Harder-Narasimhan Filtrations with respect to other stability conditions.

\paragraph{Contributions.} 
In \autoref{sec:sky_invariant}, we give an elementary introduction to Harder-Narasimhan (HN) filtrations and a post-hoc motivation for their utility in Persistence Theory. We also briefly explain that, up to an $\varepsilon$-error in the erosion distance, it suffices to compute the HN filtration at each point $\alpha$ on an $\varepsilon$-grid in the parameter space.

For quiver representations, this was recently made possible in polynomial time via Cheng's algorithm \cite{cheng2024deterministic}, described in \autoref{sec:cheng}. We present the first implementation of Cheng's algorithm\footnote{https://github.com/marcf-ox/sky-inv-quiv/} using the random method introduced in \cite{franks2023shrunk} to compute shrunk subspaces.

The obtained algorithm computes with high probability the Skyscraper Invariant in $\mathcal{O}(\thick \module^{19}\varepsilon^{-11d})$ time, up to poly-logarithmic factors, where $\thick{\module}$, the \emph{thickness} of $\module$, is the maximal pointwise dimension. We optimized Cheng's algorithm to our setting, leading in \autoref{prop:cheng_advanced_analysis} to an improvement of a factor $\thick \module^{2}\varepsilon^{-2d}$ compared to the general version of the algorithm (\autoref{cor:cheng_naive_analysis}).  Building on ideas from \cite{iwamasa2024algorithmic}, we propose further optimisations that yield considerable empirical speed-ups.

A priori, this method is generally not directly applicable to the large data considered in TDA applications.
 In \autoref{sec:approx}, we introduce a strategy which significantly cuts the computation time of the Skyscraper Invariant: \autoref{alg:skeleton} reduces the computation of the HN filtration at $\alpha \in \bbR$ to \emph{indecomposable} submodules \emph{uniquely generated} at $\alpha$. It blends the decomposition algorithm \textsc{aida} and kernel computation with \textsc{mpfree}\footnote{https://github.com/JanJend/AIDA, https://bitbucket.org/mkerber/mpfree/}.

Unfortunately, even when restricted to these small submodules, Cheng's algorithm in this form turns out to be too slow, which we prove experimentally in \autoref{tab:cheng_exp}. 
Let $k \coloneqq \max_{W \subset V} \thick{W}$ be the maximal thickness of an indecomposable, uniquely generated submodule. We observed (\autoref{obs:factor_dimension}) that for modules generated by the Persistent Homology of many typical bifiltrations, $k$ is mostly $1$ and rarely larger than $3$. 
 
 This low dimensionality enables the computation of the HN filtration with
 a brute-force method, which we describe and improve in \autoref{sec:hn_exhaustive}. Together with \autoref{alg:skeleton}, this leads to an exhaustive search algorithm to compute the Skyscraper Invariant up to an $\varepsilon$ error. Its runtime is fixed-parameter tractable (FPT) in $k$: $\Oc \left( 1/\varepsilon^d \left(T_{\mathsc{dec}}(\module)+\thick{\module} 2^{k^2 + \Oc(k)} \right) \right) $, where $T_{\mathsc{dec}}$ is the time needed for decomposition.

The bottleneck of this computation is the iteration over all $\alpha$ in an $\varepsilon$-grid, resulting in the $1/\varepsilon^d$-factor in the computation time. To avoid this, we investigate, in \autoref{sec:wall_and_chamber}, how the HN filtration at $\alpha \in \bbR^d$ changes when $\alpha$ varies. The result is our main theoretical contribution. We define an equivalence relation $\sim_{\text{HN}(V)}$ on $\bbR^d$ by requiring the HN filtrations of $V$ at two equivalent degrees to differ only by an update of the degrees of the generators and relations of their factors (\autoref{def:S_equ}).

\begin{theoremA}\label{thm:main}
Let $\module$ be a f.p. persistence module with bounded support and let $\bigcup_{j \in J} C_j$ be the rectangular tiling of the support induced by the Betti-numbers of $\module$. 
The partition induced by $\sim_{\text{HN}(V)}$ on every $C_j$ is given by the minimisation diagram of a finite set of multilinear polynomials of degree $d-1$.
\end{theoremA}

This partition is a slice of the \emph{wall-and-chamber} structure (e.g. \cite{HILLE2002205}). 

For $2$-parameter modules, the $S_i$ are convex polygons and are therefore easy to compute and store. We describe how to compute them in \autoref{alg:full_hnf}, which creates a data structure from which the exact value of the Skyscraper Invariant can be queried in average time $\log(n) +\thick{\module}\log\thick{\module}$, where $n$ is the number of generators and relations of the module. 

Our final \autoref{alg:parallel}, the \emph{parallel grid scan},
blends this exact computation with the approximation to avoid the $(1/\varepsilon^d)$-bottleneck.
We have implemented it and use it to demonstrate the applicability of the Skyscraper Invariant in data analysis by computing filtered Multiparameter Landscapes for biological data from \cite{Vipond21}.

\paragraph{Related Work.} Many informative invariants of MPM have been proposed in the literature, including, among others, \cite{botnan2021signed, blanchette2023exact, kim2021generalized, mccleary2022edit, morozov2021output, ASASHIBA2025107905}, and for some, there are also implemented algorithms \cite{dkm-computing, vipond2018multiparameter, corbet2019kernel, images2020, loiseaux2023stable, loiseaux2025multi, xin2023gril}. All of these are either not stronger than the rank invariant, not (yet) practically computable for large modules, or unstable in interleaving distance, demonstrating the advantage of our work. The idea of computing filtrations has also appeared in \cite{Bjerkevik,millerzhang24}.
An algorithm for the Skyscraper Invariant has been implemented in the special case of ladder persistence modules \cite{JacquardThesis2024}. 
The computation of HN filtrations of quiver representations is an active topic in complexity theory \cite{iwamasa2024algorithmic,franks2023shrunk}.

\section{A Conceptual Introduction to the Skyscraper Invariant}\label{sec:sky_invariant}

\subsection{Persistence Modules}

\paragraph{Filtrations of simplicial complexes.}

Topological Data Analysis seeks to compute topological invariants from geometric data. In a typical setting, one constructs a simplicial complex  filtered by $d\geq 1$ parameters. For example, given a simplicial complex $X\in \simp$ and a map $f \colon X \to \bbR^d$ which is constant on simplices, one constructs the sublevel sets $X_{\leq \alpha} \coloneqq \{ x \in X \mid f(x) \leq \alpha \}$. These assemble to a functor  $X_{\leq \bullet} \colon \bbR^d \to \simp $ via inclusion of sublevel sets. Fix a (in practice often finite) field $\K$. The composition with the homology functor $\bH_*(-;\K)$ turns $X_{\leq \bullet}$ into a functor $\module$ from $\bbR^d $ to $\vect_\K$ - a \emph{Multiparameter} (or  $d$-parameter) persistence module.

Our main examples for filtered simplicial complexes come from \emph{density-scale} bifiltrations. Here $f$ has $d=2$ components: The first indicates the \emph{size} of a simplex and the second its \emph{density} in the dataset. These bifiltrations were introduced in \cite{CZMP} as a way to reveal topological features in noisy data (\autoref{ex:discrete_module}). We review common bifiltrations in \autoref{subsec:bifiltr} and also refer to \cite[Section 5]{botnan2022introduction}.

\begin{example}\label{ex:discrete_module}
A standard 1-parameter pipeline \cite{edelsbrunner2000topological} on a point set $S$ looks as follows: One constructs the \emph{Vietoris-Rips complex}, $\Rips_\varepsilon(S)$, the flag complex of the $\varepsilon$-neighbourhood-graph, which is therefore filtered by $\varepsilon \in \bbR$. Its \emph{Persistent Homology} groups $\bH_*(\Rips (S); \K) \colon \bbR \to \vect_\K$ decompose into a barcode or persistence diagram (\autoref{fig:barcode}).
\begin{figure}[H]
    \centering
    \begin{tikzpicture}
        \node[anchor=center, inner sep=0] (image) at (0,0) 
            {\includegraphics[scale=0.36]{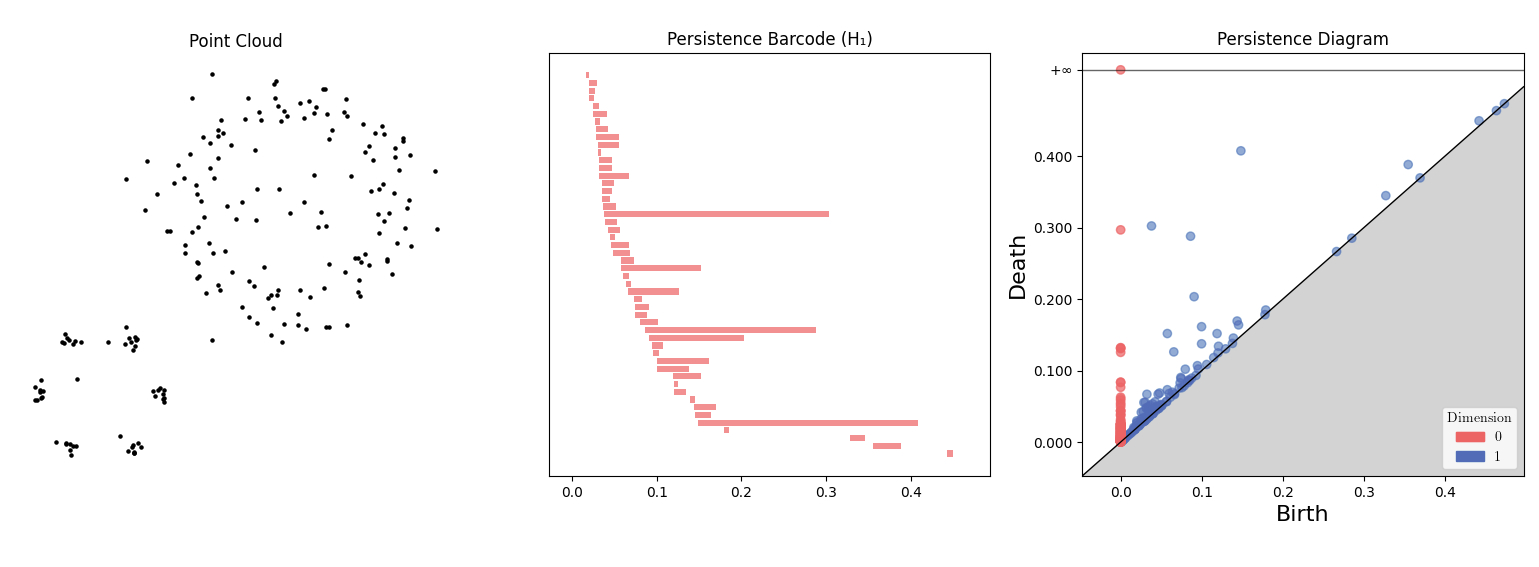}};
        
        \node[blue] at (image.south) [yshift=1.36cm, xshift = 0.3cm] {$\times$};
    \end{tikzpicture}
    \caption{ The lower left circle produces the bar marked with the \textcolor{lightblue}{blue cross}.}
    \label{fig:barcode}
\end{figure}
The large circle on the right is impossible to detect with this technique without knowing the intensity of the noise we need to remove.
Instead, we compute a kernel density estimate for each point. Then for any $\delta \in \bbR$, we filter the simplicial complex again by allowing only points of density $\geq \delta$, constructing the \emph{Density-Rips} bifiltration.
\begin{figure}[H]
    \centering
    \begin{minipage}[c]{0.42\textwidth}
    \centering
    \begin{tikzpicture}
        \node[inner sep=0] (img) {\includegraphics[scale=0.045]{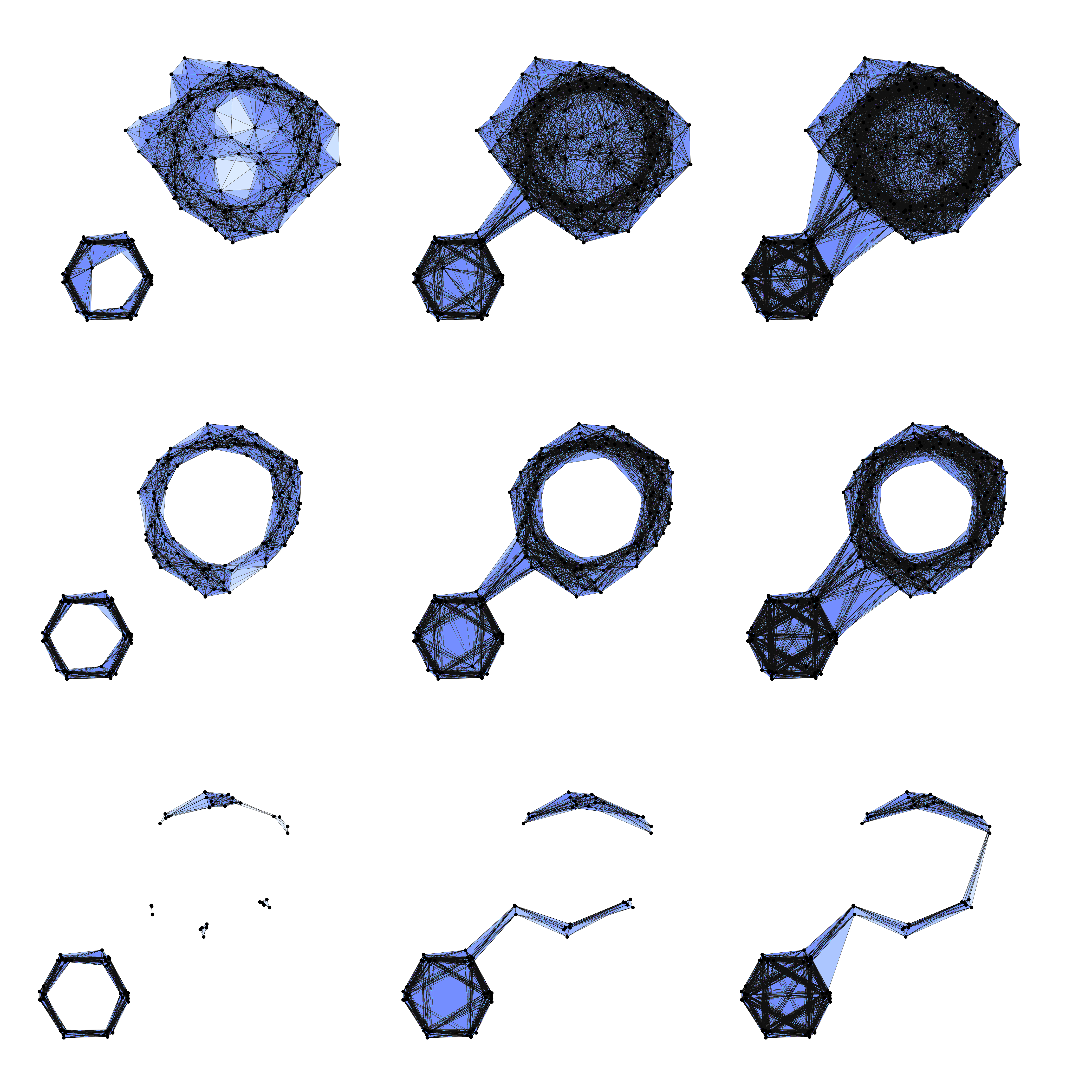}};
        
        % Horizontal arrow below (scale)
        \draw[->, thick] ([yshift=-0.3cm]img.south west) -- 
                         ([yshift=-0.3cm]img.south east) 
                         node[midway, below] {scale};
        
        % Vertical arrow to the left (codensity)
        \draw[->, thick] ([xshift=-0.3cm]img.south west) -- 
                         ([xshift=-0.3cm]img.north west) 
                         node[midway, above, rotate=90] {codensity};
    \end{tikzpicture}
\end{minipage}
    \hfill
\begin{minipage}[c]{0.55\textwidth}
        \centering
        \scalebox{1.25}{$\overset{\bH_1(- ; \bbF_2)}{\Longrightarrow}$}
        \scalebox{1.3}{%
        \begin{tikzcd}[ampersand replacement=\&]
        \bbF_2 \& 0 \& 0 \\
        \bbF_2^2 \& \bbF_2 \& \bbF_2 \\
        \bbF_2 \& 0 \& 0
        \arrow["", from=1-1, to=1-2]
        \arrow["", from=1-2, to=1-3]
        \arrow["\smat{\begin{bmatrix} 0 & 1 \end{bmatrix}}", from=2-1, to=2-2]
        \arrow["", from=2-2, to=2-3]
        \arrow["", from=3-1, to=3-2]
        \arrow["", from=3-2, to=3-3]
        \arrow["\smat{\begin{bmatrix} 1 \\ 0 \end{bmatrix}}", from=3-1, to=2-1]
        \arrow["", from=3-2, to=2-2]
        \arrow["\smat{\begin{bmatrix} 1 & 0 \end{bmatrix}}", from=2-1, to=1-1]
        \arrow["", from=2-2, to=1-2]
        \arrow["", from=3-3, to=2-3]
        \arrow["", from=2-3, to=1-3]
        \end{tikzcd}} 
\end{minipage}
    \caption{A $3 \times 3$ Density-Rips bifiltration of \autoref{fig:barcode}
     and its first persistent homology group over $\bbF_2$}
    \label{fig:density_rips}
\end{figure}
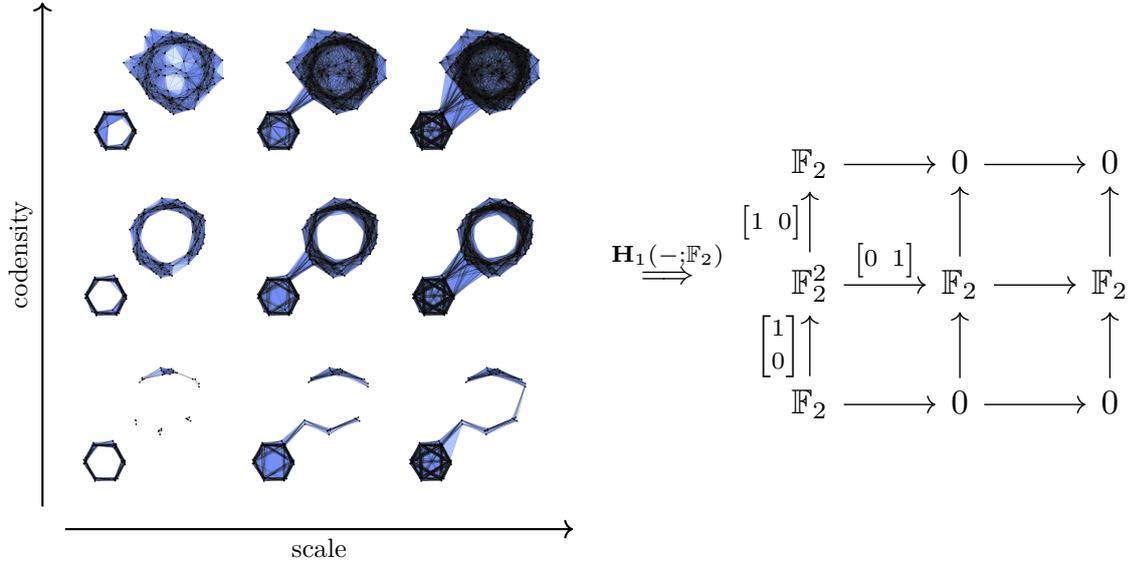
By applying $\bH_1( - ;\K)$, we can now detect the large circle in the middle row in \autoref{fig:barcode}. 
\end{example}

The \emph{support} of $\module$ is $\supp(\module) \coloneqq \set{\alpha\in \bbR^d \mid V_\alpha\neq 0}$.

\paragraph{Intervals.}
The analog of a \emph{bar} in the multiparameter case is the \emph{Interval Module}.

\begin{definition}[Interval Module]\label{def:interval-module}
A sub-poset $I \subset \bbR^d$ is an \emph{interval} if, with respect to the order, it is convex and connected. The module $\K[I]\colon \bbR^d \to \vect_\K$ is defined as 
\[ \K[I]_\alpha = \begin{cases} \K &\text{ if } \alpha \in I \\
0 &\text{ otherwise} \end{cases} 
\quad \quad 
\K[I]_{\alpha \to \beta} = \begin{cases} \Id_\K &\text{ if } \alpha, \beta \in I \\
0 &\text{ otherwise} \end{cases}.
\]
An \emph{interval module} is a persistence module that is isomorphic to $\K[I]$ for some interval $I$.
\end{definition}

When $d=1$, every pointwise finite-dimensional 1-parameter persistence module decomposes as a direct sum of interval modules. The \emph{barcode} of a functor $V\colon \bbR\to \vect_{\K}$ is obtained by depicting the supports of the indecomposable summands of $V$. Unfortunately, $\bbR^d$-persistence modules do not decompose into intervals if $d>1$. Instead, one computes invariants that lose some information.

\subsection{Harder-Narasimhan Filtrations}\label{subsec:skyscraper-informal-def}

The Skyscraper Invariant arises from the \emph{Harder-Narasimhan filtration} \cite{Harder1974} of persistence modules for a specific set of stability conditions. We will explain how one \emph{could} arrive at considering HN filtrations even without prior knowledge of geometric invariant theory \cite{Mumford94}.
 For persistence modules, this theory was originally developed in \cite{HN_discr} and \cite{fersztand2024harder}.

\paragraph{Lifetime.}
Let $\module$ be a \emph{multiparameter} persistence module.
A sensible idea to extend the barcode is the following: search for an element with a "maximal lifetime" (an idea developed further in \cite{millerzhang24}). By lifetime we mean for any element $v \in V_\alpha$, with $\alpha \in \bbR^d$, the set
\[ L_{v} \coloneqq \{ \beta \in \bbR^d \mid \alpha \leq \beta \text{ and } V_{\alpha \to \beta} \left( v \right) \neq 0 \}
\]
of parameter values under which $v$ persists and "maximal" refers here to its volume.
To assure that every lifetime has a finite volume, we will assume that every module is bounded, which can always be realised by restricting the module to a bounded subset of $\bbR^d$.

\begin{example} Consider again the bifiltered simplicial complex $\Rips(X)\colon \bbR^2\to \simp$ obtained from the point set $X$ in \autoref{fig:barcode}.
We depict this point set, coloured by a kernel density estimate, and the \emph{Hilbert function} $\udim$ -- the pointwise dimension -- of the 2-parameter module produced by $\bH_1(\Rips(X), \bbF_2)$:

\begin{figure}[H]
\begin{minipage}[t]{0.38\textwidth}\vspace{20pt}
    \centering
    
    \includegraphics[width=1.12\textwidth]{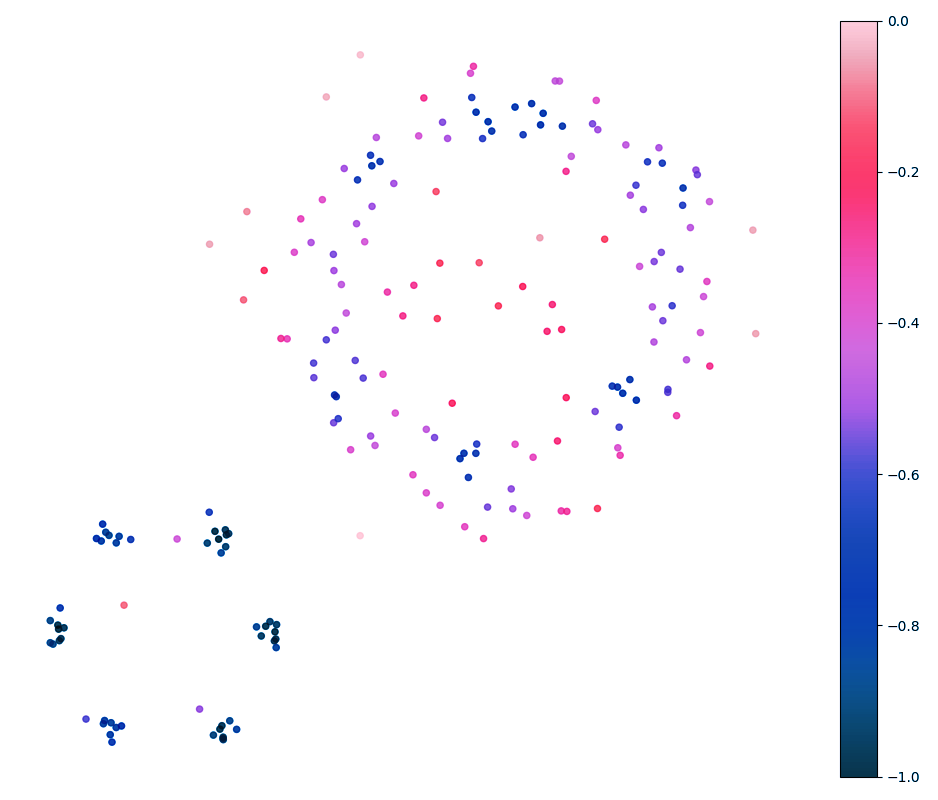}
\end{minipage}
\begin{minipage}[t]{0.57\textwidth}\vspace{0pt}
    \centering
    \includegraphics[width=1.15\textwidth]{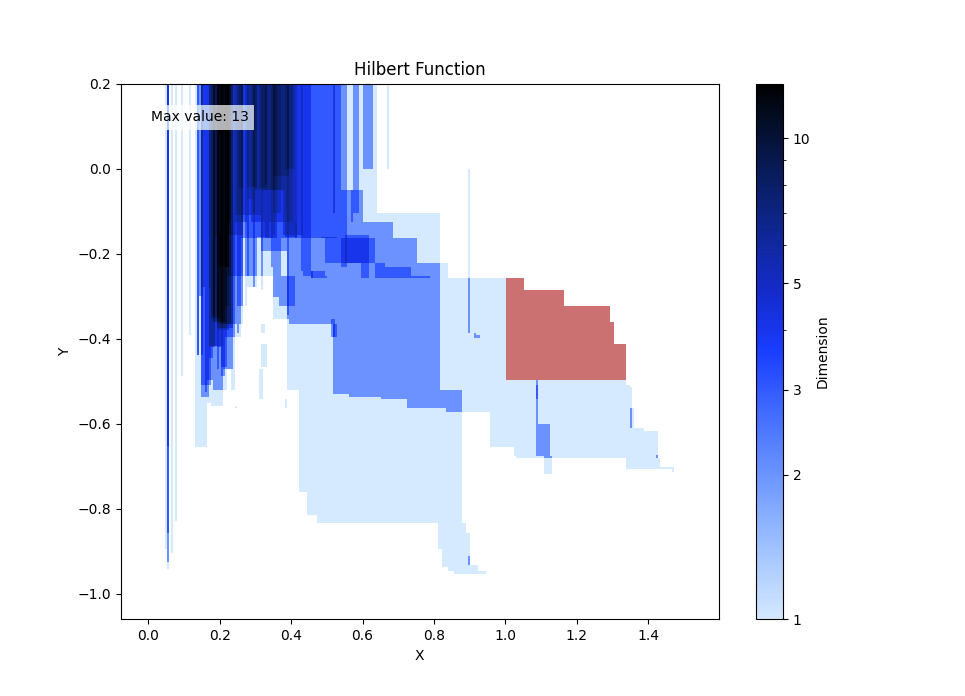}
\end{minipage}
\caption{$X$ filtered by a KDE, less dense points are redder (left) and $\udim \bH_1(\Rips(X),\bbF_2)$ (right).}
\label{fig:HF}
\end{figure}
The smaller circle produces a vertical, almost rectangular shape which lies roughly in the parameter region $[0.4,0.9] \times[-0.8, -0.1]$. The larger circle produces the slightly tilted, staircase-like shape which extends to the right. Then, the $1$-dimensional homology group at $\alpha = (1.0,-0.5)$ is generated by this large circle in $\Rips(X)_{1.0}$ after filtering out points of density lower than $0.5$. Any vector in this space has the same lifetime and we mark it \textcolor{lightred}{red} in \autoref{fig:HF}.
\end{example}

After finding the vector $v$ of maximal lifetime, the submodule $\langle v \rangle$ of $\module$ generated by $v$ is an interval module and $L_v = \supp \langle v \rangle$. Hence if we compute the quotient $\module/\langle v \rangle$, and repeat this process, we obtain a filtration of $\module$ in terms of interval modules.

\paragraph{Additivity.} The element of maximal lifetime is not compatible with direct sums of modules.

\begin{example}
We return to the discretised version of our module from \autoref{fig:density_rips}. It can be split into two indecomposable summands, each roughly corresponding to one circle in the point set. The element of maximal lifetime is the vector $(1,1)^t \in \bbF_2^2$.
\begin{figure}[H]
\centering
\scalebox{1.0}{%
        \begin{tikzcd}[ampersand replacement=\&]
        |[fill=red!20]| \bbF_2  \& 0 \& 0 \\
        |[fill=red!20]|  \bbF_2^2 \& |[fill=red!20]|  \bbF_2 \& |[fill=red!20]|  \bbF_2 \\
        \bbF_2 \& 0 \& 0
        \arrow["", from=1-1, to=1-2]
        \arrow["", from=1-2, to=1-3]
        \arrow["\smat{\begin{bmatrix} 0 & 1 \end{bmatrix}}", from=2-1, to=2-2]
        \arrow["", from=2-2, to=2-3]
        \arrow["", from=3-1, to=3-2]
        \arrow["", from=3-2, to=3-3]
        \arrow["\smat{\begin{bmatrix} 1 \\ 0 \end{bmatrix}}", from=3-1, to=2-1]
        \arrow["", from=3-2, to=2-2]
        \arrow["\smat{\begin{bmatrix} 1 & 0 \end{bmatrix}}", from=2-1, to=1-1]
        \arrow["", from=2-2, to=1-2]
        \arrow["", from=3-3, to=2-3]
        \arrow["", from=2-3, to=1-3]
        \end{tikzcd}}
        $\simeq$ 
        \scalebox{1.2}{%
\begin{tikzcd}[ampersand replacement=\&]
	 \bbF_2 \& 0 \& 0 \\
	 \bbF_2 \& 0 \& 0 \\
    \bbF_2 \& 0 \& 0
	\arrow["", from=1-1, to=1-2]
    \arrow["", from=1-2, to=1-3]
    \arrow["", from=2-1, to=2-2]
    \arrow["", from=2-2, to=2-3]
    \arrow["", from=3-1, to=3-2]
    \arrow["", from=3-2, to=3-3]
    \arrow["", from=3-1, to=2-1]
    \arrow["", from=3-2, to=2-2]
    \arrow["", from=2-1, to=1-1]
    \arrow["", from=2-2, to=1-2]
    \arrow["", from=3-3, to=2-3]
    \arrow["", from=2-3, to=1-3]
\end{tikzcd}}
$\bigoplus$
\scalebox{1.2}{%
\begin{tikzcd}[ampersand replacement=\&]
	 0 \& 0 \& 0 \\
	 \bbF_2 \& \bbF_2 \& \bbF_2 \\
    0 \& 0 \& 0
	\arrow["", from=1-1, to=1-2]
    \arrow["", from=1-2, to=1-3]
    \arrow["", from=2-1, to=2-2]
    \arrow["", from=2-2, to=2-3]
    \arrow["", from=3-1, to=3-2]
    \arrow["", from=3-2, to=3-3]
    \arrow["", from=3-1, to=2-1]
    \arrow["", from=3-2, to=2-2]
    \arrow["", from=2-1, to=1-1]
    \arrow["", from=2-2, to=1-2]
    \arrow["", from=3-3, to=2-3]
    \arrow["", from=2-3, to=1-3]
\end{tikzcd}}
\caption{The maximal lifetime in the left module is indicated by red boxes.}
\label{fig:lifetime_additivity}
\end{figure}

\end{example}

 If we instead consider the \emph{minimal} lifetime of an element, the construction becomes additive, meaning the desired vector lies in a direct summand: 
\begin{equation}\label{eq:additivity}
\text{If } v \in \module, w \in W \text{, then for } v+w \in  \module \oplus W\colon \quad   L_{v + w} = L_{v} \cup L_{w}.
\end{equation}
This enables the use of decomposition algorithms (e.g. \cite{djk25}) to speed up the search. Now, a priori, this element is not well defined, since in every non-zero persistence module we can find a sequence of elements whose lifetime approaches a set of empty volume. To fix this, we choose a concrete parameter value $\alpha \in \bbR^d$ and demand that the element be in $V_\alpha$.

\begin{definition} Let $S \subset \module$ be a sub\emph{set} of $\module$. We denote by $\langle S \rangle \subset \module$ its \emph{induced submodule}. 
\end{definition}

Since we are only considering elements at $\alpha$, we can now restrict our computations to $\langle V_\alpha \rangle$ - and this module is simpler than $V$:

 \begin{definition}
$\module$ is called \emph{uniquely generated} if it has a set of generators of the same degree.
\end{definition}

   In particular, if $V \simeq \oplus V_i$ is an indecomposable decomposition, then $\langle V_\alpha \rangle$ will in a practical setting have a much finer decomposition than $\oplus \langle (V_i)_\alpha \rangle$.
   Still, the choice of an element is not unique: there could be multiple vectors whose lifetimes have the same (minimal) volume. This is the first hint that one should not only consider single elements in $V_\alpha$, but instead consider \emph{submodules} of arbitrary dimension generated at a parameter value $\alpha$.

\paragraph{Stability.}

Let $U \subset V_\alpha$ be any sub vector space. The integral of the Hilbert function $\udim \langle U \rangle$ over $\bbR^d$ would be the natural generalisation of the volume of the lifetime. However, using this definition would never make a sub vector space of dimension higher than $1$ be the candidate for minimal lifetime. We need to average over the dimension of $U$.

\begin{definition}
Let $U \subset V_\alpha$. The \emph{slope} of $\langle U \rangle$ (at $\alpha$) is defined as
\[ \mu(\langle U \rangle) \coloneqq 
\frac{\dim_\K U}{\int_{\bbR^d} \udim_\K \langle U \rangle}
=  \frac{\dim_\K \langle U \rangle_\alpha}{\int_{\bbR^d} \udim_\K \langle U \rangle} 
. \]
By replacing, in the second fraction, $\langle U \rangle$ with \emph{any} persistence module, we get a general definition of the slope at $\alpha$. A module is called \emph{semistable} \cite{Mumford94, King94, bridgeland2007stability} if no submodule has a higher slope. 
 \end{definition}
 
The maximal submodule of highest slope $\langle U_1 \rangle \subset \langle V_\alpha \rangle$ is semistable. Analogously, one defines $\langle U_2 \rangle \subset \langle \module_\alpha \rangle / \langle U_1 \rangle$ and iterating this procedure produces a sequence of semistable modules $(\langle U_i \rangle)_{i \in [l]}$. By defining inductively $F^0 = 0, \ F^{i} \coloneqq U_i + F^{i-1}$ we obtain the \emph{Harder-Narasimhan filtration} \cite{Harder1974}
 \begin{align}\label{eq:quotient} 0= \langle F^0 \rangle \subsetneq \langle F^1 \rangle\subsetneq \dots\subsetneq \langle F^\ell \rangle =\langle V_\alpha \rangle 
 \end{align}
 of $\langle V_\alpha \rangle$ following \cite[Theorem (A)]{fersztand2024harder}. Then, the semistable modules $\langle U_i \rangle$ are isomorphic to the factors $\langle F^{i} \rangle / \langle F^{i-1} \rangle$. The sequence of slopes $\mu(\langle U^{i} \rangle )$ decreases with $i$, which can be seen with \autoref{lmm:see-saw}. 
 We can compute the slopes of the filtration inductively:

\begin{align} \mu( \langle F^{i} \rangle ) =  \mu(\langle  U_i + F^{i-1} \rangle ) = \frac{\dim_\K U_i + \dim_\K F^{i-1}}
{ \int_{\bbR^d} \udim \langle U_i \rangle + \int_{\bbR^d} \udim \langle F^{i-1} \rangle
}. \label{eq:slope_recomputation}
\end{align}
It follows that also the sequence of slopes $\mu(F^{i})$ decreases. We will introduce this again formally in \autoref{subsec:HN-formal}.

\begin{example}\label{ex:stable_module}
In \autoref{fig:ex_two_gen}, we construct a module $\module$ that is not an interval, but is semistable at $(0,0)$.

\begin{figure}[H]
    \centering
\resizebox{1\textwidth}{!}{%repeat commands from main
\definecolor{lightred}{rgb}{1, 0.5, 0.5}
\definecolor{lightblue}{rgb}{0.4, 0.5, 1}
\definecolor{darkblue}{rgb}{0.3, 0.3, 0.8}
\newcommand{\bbracks}[1]{\textcolor{lightblue}{\{} #1 \textcolor{lightblue}{\}} }
\newcommand{\gbracks}[1]{\textcolor{lightred}{(} #1 \textcolor{lightred}{)} }
\newcommand{\lrcell}{\cellcolor{lightred}}
\newcommand{\pblcell}{\cellcolor{lightblue}}

\centering
% First row: Left and Right panels
\begin{minipage}{0.45\textwidth}
    \centering
\begin{tikzcd}[row sep=1em, column sep=1.25em, scale=1, transform shape]
    \ &\\
 \gbracks   0 \ar[r,scale=1.5,"\dots" description,draw=none] \ar[u,"\vdots" description,draw=none]&\  &  & \\
   \gbracks \K \arrow[u] \arrow[r] & 0 \ar[r,scale=1.5,"\dots" description,draw=none] \ar[u,"\vdots" description,draw=none] & \ & \\
    \K^2 \arrow[u,"{[0\ 1]}"] \arrow[r,"{[1\ 1]}"] & \gbracks\K \arrow[u] \arrow[r] & 0 \ar[r,scale=1.5,"\dots" description,draw=none] \ar[u,"\vdots" description,draw=none]& \ \\
   \bbracks{\bbracks{\K^2}} \arrow[u,equals] \arrow[r,equals] & \K^2 \arrow[r,"{[1\ 0]}"] \arrow[u,"{[1\ 1]}"] & \gbracks\K \arrow[r] \arrow[u] & \gbracks 0\ar[u,"\vdots" description,draw=none]\ar[r,scale=1.5,"\dots" description,draw=none] &\ 
\end{tikzcd}

\begin{tikzpicture}[scale=0.9, overlay]
\node at (1.5,3.5) {\Large $V_{|\bbN}$};
\end{tikzpicture}

\end{minipage}

\begin{minipage}{0.45\textwidth}
    \centering
\begin{tikzpicture}[scale=0.9]
    % Axes
    \draw[->, thick] (0,0) -- (4,0) node[right] {$x$};
    \draw[->, thick] (0,0) -- (0,4) node[above] {$y$};

    % Axis labels
    \foreach \x in {0,1,2,3} {
        \node at (\x, -0.2) [below] {\x};
    }
    \foreach \y in {0,1,2,3} {
        \node at (-0.2, \y) [left] {\y};
    }

    % Darker blue region: [0,2)^2 minus [1,2)^2
    \fill[darkblue, opacity=0.8] (0,0) rectangle (2,2);
    \fill[white] (1,1) rectangle (2,2);

    % Lighter blue regions:
    \fill[lightblue, opacity=0.8] (0,2) rectangle (1,3);
    \fill[lightblue, opacity=0.8] (1,1) rectangle (2,2);
    \fill[lightblue, opacity=0.8] (2,0) rectangle (3,1);

    % Generator at (0,0)
    \fill[lightblue, draw=black] (0,0) circle (5pt);
    \draw[black] (0,0) circle (3pt);

    % Relations
    \foreach \x/\y in {2/0, 0/2, 3/0, 1/1, 0/3} {
        \node[rectangle, draw=black, fill=lightred, inner sep=2pt] at (\x,\y) {};
    }
\end{tikzpicture}
\begin{tikzpicture}[scale=0.9, overlay]
\node at (-1,3.8) {\Large $\udim V$};
\end{tikzpicture}
\end{minipage}}
\end{figure}
\begin{figure}[H]
 \centering
\resizebox{1\textwidth}{!}{%repeat commands from main
\definecolor{lightred}{rgb}{1, 0.5, 0.5}
\definecolor{lightblue}{rgb}{0.4, 0.5, 1}
\definecolor{darkblue}{rgb}{0.3, 0.3, 0.8}
\newcommand{\bbracks}[1]{\textcolor{lightblue}{\{} #1 \textcolor{lightblue}{\}} }
\newcommand{\gbracks}[1]{\textcolor{lightred}{(} #1 \textcolor{lightred}{)} }
\newcommand{\lrcell}{\cellcolor{lightred}}
\newcommand{\pblcell}{\cellcolor{lightblue}}

\begin{minipage}{2.25\textwidth}
\centering
\begin{tikzpicture}[scale=5,transform shape]
    % Axes
    \draw[->, thick] (0,0) -- (4,0) node[right] {$x$};
    \draw[->, thick] (0,0) -- (0,4) node[above] {$y$};
    % Axis labels
    \foreach \x in {0,1,2,3} {
        \node at (\x, -0.15) [below] {\x};
    }
    \foreach \y in {0,1,2,3} {
        \node at (-0.15, \y) [left] {\y};
    }
    % Lighter blue regions:
    \fill[lightblue, opacity=0.8] (0,0) -- (3,0) -- (3,1) -- (2,1)--(2,2)--(0,2)-- cycle;
        \node at (2,4){\Large $ \udim \left\langle \begin{pmatrix}
  1 \\
  0
  \end{pmatrix}
\right\rangle$
};
\end{tikzpicture}
\end{minipage}
\hfill
\begin{minipage}{2.25\textwidth}
\centering
\begin{tikzpicture}[scale=5,transform shape]
    % Axes
    \draw[->, thick] (0,0) -- (4,0) node[right] {$x$};
    \draw[->, thick] (0,0) -- (0,4) node[above] {$y$};
    % Axis labels
    \foreach \x in {0,1,2,3} {
        \node at (\x, -0.15) [below] {\x};
    }
    \foreach \y in {0,1,2,3} {
        \node at (-0.15, \y) [left] {\y};
    }
    % Lighter blue regions:
    \fill[lightblue, opacity=0.8] (0,0) -- (0,3) -- (1,3) -- (1,2) -- (2,2) -- (2,0) -- cycle;
        \node at (2,4){\Large  $\udim \left\langle   \begin{pmatrix}
   0\\
    1
  \end{pmatrix}
 \right\rangle$
};
\end{tikzpicture}
\end{minipage}
\hfill
\begin{minipage}{2.25\textwidth}
\centering
\begin{tikzpicture}[scale=5,transform shape]
    % Axes
    \draw[->, thick] (0,0) -- (4,0) node[right] {$x$};
    \draw[->, thick] (0,0) -- (0,4) node[above] {$y$};
    % Axis labels
    \foreach \x in {0,1,2,3} {
        \node at (\x, -0.15) [below] {\x};
    }
    \foreach \y in {0,1,2,3} {
        \node at (-0.15, \y) [left] {\y};
    }
    % Lighter blue regions:
    \fill[lightblue, opacity=0.8] (0,0) -- (3,0)--(3,1)--(1,1)-- (1,3)--(0,3)--cycle;
    \node at (2,4){\Large $ \udim \left\langle \begin{pmatrix}
  1 \\
    -1
  \end{pmatrix}
 \right\rangle$
};
\end{tikzpicture}
\end{minipage}}
    \caption{A module, its dimension over $\bbR^2$, and candidates for submodules of high slope at $(0,0)$.}
    \label{fig:ex_two_gen}
\end{figure}
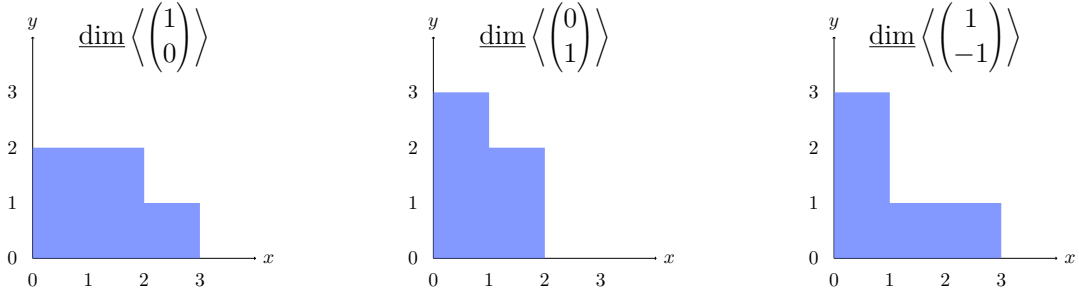

At $\alpha\coloneqq (0,0)$, each of the three shown submodules has slope $1/5$. Every other dimension $1$ subspace has a lifetime equal to the whole support of the module, and therefore slope $1/6$. The whole module on the other hand has slope $2/9$, so it is semistable.
\end{example}

Because of the existence of modules like \autoref{ex:stable_module}, the factors of the Harder-Narasimhan filtration at $\alpha$ should generally not all be intervals. Surprisingly, they did turn out to be intervals for all of our examples of density-scale bifiltrations  of point sets in $\bbR^2$.

\begin{conjecture}\label{obs:factor_dimension}
Modules coming from the first persistent homology group of a density-scale bifiltration (\autoref{subsec:bifiltr}) of a generic pointset in $\bbR^2$ do not contain indecomposable uniquely generated submodules of dimension larger than 1. In particular, all semistable factors of the Harder-Narasimhan filtration are interval modules. 
\end{conjecture}

\paragraph{The Skyscraper Invariant.}
Let $0 \into F^1 \into \dots \into F^{\ell} = V_\alpha$ generate the HN filtration at $\alpha$.

\begin{definition}{\cite{fersztand2024harder}}
For $\alpha \leq \beta \in \bbR^d$ and $\theta \in \bbR$, the Skyscraper Invariant is
\[ s^\theta_V(\alpha, \beta) \coloneqq \dim_\K \langle F^{i} \rangle_\beta = \sum\limits^{i}_{j = 1} \dim_\K \left\langle U_j \right\rangle_\beta \quad 
\text{ where $i$ is maximal with $\mu(U^{i}) \geq \theta$.} \]
\end{definition}

This value only depends on the Hilbert function of the factors ${ \udim \langle U_j\rangle}$ and their order. Since
$s^0_V(\alpha, \beta) = \dim_\K \langle V_\alpha \rangle_\beta = \rank (V_{\alpha \to \beta})$, the Skyscraper filters the rank invariant.

\begin{example}\label{ex:s_rank_representable} The modules $\module$ and $W$ on the left cannot be distinguished by the rank invariant, whereas their Skyscraper Invariants $s_{-}^\theta(\bullet, \bullet)$ are different for any $\theta \in (\frac13,\infty)$ 
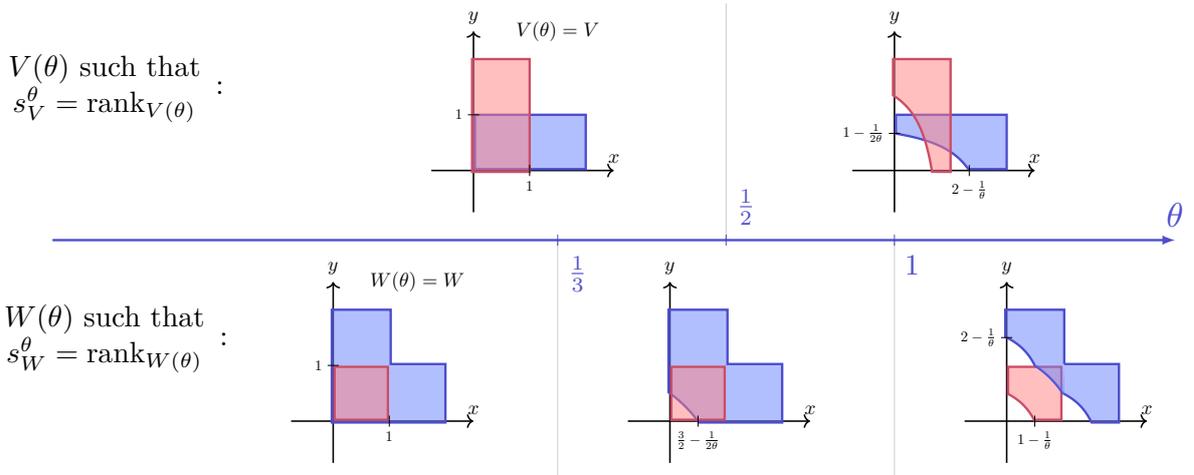
\begin{figure}[H]
        \centering
      \resizebox{1.01\textwidth}{!}{  \begin{tikzpicture}[scale =1]

%theta
\draw[->,>=latex,darkblue,very thick] (-7,-2.25) --  (13,-2.25) node[scale=1.75,above]{$\theta$} ;

\draw[darkblue] (8,-2.35) node[scale=1.5,below right]{1} -- (8,-2.15) ;
\draw[darkblue] (5,-2.35)  -- (5,.-2.15) node[scale=1.5,above right]{$\frac12$};
\draw[darkblue] (2,-2.35 ) node[scale=1.5,below right]{$\frac13$} -- (2,-2.15) ;

\draw[opacity=0.2] (8,-2.25) -- (8,-6.5);
\draw[opacity=0.2] (2,-2.25) -- (2,-6.5);
\draw[opacity=0.2] (5,2) -- (5,-2.25);
% V
\node[scale=1.5] (sV) at (-6,.5) {$\begin{array}{c}
     V(\theta)\text{ such that}\\ s^\theta_{V}=\rank_{V(\theta)}
\end{array}$:};
\node at (2,1.5) {$V(\theta)=V$};

\node[scale=1.5] (sW) at (-6,-4) {$\begin{array}{c}
     W(\theta)\text{ such that}\\ s^\theta_{W}=\rank_{W(\theta)}
\end{array}$:};

\node at (-.5,-3) {$W(\theta)=W$};

% origins

\node (000V) at (0.5,-1){};
\def\xOneV{8}
\def\yOneV{-1}
\node (001V) at (\xOneV,\yOneV){};
\node (000W) at (-2,-5.5){};
\def\xOneW{4}
\def\yOneW{-5.5}
\node (001W) at (\xOneW,\yOneW){};
\def\xTwoW{10}
\def\yTwoW{-5.5}
\node (002W) at (\xTwoW,\yTwoW){};

%SUPPORTS
\def\shift{0.025cm}
\def\bigshift{0.36cm}
\def\verybigshift{0.5cm}

%draw axes
\foreach \x in {000W,001W,000V,002W,001V}{
\coordinate (axisOrigin) at (\x.center);
% x-axis:
\draw[thick,->]   ([xshift=-0.75cm]axisOrigin) -- 
  ([xshift=2.5cm]axisOrigin) node[above] {$x$};
  %y axis
\draw[thick,->]   ([yshift=-0.75cm]axisOrigin) -- 
  ([yshift=2.5cm]axisOrigin) node[above] {$y$};
}

\def\ptconv{28.45274}
%%supportV0
\draw[darkblue, very thick, draw opacity=1,fill= lightblue, fill opacity=0.5] ([xshift=2cm,yshift=1cm]000V.center) rectangle  ([xshift=\shift,yshift=\shift]000V.center) ;
\draw[darkred, very thick, draw opacity=1, fill= lightred,  fill opacity=0.5] ([xshift=1cm,yshift=2cm]000V.center) rectangle   ([xshift=-\shift,yshift=-\shift]000V.center) ;
\draw  ([yshift=1cm,xshift=-0.1cm]000V.center) node[scale=0.75,left]{$1$}-- ([yshift=1cm,xshift=0.1cm]000V.center);
\draw  ([xshift=1cm,yshift=-0.1cm]000V.center) node[scale=0.75,below]{$1$}-- ([xshift=1cm,yshift=0.1cm]000V.center);
%%support V1
\def\thetaVal{1.5}
%x 001V
\draw[darkblue, very thick, draw opacity=1,fill= lightblue, fill opacity=0.5,domain=\xOneV+\shift/\ptconv:\xOneV+ 2-1/(\thetaVal*(1-\shift/\ptconv)), variable=\x] ([xshift=2cm,yshift=1cm]001V.center)   --([xshift=\shift,yshift=1cm]001V.center) --plot ({\x}, {\yOneV+1-1/((2+\xOneV-\x)*\thetaVal))})--  ([xshift=2cm,yshift=\shift]001V.center)    --cycle;
\draw  ([xshift=2cm-\ptconv/\thetaVal,yshift=-0.1cm]001V.center) node[scale=0.75,below]{$2-\frac1\theta$}-- ([xshift=2cm-\ptconv/\thetaVal,yshift=0.1cm]001V.center);
\draw  ([yshift=1cm-\ptconv/(2*\thetaVal),xshift=-0.1cm]001V.center) node[scale=0.75,left]{$1-\frac1{2\theta}$}-- ([yshift=1cm-\ptconv/(2*\thetaVal),xshift=0.1cm]001V.center);
%y 001V
\draw[darkred, very thick, draw opacity=1,fill= lightred, fill opacity=0.5,domain=\xOneV-\shift/\ptconv:\xOneV+ 1-1/(\thetaVal*(2+\shift/\ptconv)), variable=\x]   ([yshift=2cm,xshift=1cm]001V.center)   -- ([yshift=2cm,xshift=-\shift]001V.center) -- plot ({\x}, {\yOneV+2-1/((1+\xOneV-\x)*\thetaVal))})  --([yshift=-\shift,xshift=1cm]001V.center) -- cycle;

%W0
\draw[darkblue, very thick, draw opacity=1,fill= lightblue, fill opacity=0.5] ([xshift=2cm,yshift=1cm+\shift]000W.center) --  ([xshift=2cm,yshift=-\shift]000W.center) --([xshift=-\shift,yshift=-\shift]000W.center) -- ([yshift=2cm,xshift=-\shift]000W.center) -- ([yshift=2cm,xshift=1cm+\shift]000W.center) -- ([xshift=1cm+\shift,yshift=1cm+\shift]000W.center) -- cycle ;
\draw[darkred, very thick, draw opacity=1, fill=lightred,  fill opacity=0.5] ([xshift=1cm-\shift,yshift=1cm-\shift]000W.center) rectangle   ([xshift=\shift,yshift=\shift]000W.center) ;
\draw  ([xshift=1cm,yshift=-0.1cm]000W.center) node[scale=0.75,below]{$1$}-- ([xshift=1cm,yshift=0.1cm]000W.center);
\draw  ([yshift=1cm,xshift=-0.1cm]000W.center) node[scale=0.75,left]{$1$}-- ([yshift=1cm,xshift=0.1cm]000W.center);
%W1
\def\thetaVal{0.5}
\draw[darkblue, very thick, draw opacity=1,fill= lightblue, fill opacity=0.5, domain  =\xOneW-\shift/\ptconv:\xOneW + 2- (\thetaVal+1)/(\thetaVal* (2+\shift/\ptconv)),variable=\x] 
  ([xshift=2cm,yshift=-\shift]001W.center) --([xshift=2cm,yshift=1cm+\shift]001W.center) -- ([xshift=1cm+\shift,yshift=1cm+\shift]001W.center) -- ([yshift=2cm,xshift=1cm+\shift]001W.center)  -- ([yshift=2cm,xshift=-\shift]001W.center)  -- 
  plot ( {\x}, {\yOneW+2- (\thetaVal+1)/(\thetaVal* (2-\x+\xOneW))})-- 
  cycle ;
\draw[darkred, very thick, draw opacity=1, fill=lightred,  fill opacity=0.5] ([xshift=1cm-\shift,yshift=1cm-\shift]001W.center) rectangle   ([xshift=\shift,yshift=\shift]001W.center) ;
\draw  ([xshift=1.5cm-\ptconv/(2*\thetaVal),yshift=-0.1cm]001W.center) node[scale=0.75,below]{$ \frac 32-\frac1{2\theta}$}-- ([xshift=1.5cm-\ptconv/(2*\thetaVal),yshift=0.1cm]001W.center);
%W2
\def\thetaVal{2}
\draw[darkred, very thick, draw opacity=1, fill=lightred,  fill opacity=0.5,domain =\xTwoW+\shift/\ptconv:\xTwoW + 1-1/(\thetaVal*(1-\shift/\ptconv))
] ([xshift=1cm-\shift,yshift=1cm-\shift]002W.center) -- ([xshift=\shift,yshift=1cm-\shift]002W.center) -- plot ({\x}, {\yTwoW+1-1/((1+\xTwoW-\x)*\thetaVal))})--  ([xshift=1cm-\shift,yshift=\shift]002W.center) --cycle;
\def\xMid{\xTwoW + 1-1/(\thetaVal*(1-\shift/\ptconv))}
\draw[darkblue, very thick, draw opacity=1, fill=lightblue,  fill opacity=0.5,domain =\xTwoW-\shift/\ptconv:\xMid+0.01
]  ([xshift=1cm+\shift,yshift=1cm+\shift]002W.center) -- ([xshift=1cm+\shift,yshift=2cm]002W.center) -- ([xshift=-\shift,yshift=2cm]002W.center) -- plot ({\x}, {\yTwoW+2-1/((1+\xTwoW-\x)*\thetaVal))})  ;
\draw[darkblue, very thick, draw opacity=1, fill=lightblue,  fill opacity=0.5,domain =1+\xTwoW:1+\xMid+\shift/\ptconv
] plot ({\x}, {\yTwoW+1-1/((2+\xTwoW-\x)*\thetaVal))}) --  ([yshift=-\shift,xshift=2cm]002W.center) -- ([yshift=1cm+\shift,xshift=2cm]002W.center) -- ([yshift=1cm+\shift,xshift=1cm+\shift]002W.center) ;
\draw[darkblue, very thick, draw opacity=1,domain =\xMid:\xTwoW+1
]  plot ( {\x}, {\yTwoW+2- (\thetaVal+1)/(\thetaVal* (2-\x+\xTwoW))})  ;
\fill[darkblue,  fill=lightblue,  fill opacity=0.5,domain =\xMid:\xTwoW+1+\shift/\ptconv
]  plot ( {\x}, {\yTwoW+2- (\thetaVal+1)/(\thetaVal* (2-\x+\xTwoW))}) --([yshift=1cm+\shift,xshift=1cm+\shift]002W.center) --cycle;
\draw  ([xshift=1cm-\ptconv/(\thetaVal),yshift=-0.1cm]002W.center) node[scale=0.75,below]{$1-\frac1{\theta}$}-- ([xshift=1cm-\ptconv/(\thetaVal),yshift=0.1cm]002W.center);
\draw  ([yshift=2cm-\ptconv/(\thetaVal),xshift=-0.1cm]002W.center) node[scale=0.75,left]{$2-\frac1{\theta}$}-- ([yshift=2cm-\ptconv/(\thetaVal),xshift=0.1cm]002W.center);

\end{tikzpicture}}
    \caption{
    We visualise the Skyscraper Invariant by modules whose rank invariant is $s^\theta(\bullet, \bullet)$.}
        \label{fig:ex_skyscraper}
    \end{figure}
\end{example}

\section{Preliminaries.}

In the previous section, we have touched on many notions without properly defining them. The reader should view this section as a way to read further on this material if they are not familiar with it or, alternatively, skip it and come back to it if necessary.

In \autoref{subsec:bifiltr} we review bifiltrations as used in \autoref{ex:discrete_module} and the important \autoref{obs:factor_dimension}.
In \autoref{subsec:MPM-formal}, we introduce notation for multiparameter persistence theory, and in \autoref{subsec:kan} review discretization and restriction via Kan extensions.
Our definition of the HN filtration at $\alpha\in \bbR^d$, given in the previous section, is actually not consistent with the classical one \cite{Harder1974,King94}. Instead of a filtration of $\langle V_\alpha \rangle$, one really considers a filtration of $V$. In \autoref{subsec:HN-formal}, we will give a small overview of this theory. 

\subsection{Filtrations of simplicial complexes} \label{subsec:bifiltr}

We denote by $\simp$ the category of finite simplicial complexes and fix $d\geq 1$. We see $\bbR^d$ as a poset endowed with the product partial order. A $\bbR^d$-\emph{filtered simplicial complex} is given by a functor $K_\bullet\colon \bbR^d \to \simp$ such that for $\alpha\leq \beta\in \bbR^d$, the simplicial map $K_{\alpha\to \beta}$ is injective.  Depending on the application, the construction of a filtered simplicial complex from the data can take many forms. We now present several standard constructions.

\paragraph{Scale parameter} In one-parameter persistence ($d=1$), we often build a scale-indexed filtration of simplicial complexes from a finite metric space $(X,d_X)$. Given $r\in \bbR$, let $N_r(X)$ be the graph whose vertex set is $X$ and whose edges are all the $\set{x,y}\subset X$ such that $d_X(x,y)\leq r$. The \emph{Rips filtration} of $X$ is the $\bbR$-filtered simplicial complex $\Rips(X)$ given for $r\in \bbR$ by the clique complex of $N_r(X)$. Given $x\in X$, we denote by $\mathrm{Vor}(x)$ the cell of $x\in X$ in the Voronoi diagram of $X$. Moreover, for all $r\in \bbR$, we write $B(x,r)$ for the closed ball of radius $r$ around $x$. The \emph{alpha filtration} of $X$ is the $\bbR$-indexed simplicial complex $\Alpha (X)$ given for $r\in \bbR$ by the nerve of the sets $(B(x,r)\cap \mathrm{Vor}(x))_{x\in X}$. 

\paragraph{Density parameter} Assume that we have built the Rips filtration $\Rips(X)$ of a finite metric space $(X,d_X)$. We now describe ways to add a density parameter to $\Rips(X)$ and obtain a $\bbR^2$-filtered simplicial complex that can deal with noisy pointclouds like the one in \autoref{ex:discrete_module}.

When $(X,d_X)$ is a finite pointcloud in an euclidean space $E$, one can define the \emph{multicover filtration} \cite{sheehy12multicover}, given for $ (r,m)\in \bbR^2$ by the topological space
\begin{equation}
    \label{eqn:multicover}
    \multicover(X)_{r,m}\coloneqq\set{x\in E \mid \mu_X(B(x,r))\geq m},
\end{equation}
where $\mu_X$ denotes the counting measure of $X$. There exist several constructions of  bifiltered simplicial complexes which are topologically equivalent to $\multicover(X)$ \cite{multicover23}. In the rest of this paper, $\multicover(X)$ will denote any of those constructions.

Another approach is to use an external density estimate $\gamma \colon X\to P$. The $\gamma$-\emph{density-Rips complex} of $X$ is the $\bbR^2$-filtered simplicial complex $\Rips^\gamma(X)$ given for $(r,g)\in \bbR^2$ by
\begin{equation}
    \label{eqn:density-rips}\Rips^\gamma(X)_{r,g}\coloneqq\Rips( \gamma^{-1}([-g,+\infty))_r.
\end{equation}
The same construction can be applied to $\Alpha(X)$ instead of $\Rips(X)$, yielding the $\gamma$-\emph{density-alpha complex}. Namely, the $\gamma$-\emph{density-alpha complex} is given for $(r,g)\in \bbR^2$ by
\begin{equation}
    \label{eqn:density-alpha}
    \Alpha ^\gamma(X)_{r,g}\coloneqq\Alpha ( \gamma^{-1}([-g,+\infty))_r.
\end{equation}

\subsection{Persistence Modules}\label{subsec:MPM-formal}

 Let $\K$ be any field. We denote by $\Pers{P} \coloneqq \fun(P,\vect_\K)$ the category of functors from a poset $P$ to the category of finite-dimensional $\K$-vector spaces and call the objects of any such category \emph{Persistence Modules}. We say that $V\in \Pers{\bbR^d}$ is \emph{bounded} if its support is bounded in $\bbR^d$.
A $d$-parameter persistence module can also be seen as a module over the $\bbR^d$-graded algebra $
A \coloneqq \K\{x_1, \dots, x_d\} $
of polynomials in $d$ variables with \underline{real} exponents \cite[Theorem 1]{CZMP}.

\paragraph{Presentations.}  A module $F\colon \bbR^d \to\vect_\K$ is \emph{free} if it is isomorphic to a direct sum of interval modules whose support is $\set{\beta\in\bbR^d\mid \beta\geq\alpha}$ for some $\alpha\in \bbR^d$. A \emph{finite presentation} of $V$ is a map 
\[
F_1\overset\phi\longrightarrow F_0
\]
 between free modules such that $V\simeq \coker \phi$. When $V$ admits a finite presentation, we say that $V$ is \emph{finitely presentable}.
 When a Persistence Module $\module$ is obtained by the homology of a filtered \emph{finite} simplicial complex, then $\module$ must be finitely presentable.
 Henceforth, we will restrict ourselves to the category $\Persfp{\bbR^d}$ of finitely presentable persistence modules. We further denote by $\Persfpb{\bbR^d}$ the category of bounded and finitely presentable modules.

Each finitely presentable persistence module $V\in \Persfp{\bbR^d}$ has a finite minimal presentation, unique up to natural isomorphism. This minimal presentation can be extended to a \emph{minimal resolution} that is at most of length $d+1$. The degrees of the generators, relations, higher relations etc. in a minimal resolution are independent of the chosen resolution and form multi-subsets $b_0(V), \ b_1(V),  \dots , \ b_d(V) \subset \bbR^d$ -- the graded \emph{Betti numbers} of $V$.

A presentation (or rather a map between free modules) can be encoded in a \emph{graded matrix} $M$. Each row of $M$ corresponds to a generator and each column to a relation. The rows and columns are endowed with the parameter value -- or \emph{degree} -- of the corresponding generator or relation, making $M$ a \emph{graded matrix} or \emph{monomial matrix} \cite{miller00}. If we write $A[G]$ for the free module generated by a tuple of degrees $G$, then a graded matrix is nothing but a linear map $M \colon A[R] \to A[G]$ and it presents $V$ if $V \iso A[G]/\Img(M)$.

\paragraph{Invariants.}  In order to compare or vectorise Persistence Modules, we will consider invariants: assignments $I$, from Persistence Modules to some set $\mathcal P$ that satisfy
\[V\simeq W \implies I(V)=I(W).\]
We say that $I$ is an \emph{additive} invariant if $\mathcal P$ is an abelian group and for every $V, \ W$, we have 
\begin{equation}\label{eqn:additive_inv}
   I(V\oplus W) = I(V) + I(W).
\end{equation}

The \emph{Hilbert function} of $\module$, given by the pointwise dimension $\udim \module \colon  \alpha\in\bbR^d\mapsto \dim_\K \module_\alpha$, is an additive invariant; and so is the \emph{rank invariant}, which assigns to $V\colon \bbR^d\to\vect_{\K}$, the collection $\rank_V$ of the integers $(\rank V_{\alpha\to \beta})_{\alpha\leq \beta\in \bbR^d}$. 

\subsection{Change of Poset Domain.}\label{subsec:kan}

In the course of this paper, we will restrict our attention to sub-posets of $\bbR^d$ in multiple ways. To restrict or induce persistence modules supported on one poset to another, we make liberal use of Kan extensions because they make book-keeping easier. This comes at the price of losing some clarity in the exposition, so we will briefly review the theory here. More details on Kan extensions can be found in \cite[\S X.3]{mac1998categories}.

\paragraph{Kan Extensions.} Let $h\colon P\to Q$ be a map of posets. We denote by $h^*$ the restriction functor $-\circ h\colon \Pers Q\to \Pers P$. We denote by, respectively, $h_!\colon \Pers P\to \Pers Q$ and $h_*\colon \Pers P \to \Pers Q$ the left and right \emph{Kan extensions} along $h$. Equivalently, we can define them as the left adjoint $h_! \dashv h^*$ and right adjoint $h^* \dashv h_*$ of $h^*$.

\begin{remark}[{\cite[X.3, Theorem 1]{mac1998categories}}]\label{rem:Kan}
For $V\in \Persfp{P}$, the vector space at $\gamma \in Q$ of a Kan extension along $h$ can be computed via (co)limits with the formulas
\begin{equation}
    \label{eqn:Kan_pointwise}h_! (V)_\gamma \simeq \colim\limits_{h(\alpha)\leq \gamma } V_\alpha \qquad \text{ and } \qquad h_* (V)_\gamma \simeq \lim\limits_{h(\alpha)\geq \gamma} V_\alpha. 
\end{equation}
  Observe that when $h$ is fully faithful, the first unit $u^!_h \colon \Id_{\Pers P} \to h^* h_!$ and the second counit $c^*_h \colon h^* h_* \to \Id_{\Pers P}$ are isomorphisms. 
\end{remark}

\begin{definition} 
    For any subset $S \subset \bbR^d$, we also write $\langle S \rangle \coloneqq \{\alpha \in \bbR^d \mid \exists s \in S \colon s \leq \alpha \}$ for the induced sub-poset.
Let $B \subset \bbR^d$ be any subset, then we treat it as an induced sub-poset and denote by $\iota_B \colon B \to \bbR^d$ the corresponding fully faithful functor. In particular, the unit $u^!_{\iota_B}$ and counit $c^*_{\iota_B}$ are isomorphisms by the previous \autoref{rem:Kan}.
For $\alpha \in \bbR^d$, we write $\iota_\alpha$ as shorthand for $\iota_{\langle \alpha \rangle} \colon \up{\alpha} \into \bbR^d$.
\end{definition}

\begin{example}\label{ex:kan}
    Let $\alpha\in\bbR^d$ and $V \in \Pers{\langle \alpha \rangle}$. (\autoref{eqn:Kan_pointwise}) immediately verifies that $({\iota_\alpha})_! $ extends $V$ by $0$, whereas $({\iota_\alpha})_*$ extends $V$ via isomorphisms along its values on the boundary of $\langle \alpha \rangle$. That is, for every $\gamma \in \bbR^d$, 
$\left( (\iota_\alpha)_* V\right)_\gamma \iso V_{\gamma \vee \alpha}$, where $\vee$ denotes the join in $\bbR^d$.
\end{example}

\paragraph{Grids, Discretisation, and Kan extensions.}
A \emph{grid} $\gri\subset \bbR^d$ is the product of $d$ finite subsets of $\bbR$.  A grid $\gri$ is called \emph{regular} of width $\varepsilon$ if it is the intersection of a cube of $\bbR^d$ with $\varepsilon\bbZ^d$. Let $\gri$ be a grid, we denote by $\gri^+$ the extension of $\gri$ where we formally add $\pm \infty$ as a value in every coordinate. We define for $\alpha \in\bbR^d$
\[
\lfloor \alpha \rfloor_\gri\coloneqq \max\set{g \in \gri^+,\
g \leq \alpha} \quad \text{ and } \quad \lceil \alpha \rceil_\gri \coloneqq \min\set{g \in \gri^+,\
g \geq \alpha}.
\]
The fibres of $\lfloor-\rfloor_\gri$ define a tiling of $\bbR^d$ into rectangular cells:
\[
\bbR^d = \bigcup_{\alpha \in \gri^+}  \set{
\gamma\in \bbR^d\mid \lfloor \gamma\rfloor_\gri = \alpha
}
\]
and we will say that $\beta, \gamma$ are in the same cell iff $\lfloor \beta\rfloor_\gri = \lfloor \gamma\rfloor_\gri $
We treat a grid $\gri\subset \bbR^d$ as a sub-poset. Since, as a sub-poset, a grid has all limits and colimits, the formulas for Kan extensions become easy for the inclusion $\iota_\gri$:

\begin{proposition}[cf. \cite{botnan2021signed,chacholski2021realisations,blanchette2023exact}]
\label{prop:computation_kan_extension}
Let $V \in \Pers{\bbR^d}$. The left Kan extension wrt. $\gri$ satisfies
\[
((i_\gri)_!V)_\alpha = \begin{cases}
    \module_{\lfloor \alpha\rfloor_\gri}&\text{if }\lfloor\alpha\rfloor_\gri \in \gri\\
    0&\text{otherwise}
\end{cases}
\quad \text{ and } \quad 
((i_\gri)_!V)_{\alpha\to \beta} = V_{\lfloor \alpha\rfloor_\gri\to \lfloor \beta\rfloor_\gri}
\]
and the right satisfies the dual statement.
\end{proposition}

\begin{proof}
Follows directly from \autoref{rem:Kan}, because the comma categories will have a final (cofinal) element given by the maximum (minimum).
\end{proof}

\begin{definition}
    The \emph{induced} grid of a module $V \in \Persfp{\bbR^d}$, which we denote by $\grid(V)$, is the smallest grid  of $\bbR^d$ containing $b_0(V) \cup b_1(V)$. 
\end{definition}
The induced grid of a module contains all of its graded Betti numbers \cite[Theorem 3.1]{Bruns_Herzog_1995}. Finitely presentable persistence modules can be discretised over a grid (e.g. \cite[Lemma 7.13]{blanchette2023exact}): For every grid $\gri$ containing $\grid (V)$, we have that the counit
$ c^!_{\iota_\gri}\colon(i_\gri)_!(i_\gri)^* \module \to \module $ is an isomorphism.

\begin{lemma}
    \label{lem:same-pos-iso}
    Given $V\in\Persfp{\bbR^d}$ and $\alpha,\beta\in \bbR^d$ belonging to the same cell of $\grid (V)$, the map $V_{\alpha\to\beta}$ is an isomorphism.
\end{lemma}
\begin{proof}
Let $\gri = \grid (V)$, we have $(i_\gri)_!i_\gri^* V \overset{c^!_{\iota_\gri}}{\iso} V$ and since $\lfloor \alpha\rfloor_\gri =  \lfloor \beta\rfloor_\gri$
\[
 ( (i_\gri)_!i_\gri^* V) _{\alpha\to \beta} \overset{\ref{prop:computation_kan_extension}}= (i_\gri^* V)_{\lfloor \alpha\rfloor_\gri\to \lfloor \beta\rfloor_\gri}  =\Id.
\]
Hence, $V_{\alpha\to \beta} $ is an isomorphism.
\end{proof}

\paragraph{More Kan extensions.}
When the sub-poset $\iota \colon Q \into \bbR^d$ is a grid, we can see by calculating the Kan extensions $\iota_!$, $\iota_*$ pointwise that they are actually exact functors, and so themselves have a left and right adjoint respectively. If $Q$ is an up or down set in $\bbR^d$, we have seen the same in \autoref{rem:Kan} and for this case we will compute them explicitly, purely out of curiosity:

\begin{definition}
For $\module \in \Persfp{\bbR^d}$ and $B \subset \bbR^d$, the cokernel of the counit $c_! \colon (\iota_{\bbR^d \setminus B })_! (\iota_{\bbR^d \setminus B })^* \module \to V$ defines a functor, $\coker_B V$, equipped with a natural transformation $q^B \colon \module \onto \coker_B V$. 
Dually, the kernel of the unit $u_* \colon \module \to (\iota_{\bbR^d \setminus B})_* (\iota_{\bbR^d \setminus B})^* \module $ defines a functor, $\ker_B$, equipped with a natural transformation, $i^B \colon \ker_B \module \into \module$
\end{definition}

\begin{proposition}\label{prop:left_kan_exact}
If $B$ is an up set, then the functor $\coker_B - $ is left adjoint to $(\iota_B)_!$. In particular, $(\iota_B)_!$ is exact.
Dually, if $B$ is a down set, then $ker_B(-)$ is right adjoint to $(\iota_B)_*$. 
\end{proposition}

\ignore{
\begin{remark}
We cannot directly use Eilenberg-Watts in this graded case. Instead we can represent these functors pointwise. For $\alpha \in \bbR^d$, $\left(\ker_B V \right)_\alpha \coloneqq \Hom\left ( \K [B \cap \langle \alpha \rangle], V \right)_0$
and $\left( \coker_B \module \right)_\alpha \coloneqq \left( V \otimes ??? \right)_0$.
For $\alpha \leq \beta$, $\left(\ker_B V \right)_{\alpha \to \beta} \coloneqq \left( \K (B \cap \langle \beta \rangle) \into \K (B \cap \langle \alpha \rangle) \right)^*$ and $\left( \coker_B \module \right)_{\alpha \to \beta} \coloneqq ??$. 
\end{remark}
 }

\begin{proof}
   Let $f \colon V \to (\iota_B)_!W$, then this diagram % https://q.uiver.app/#q=WzAsOCxbMSwwLCJWIl0sWzEsMSwiKFxcaW90YV9cXGFscGhhKV8hVyJdLFswLDAsIihcXGlvdGFfe1xcbGFuZ2xlIFxcYWxwaGEgXFxyYW5nbGVeXFxjb21wbGVtZW50fSlfIVxcaW90YV97XFxsYW5nbGUgXFxhbHBoYSBcXHJhbmdsZV5cXGNvbXBsZW1lbnR9XipWIl0sWzAsMSwiKFxcaW90YV97XFxsYW5nbGUgXFxhbHBoYSBcXHJhbmdsZV5cXGNvbXBsZW1lbnR9KV8hXFxpb3RhX3tcXGxhbmdsZSBcXGFscGhhIFxccmFuZ2xlXlxcY29tcGxlbWVudH1eKihcXGlvdGFfXFxhbHBoYSlfIVcgXFxzaW1lcSAwIl0sWzIsMCwiVl97fD8/fSJdLFsyLDEsIlcgXFxzaW1lcSBcXGxlZnQoIChcXGlvdGFfXFxhbHBoYSlfIVcgXFxyaWdodClfe3w/P30iXSxbMywwLCIwIl0sWzMsMSwiMCJdLFsyLDAsImMiXSxbMywxLCJjIl0sWzIsMywiKFxcaW90YV97XFxsYW5nbGUgXFxhbHBoYSBcXHJhbmdsZV5cXGNvbXBsZW1lbnR9KV8hXFxpb3RhX3tcXGxhbmdsZSBcXGFscGhhIFxccmFuZ2xlXlxcY29tcGxlbWVudH1eKmYiLDJdLFswLDEsImYiXSxbMCw0LCI/PyJdLFsxLDVdLFs0LDUsIlxcdGlsZGUgZiIsMCx7InN0eWxlIjp7ImJvZHkiOnsibmFtZSI6ImRvdHRlZCJ9fX1dLFs0LDZdLFs1LDddXQ==
\begin{center}\begin{tikzcd}[ampersand replacement=\&]
	{(\iota_{\bbR^d \setminus B })_! (\iota_{\bbR^d \setminus B  })^*V} \& V \& {\coker_B V} \& 0 \\
	{(\iota_{\bbR^d \setminus B })_! (\iota_{\bbR^d \setminus B  })^*(\iota_B)_!W \simeq 0} \& {(\iota_B)_!W} \& {W \simeq \coker_B \left( (\iota_B)_!W \right)} \& 0
	\arrow["c", from=1-1, to=1-2]
	\arrow["{(\iota_{\bbR^d \setminus B })_! (\iota_{\bbR^d \setminus B  })^*f}"', from=1-1, to=2-1]
	\arrow["{q^B_V}", from=1-2, to=1-3]
	\arrow["f", from=1-2, to=2-2]
	\arrow[from=1-3, to=1-4]
	\arrow["{\tilde f}", dotted, from=1-3, to=2-3]
	\arrow["c", from=2-1, to=2-2]
	\arrow["{q^B_{(\iota_B)_!W}}", from=2-2, to=2-3]
	\arrow[from=2-3, to=2-4]
\end{tikzcd}
\end{center}
 defines a map $\widetilde f \colon \coker_B V \to (\iota_B)_!W$, which is equivalent to a map $\iota_B^* \coker_B V \to W$, since the former is supported on $B$. Similarly, every map $g \colon \iota_B^* \coker_B V \to W$ induces a map $\bar g \colon  \coker_B V \to (\iota_B)_! W$ which produces a map $q^B_V \circ \bar g \colon V \to (\iota_B)_! W$ and these two constructions are inverse to each other. The other statement is dual.
\end{proof}

At last, we want to understand how the left Kan extension changes the presentation matrix of a persistence module.

\begin{lemma}\label{lmm:restriction}
Let $M \in \K^{G \times R}$ be a graded matrix representing a map $A[R] \to A[G]$ of free persistence modules and $\alpha \in \bbR^d$. We denote by $G \vee \alpha$ and $R \vee \alpha$ the component-wise join with $\alpha$ and by $M_{\vee \alpha} \in \K^{(G \vee \alpha) \times (R \vee \alpha)}$ the matrix with the same entries as $M$. Then $(\iota_{\alpha})_! (\iota_{\alpha})^* M = M_{\vee \alpha}$ and, in particular, $M_{\vee \alpha}$ is again a graded matrix.
\end{lemma}

\begin{proof}
Follows from \cite[Lemma 4.2]{djk_arxiv}.
\end{proof}

\begin{corollary}\label{cor:res_res}
    If $ \dots \to A[S] \oto{N} A[R] \oto{M} A[G] \to V \to 0$ is a resolution of $V$, then 
    \[ \dots \to A[S \vee \alpha] \oto{N_{\vee \alpha}} A[R \vee \alpha] \oto{M_{\vee \alpha}} A[G \vee \alpha] \to (\iota_{\alpha})_! (\iota_{\alpha})^* V \to 0 \] 
    is a resolution of $(\iota_{\alpha})_! (\iota_{\alpha})^* V$.
\end{corollary}

\begin{proof}
The functor $(\iota_{\alpha})_! (\iota_{\alpha})^*$ is exact by \autoref{prop:left_kan_exact} and the identifications follow from \autoref{lmm:restriction}. 
\end{proof}

If we extend from another poset $B$ where the functor $(\iota_B)_!$ is exact (for example, a grid), then the effect on the presentation matrix can be described analogously.

\subsection{Classical definition of HN filtrations}
\label{subsec:HN-formal}
We now describe why (\autoref{eq:quotient}) is a special case of the classical definition of HN filtrations. Given an abelian category $\calA$, the \emph{Grothendieck group} of $\mathcal A$ is defined as the isomorphism classes of objects in $\calA$ subject to the relations $[V]=[U]+[W]$ for every short exact sequence $0\to U\to V\to W\to 0$ in $\calA$. A \emph{stability condition} is a group morphism  $Z\colon K(\calA)\to \bbC$ such that for every nonzero $V\in \calA$, we have $\mathrm{Real}Z([V])>0$.  The slope $\mu_Z$ assigns to a nonzero $V\in \calA$ the number $\mathrm{Imag}Z([V])/\mathrm{Real}Z([V])$. Moreover, $V$ is said to be $Z$-\emph{semistable} if for $0\subsetneq W\subsetneq V$ we have $\mu_Z(W)\leq \mu_Z(V)$. If $V\in\calA$ admits a filtration $0=F^0\subsetneq F^1\subsetneq \dots \subsetneq F^\ell = V$ such that the successive subquotients $(F^i/F^{i-1})_{1\leq i\leq \ell}$ are $Z$-semistable and are sorted by decreasing $Z$-slope, then the filtration $F^\bullet $ is unique and is called the \emph{classical HN filtration} of $V$ along $Z$ \cite{bridgeland2007stability}. 

    In this paper, we are mainly interested in the setting introduced in \autoref{subsec:skyscraper-informal-def}  and in its discretization:
    \begin{enumerate}
        \item Bounded finitely presented $d$-parameter persistence modules: $\calA\coloneqq \Persfpb{\bbR^d}$ and for $\alpha\in \bbR^d$ we consider the stability condition $Z_\alpha$ defined for $W\in\Persfpb{\bbR^d}$ by 
\begin{equation}\label{eqn:Zalpha-cont}Z_\alpha([W]) \coloneqq \left(\int_{\bbR^d}\udim W\right)+i\dim_\K W_\alpha.\end{equation}
        \item Persistence modules over a finite poset $P$:  $\calA$ is the category $\Pers P$ and for $\alpha\in P$ we consider the stability condition $\widehat Z_\alpha$ defined for $\widehat V\in\Pers P$ by 
\begin{equation}\label{eqn:Zalpha-discr}
\widehat Z_\alpha([\widehat V]) \coloneqq \left(\sum_{\beta\in P}\dim_\K \widehat V_{\beta}\right)+i\dim_\K \widehat V_\alpha.\end{equation}
    \end{enumerate}

One can go from one setting to the other because HN filtrations behave well with restrictions and extensions: 
\begin{proposition}
    [{\cite[Theorem 2.9]{fersztand2024harder}}]\label{prop:ext-cont}
    Let $V\in \Persfpb{\bbR^d}$ and let $\alpha\in \bbR^d$. Assume that there is a regular grid $\gri$ containing $\grid (V)\cup\set{\alpha}$. If $0\subsetneq \widehat F^1\subsetneq \dots\subsetneq \widehat F^{\ell-1}\subsetneq \widehat F^\ell$ is the HN filtration of $i_\gri^*V$ along the stability condition $\widehat Z_\alpha$  (\autoref{eqn:Zalpha-discr}), then the HN filtration of $V$ along $Z_\alpha$ (\autoref{eqn:Zalpha-cont}) is given by 
    \[
    0\subsetneq (i_\gri)_!\widehat F^1 \subsetneq \dots\subsetneq  (i_\gri)_! \widehat F^{\ell-1}\subsetneq  (i_\gri)_! \widehat F^\ell.\qedhere
    \]    
\end{proposition}

\begin{remark}[Other stability conditions]\label{rem:other_stab_cond} We refer the reader to {\cite[Theorem 2.9]{fersztand2024harder}} for a more general version of \autoref{prop:ext-cont}, which allows other choices of stability conditions and discretising grids. For example, one could replace the real part of (\autoref{eqn:Zalpha-cont}) with $\int_{\bbR^d}f\udim W$ and the real part of (\autoref{eqn:Zalpha-discr}) with $\sum_{\beta\in P}\left( \int_{\lfloor \alpha\rfloor_\gri = \beta} f\right)\dim_\K  \widehat V_{\beta}$, where $f\colon \bbR^d\to\bbR$ is an arbitrary locally integrable  positive function. In that case, we can discretise the HN filtrations along $Z_\alpha$  over any grid $\gri$ containing $\grid(V)\cup\set\alpha$. One can also replace (\autoref{eqn:Zalpha-cont}) with $\int h \udim V$, where $h$ is a locally integrable step-function whose real part is positive and whose imaginary part is bounded.  If the function $f$ (\textit{resp.} $h$) is a step-function with values in $\bbZ$ (\text{resp.} in $\bbZ[i]$), Cheng's algorithm (\autoref{sec:cheng}) can still compute an approximation of the Skyscraper Invariant. However, the time-complexity would then depend polynomially on the maximum of $f$ and the number of value changes of $f$.
\end{remark}
    
Note that the slope $\mu_{Z_\alpha}$ coincides with the slope introduced in (\autoref{eq:quotient}). Moreover, we observe that for every nonzero $V\in \Persfpb{\bbR^d}$, we have
\[
\mu_{Z_\alpha}(V) \leq \mu_{Z_\alpha}(\langle V_\alpha\rangle).
\]
Hence, semistability at $\alpha$ can be checked by only looking at submodules induced by subspaces of $V_\alpha$. One can further prove that the notions of HN filtrations introduced in (\autoref{eq:quotient}) coincides with the classical notion of HN filtration:

\begin{proposition}
      [{\cite[Proposition 3.3]{HN_discr}}]\label{prop:HN-formal} 
Fix $\alpha\in \bbR^d$ and  $V\in\Persfpb{\bbR^d}$. Let $0= \langle F^0 \rangle \subsetneq \langle F^1 \rangle\subsetneq \dots\subsetneq \langle F^\ell \rangle =\langle V_\alpha \rangle$ be the filtration (\autoref{eq:quotient}). Then 
\[
0= \langle F^0 \rangle \subsetneq \langle F^1 \rangle\subsetneq \dots\subsetneq \langle F^\ell \rangle =\langle V_\alpha \rangle\subseteq V.
\]
is the classical HN filtration of $V$ along the stability condition $Z_\alpha$ defined in (\autoref{eqn:Zalpha-cont}).
\end{proposition}
\begin{proof}
    {\cite[Proposition 3.3]{HN_discr}} shows the results for persistence modules over finite posets. {\cite[Theorem 2.9]{fersztand2024harder}} extends this result to $\Persfpb{\bbR^d}$.
\end{proof}

It follows from the classical definition of HN filtrations that the factors $U^\bullet$ in (\eqref{eq:quotient}) have decreasing slope at $\alpha$. The modules $F^\bullet$ also have decreasing slopes at $\alpha $ by the following see-saw lemma:

\begin{lemma} {\cite[Lemma 2.1]{HILLE2002205}}\label{lmm:see-saw} Let $0\to V^1\to V^2\to V^3\to 0 $ be a short exact sequence of persistence modules and let $Z$ be a stability condition. The sequence of slopes $(\mu_Z(V^i))_{1\leq i\leq 3}$ is either strictly increasing, strictly decreasing, or constant. 
\end{lemma}

\begin{remark}\label{rmk:skyscraper-finite}
    The construction of the Skyscraper Invariant can be adapted to a persistence module $\widehat V$ indexed by a finite poset $P$. Given $\alpha\leq \beta$ in $P$, we denote by $\widehat F^\bullet$ the HN filtration of $\widehat V$ along $\widehat Z_\alpha$, and define for $\theta\in\bbR$
   \[ s^\theta_{\widehat V}(\alpha, \beta) \coloneqq \dim_\K \langle \widehat U^{i} \rangle_\beta  \quad 
\text{ where $i$ is maximal with $\mu(\widehat U^{i}) \geq \theta$.} \qedhere\]
\end{remark}

\section{Computing the Skyscraper Invariant Approximately}\label{sec:approx}

In \autoref{subsec:skyscraper-informal-def}, we have presented some basic properties of the Skyscraper Invariant.
The first simplification for its computation is the computation and the decomposition into indecomposable summands of the submodule $\langle V_\alpha \rangle$ for each $\alpha\in\bbR^d$. We gather this idea into \autoref{alg:skeleton}, the main skeleton for our algorithms.

\subsection{Computing Induced Submodules.} Computing a presentation for any induced submodule can be done by computing the kernel of a graded matrix in the category of persistence modules. For $d=2$, this can be done with the LW-algorithm (Lesnick, Wright \cite{lesnick_wright}) and its improvement by Kerber, and Rolle, \textsc{mpfree} \cite{mpfree} which we will henceforth call the LWKR-algorithm.  For $d=3$, there is an algorithm \cite{kernels} by Bender, Gäfvert, and Lesnick, but it is still unpublished as of writing this and we were not yet able to integrate it into our library. When $d$ is larger, we rely on classical machinery, like Schreyer's algorithm \cite{schreyer}, which is too slow in practice. We will therefore focus on the $2$-parameter case for this paper.

\begin{algorithm}
\caption{\texttt{SubmoduleGeneration} (cf. \cite[Section 7.2]{allo})}
\label{alg:submodule}
\DontPrintSemicolon
\KwIn{$M \in \K^{G \times R}$ presenting $\module$ and a finite subset $S \subset \module$ as a graded matrix $M_S \in \K^{G \times \deg(S)}$}
\KwOut{A presentation $N \in \K^{\deg(S)\times R_S}$ of $\langle S \rangle$}\;
Compute $\begin{bmatrix} N \\ K \end{bmatrix} \coloneqq \ker \left[ M_S \ M \right]$. $\begin{bmatrix} N \\ K \end{bmatrix} \in \K^{ (\deg(S) \cup R) \times R_S}$ \;
\Return $N \in \K^{\deg(S) \times R_S}$ \;
\end{algorithm}

\begin{proposition}\label{prop:submodule_generation}
\autoref{alg:submodule} computes a presentation of  $\langle S \rangle$.
\end{proposition}

\begin{proof}
Denote by $p \colon A[G] \to \module$ the cokernel of $M$, then $\langle S \rangle \simeq  \Img \left( p \circ M_S \right) $.
By construction we have
\[ 0 = \begin{bmatrix} M_S & M \end{bmatrix} \begin{bmatrix} N \\ K \end{bmatrix} = M_S N + M K. \]
Then the following diagram of exact sequences commutes, which implies $q \circ N  =0$.
% https://q.uiver.app/#q=WzAsOCxbMSwxLCJBXkciXSxbMCwxLCJBXlIiXSxbMSwwLCJBXntcXGRlZyhTKX0iXSxbMiwxLCJWIl0sWzIsMCwiXFxsYW5nbGUgUyBcXHJhbmdsZSJdLFswLDAsIkFee1JfU30iXSxbMywwLCIwIl0sWzMsMSwiMCJdLFsyLDAsIlMiXSxbMSwwLCJNIl0sWzAsMywicCJdLFs0LDMsIlxcaW90YSIsMCx7InN0eWxlIjp7InRhaWwiOnsibmFtZSI6Imhvb2siLCJzaWRlIjoiYm90dG9tIn19fV0sWzIsNCwicSJdLFs1LDEsIi1LIiwyXSxbNSwyLCJOIl0sWzQsNl0sWzMsN11d
\begin{center}
\begin{tikzcd}[ampersand replacement=\&]
	{A[R_S]} \& {A[\deg(S)]} \& {\langle S \rangle} \& 0 \\
	{A[R]} \& {A[G]} \& \module \& 0
	\arrow["N", from=1-1, to=1-2]
	\arrow["{-K}"', from=1-1, to=2-1]
	\arrow["q", from=1-2, to=1-3]
	\arrow["M_S", from=1-2, to=2-2]
	\arrow[from=1-3, to=1-4]
	\arrow["\iota", hook', from=1-3, to=2-3]
	\arrow["M", from=2-1, to=2-2]
	\arrow["p", from=2-2, to=2-3]
	\arrow[from=2-3, to=2-4]
\end{tikzcd}    
\end{center}
Let $t \in A[\deg(S)]$ be such that $q(t) =0$. A diagram chase then implies the existence of $u \in A[R]$ such that $Mu - M_S t =0$. By construction of $[N \ K]$ there is $v \in A[R_S]$ such that $Nv = t$.
\end{proof}

\subsection{Storing the Skyscraper Invariant}\label{subsec:storing-skyscraper}
Given $V\in \Persfpb{\bbR^d}$ and $\alpha\in \bbR^d$, we will represent $s^\bullet_V(\alpha,\bullet)$ by the Hilbert functions of the factors $\langle U_i \rangle$ of the HN filtration of $\langle V_\alpha\rangle$, together with their slope.

To store the Hilbert function of a module $\module$, 
we will use a collection of intervals whose indicator functions sum up to $\udim V$. Since the filtration factors are always uniquely generated,  these intervals will be too. Such modules are called \emph{staircases}, because of their shape. That is, even if the HN filtration itself is not always built from intervals, the Skyscraper Invariant is represented by collections of intervals, which we store via their graded Betti numbers.

One could also represent the Hilbert function of a factor directly via its graded Betti numbers without cutting it up into intervals. The advantage of our approach is that over $\bbR^2$, the relations of an interval form an anti-chain by ordering with respect to $x$ or $y$ coordinate. Point-location in the interval then becomes logarithmic in the number of relations.
 
\begin{definition}
We define the \emph{thickness} of $V$ as $\thick{V} \coloneqq \max\limits_{\alpha \in \bbR^d}\dim_\K V_\alpha$. 
\end{definition}

For each $\alpha\in\bbR^d$, our data structure for  $s_V^\bullet(\alpha,\bullet)$ has at most $\thick{V}$ entries. To compute $s_V^\theta (\alpha,\beta)$, we have to count how many staircases with associated slope larger than $\theta$ contain $\beta$. This can be done in time $\Oc(\log(m))$, where $m$ is the number of relations of the interval. $m$ is bounded by the number of all relations ($\leq n$) of the module. This puts the query time in $\Oc(\thick{V}\log(n))$.

\subsection{Approximating the Skyscraper Invariant.}\label{subsec:approx-of-skyscr} 

Assume now that we have a subroutine $\text{ HN-Filtration}_\alpha  \text{-Factors}$ to compute the filtration factors $U_i$ and their slopes from a uniquely generated module. 
We form \autoref{alg:skeleton} using the data structure described in \autoref{subsec:storing-skyscraper} and the ideas from \autoref{subsec:skyscraper-informal-def}. 

\begin{algorithm}[!ht]
\caption{$\varepsilon$-approximation of the Skyscraper Invariant}
\label{alg:skeleton}
\DontPrintSemicolon
\KwIn{ $\varepsilon \in \bbR$; a graded matrix $M \in \K^{G \times R}$ presenting $V \in \Persfpb{\bbR^d}$.} 
\KwOut{For each $\alpha \in \varepsilon \bbZ^d \cap \supp(\module)$: Lists $(W_{t,j}, \mu_{t})_{j \in \mathcal{J}}$ of elements in $\mathcal{P}(\bbR^d) \times \bbR$.}

Decompose $V$ into $\bigoplus_{i \in I} V_i$ \;

\For{$\alpha \in \varepsilon \bbZ^d \cap \supp(\module)$}{
    \For{$i \in I$}{
     Compute a presentation $N_{i}$ for $\langle (V_{i})_{\alpha} \rangle$ with \autoref{alg:submodule}\;
      Decompose $N_{i}  \iso  \bigoplus_{j \in J_{i}} N_{i,j} $\;
       \For{$j \in J_{i}$}{
            $ ( U_{i, j, t}, \ \mu_{i,j,t} )_{t \in [l_{i, j}]} \gets \text{ HN-Filtration}_\alpha  \text{-Factors}(  N_{i,j} )$ \; \label{line_HN}
       }
    }
    \label{line:recompute}
    $(U_t, \mu_t)_{t \in[l]} \gets \mathsc{sort} \left( \{ ( (U_{i, j,t}, \  \mu_{i, j, t} )_{t \in [l_{i,\alpha, j}]} \}_{i,j} \right)$\;
    Decompose $\udim U_{t}$ into intervals $ \{W_{t,j}\}$ in $\bbR^d$\;
    \Return $(W_{t,j}, \mu_t)_{j \in \mathcal{J}}$\;
}
\end{algorithm}

It is not clear why the output of \autoref{alg:skeleton} should give an $\varepsilon$-approximation of the Skyscraper Invariant. To be able to even state such an approximation result, we need to have a metric to compare different Skyscraper Invariants.

For each $\theta\in \bbR$, $s_V^\theta(\bullet,\bullet)$ can be seen as a \emph{functor} (i.e. monotonic function) from $(\bbR^d)^\op\times \bbR^d$ to $(\bbZ\cup\set{+\infty})^\op$ by setting $s_V^\theta(\alpha,\beta)\coloneqq +\infty$ whenever $\alpha\not\leq\beta$. Such functors can be equipped with the \emph{erosion distance}:
\begin{definition}[\cite{patel2018generalized,puuska2020erosion}]
   \label{def:erosion-distance}
    The \emph{erosion distance} $d_\text{erosion}(r,s)$ between two functors  $r,s\colon (\bbR^d)^\op\times\bbR^d\to (\bbZ\cup\set{+\infty})^\op$ is defined as the infimum of all $\varepsilon>0$ such that for all $\alpha,\beta\in\bbR^d$
\[
s(\alpha-\vec\varepsilon,\beta+\vec\varepsilon)\leq r(\alpha,\beta)\qquad\text{and}\qquad r(\alpha-\vec\varepsilon,\beta+\vec\varepsilon)\leq s(\alpha,\beta).
\]
where $\vec\varepsilon$ is the vector $(\varepsilon\dots\varepsilon)^t\in\bbR^d$.
\end{definition}

Let  $s^\bullet(\bullet,\bullet)\colon \bbR\times(\bbR^d)^\op\times\bbR^d\to (\bbZ\cup\set{+\infty})^\op$ be a functor and let $\varepsilon>0$. An $\varepsilon$-\emph{approximation} of $s$ is the data for every $\theta\in\bbR$ of a functor $\widetilde s^\theta\colon (\bbR^d)^\op\times\bbR^d\to (\bbZ\cup\set{+\infty})^\op$ such that
$d_\text{erosion}(\widetilde s^\theta,s^\theta)\leq\varepsilon$. Let $\vec\varepsilon$ be the vector $(\varepsilon\dots\varepsilon)^t\in\bbR^d$. One way to obtain an $\varepsilon$-approximation is to discretize $\alpha$ over the smallest grid $\gri$ which contains $\varepsilon\bbZ^d \cap (\supp(V) +[0,\vec\varepsilon))$
\begin{equation}
    \label{eqn:approx}
\widetilde s^\bullet(\bullet,\bullet)\coloneqq 
    s^\bullet(\lfloor\bullet\rfloor_{\gri},\bullet).
\end{equation}

\begin{lemma}\label{lmm:restriction_of_inv}
$\widetilde s^\bullet(\bullet,\bullet)$ is an $\varepsilon$-approximation of $s^\bullet(\bullet,\bullet)$. 
\end{lemma}

\begin{proof}
   Let $\theta\in \bbR$ and let $\alpha,\beta \in \bbR$. Since $\alpha-\vec\varepsilon\leq \lfloor \alpha\rfloor_\gri\leq \alpha$ and $s^\theta\colon (\bbR^d)^\op\times \bbR^d\to (\bbZ\cup\set{+\infty})^\op$ is a functor, we have
    \begin{align*}
        \widetilde s^\theta(\alpha-\vec\varepsilon,\beta+\vec\varepsilon) & = s^\theta(\lfloor \alpha-\vec\varepsilon\rfloor_\gri,\beta +\vec\varepsilon) \leq s^\theta(\alpha,\beta)\\
        s^\theta(\alpha-\vec\varepsilon,\beta+\vec\varepsilon) &\leq s^\theta(\lfloor\alpha\rfloor_\gri,\beta) = \widetilde s^\theta(\alpha,\beta).
    \end{align*}
Hence, $d_\text{erosion}(\widetilde s^\theta,s^\theta)\leq\varepsilon$. 
\end{proof}
The approximation of the Skyscraper Invariant $\widetilde s_V^\bullet(\bullet,\bullet)$ defined in (\autoref{eqn:approx}) is represented as a dictionary $\set{\alpha:s^\bullet(\alpha,\bullet),\ \alpha\in \varepsilon\bbZ^d\cap\supp (V)}$. Given $\alpha\in\bbR^d$, the functor $\widetilde s^\bullet(\alpha,\bullet)$ can be obtained by returning $s^\bullet(\lfloor \alpha\rfloor_{\gri},\bullet)$ if 
the key $\lfloor \alpha\rfloor_{\gri}$ exists in the dictionary and 0 otherwise.

\subsection{Correctness and Runtime.}

Denote by $n = |G| + |R|$ the input size, $\thick{\module}$ the thickness of $\module$, $k$ the maximal thickness of a uniquely generated indecomposable submodule, and by $T_{\mathsc{HN-f}}$, $T_{\mathsc{ker}}$, and $T_{\mathsc{dec}}$ the times to compute the HN filtration, kernels, and decomposition.

\begin{proposition}\label{prop:skeleton_time}
 \autoref{alg:skeleton} returns an $\varepsilon$-approximation of the Skyscraper Invariant in 
 \[ \Oc \left( \frac{\mathrm{diam}(\supp V)}{\varepsilon^d} \left(T_{\mathsc{ker}}(n,d) + T_{\mathsc{dec}}(n,d) + (\thick{\module}/k)  (T_{\mathsc{HN-f}}(n, k, d) \right) \right). \]

\end{proposition}

\begin{proof}
It follows from additivity (\autoref{eq:additivity}), the recursive nature of the slopes (\autoref{eq:slope_recomputation}), and our choice of data structure described in \autoref{subsec:storing-skyscraper} that this algorithm outputs the functor $\widetilde s_V^\bullet(\bullet,\bullet)$ defined in (\autoref{eqn:approx}). By \autoref{lmm:restriction_of_inv}, this functor is an $\varepsilon$-approximation of the Skyscraper Invariant. At each step, there are at most $\thick{V}/k$ indecomposable submodules of thickness $k$ generated at $\alpha$.
\end{proof}

\section{Cheng's Algorithm}\label{sec:cheng}

 Building on the work of \cite{derksen2017polynomial, huszar2021non}, Cheng reduced the computation of HN filtrations of representations of finite acyclic quivers to the shrunk subspace problem which can be solved in deterministic polynomial-time for large enough fields \cite{cheng2024deterministic}. This algorithm can be directly applied to persistence modules over a finite poset \cite[Remark 1.13]{fersztand2024harder}. 

 After introducing the shrunk subspace problem (\autoref{pb:c-shrunk}), we describe how to apply Cheng's algorithm to obtain an approximation of the Skyscraper Invariant (\autoref{alg:compl-naive} and \autoref{cor:cheng_naive_analysis}). We then give an overview of the random algorithm for computing shrunk subspaces introduced in \cite[Algorithm 5]{franks2022shrunk_arxiv} (\autoref{algo:compute-cshrunk}), which we implement here for the first time. Finally, we improve the complexity analysis of \cite[Algorithm 5]{franks2022shrunk_arxiv} in the special case of inputs coming from HN computations (\autoref{prop:cheng_advanced_analysis}).

\subsection{From Skyscraper Invariant to shrunk subspace}

For this section, fix two integers $N,N'>0$ and a field $\K$. The space of $N\times N'$ matrices with coefficients in $\K$ is denoted $\K^{N\times N'}$. We further denote by $\otimes$ the Kronecker product.

\begin{definition}[shrunk subspace]\label{pb:c-shrunk}
  Let $\calA\subset \K^{N\times N'}$ be a subspace. A subspace $U\subset \K^{N'}$ 
    \[ \text{is a \emph{shrunk subspace} of $\calA$ if it maximizes } \quad 
    \dim_\K U- \dim_\K \sum_{A\in \calA}AU.
    \]    
    Observe that being a shrunk subspace of $\calA$ is closed under intersections. We can thus define the \emph{minimal shrunk subspace} of $\calA$.
\end{definition}

For the purpose of this paper, we present a simple  random algorithm introduced in \cite[Theorem 6.7]{franks2022shrunk_arxiv} for $\K=\bbQ$. We extend this algorithm to any field using \cite[Lemma 5.3]{ivanyos2018constructive}. For completeness, a detailed proof of the following proposition is given in \autoref{subsec:shrunk}.

\begin{proposition}[{\cite{franks2022shrunk_arxiv, ivanyos2018constructive}}]\label{prop:shrunk}
    Fix a field $\K$ and $\eta>0$. Given an integer $N>0$ and a subspace $\calA\subset \K^{N\times N}$ generated by a family of $\ell$ matrices, there exists a random algorithm to compute the minimal shrunk subspace of $\calA$ using at most $C_{\K} \ell N^{8+\eta} $ arithmetic operations for some constant $C_{\K}>0$. 
\end{proposition}

Fix $V \in \Persfpb{\bbR^d}$ and a number $\varepsilon>0$. Let $\gri$ be the smallest grid containing $ \varepsilon\bbZ^n\cap (\supp(V)+[0,\varepsilon)^d)$.  Using the \textsc{Persistence-Algebra} library\footnote{github.com/JanJend/Persistence-Algebra/} one can compute the structure maps of $i_\gri^*\module$ in time $n^3 (1/ \varepsilon)^d + n^2(1/ \varepsilon)^{2d}$. 

Let $\alpha\in \gri$, and let $\set{\alpha,\beta_1,\dots,\beta_k}\coloneqq\gri\cap\langle \alpha \rangle$. Define $ p_0\coloneqq \dim_\K V_\alpha $ and $ q_0\coloneqq  \sum_{i=1}^k \dim_\K V_{\beta_i}$. Consider the following subspace of $\K^{q_0\times p_0}$: 

\begin{equation}
\label{eqn:calA}\calA_\alpha\coloneqq \set{
\left(
\begin{array}{c}
\lambda_{1} V_{\alpha \to \beta_1} \\
\vdots \\
\lambda_{k} V_{\alpha \to \beta_k}
\end{array}
\right) 
{\Big|} \  \lambda\in \K^k}
\end{equation}

The minimal shrunk subspace of $\calA_\alpha$ is related to submodules of $V$ via the following lemma: 

\begin{lemma}[{\cite[Theorem 3.3]{huszar2021non}, \cite[Lemma 3.1]{iwamasa2024algorithmic}}] \label{lmm:shrunk}
Given $p,q>0$, the minimal shrunk subspace of $\K^{p\times q}\otimes\calA_\alpha $ is $ \K^q\otimes \module'_\alpha$ where $\module'\subset \module$ is the minimal submodule which maximizes :
\begin{equation}\label{eqn:g_pq}
    g_{p,q}\colon \module'\mapsto q \dim \module'_\alpha  - p\sum_i \dim \module'_{\beta_i}.\qedhere
\end{equation}   
\end{lemma}

\begin{proposition}[{\cite[Theorem 4.3]{iwamasa2025arxiv}}]
\label{prop:shrunk-discrepency}
Let $p,q>0$ and let $0= \langle F^0\rangle \hookrightarrow \langle F^1 \rangle \hookrightarrow\dots \hookrightarrow \langle F^{\ell}\rangle = \langle i_\gri V_\alpha\rangle$ be the HN filtration of $\langle i_\gri^*\module_\alpha\rangle$ at $\alpha$ -- see \autoref{subsec:HN-formal}. Let $ i(p,q)$ be the maximum integer $0\leq i\leq \ell$ \\ 
\[
\text{such that either $i=0$} \quad   \text{or} \quad \mu(\langle F^i\rangle /\langle F^{i-1}\rangle ) >\frac{p}{p+q}.
\]
The minimal shrunk subspace of $ \K^{p\times q}\otimes\calA_\alpha$ is $  \K^q\otimes F^{i(p,q)}$.
\end{proposition}
\begin{proof} 
    This is a direct application of {\cite[Theorem 4.3]{iwamasa2025arxiv}}, where, in the notation of \cite{iwamasa2025arxiv}, $\lambda\coloneqq \frac{p}{p+q}$ and $f_\lambda \coloneqq -\frac{g_{p,q}}{p+q}$ (see (\autoref{eqn:g_pq})). The minimal maximizer of $g_{p,q}$ coincides with $F^{i(p,q)}$. By \autoref{lmm:shrunk}, $ \K^q\otimes F^{i(p,q)}$ is the minimal shrunk subspace of $\K^{p\times q}\otimes\calA_\alpha$.
\end{proof}

Cheng also observed in the above Theorem that when $(p,q)\coloneqq (p_0,q_0)$, either $\ell=1$ or $0<i (p_0,q_0) <\ell$ \cite[Corollary 3.4]{cheng2024deterministic}. They thus devised \autoref{alg:compl-naive}.

\begin{algorithm}
\texttt{Input}:$\alpha\in \gri$ and $\module\in \Pers\gri$ uniquely generated at $\alpha$.\\
\texttt{Output}: The list of modules in the HN filtration of $\module$ along $\alpha$\\
Let $U\subset V_\alpha$ such that $\K^{q_0}\otimes U$ is the minimal shrunk subspace of $\K^{p_0\times q_0}\otimes\calA_\alpha $ computed by \autoref{prop:shrunk} \label{line:shrunk_alg_cheng}\;
\If{$U\in\set{0,V_\alpha}$}{
\Return $\set{0,\module}$\;
}
\Return {$\text{HN-Cheng}(\gri,\alpha,\langle U\rangle)\cup (U+\text{HN-Cheng}(\gri,\alpha,\langle V_\alpha/U\rangle))$}\;

\caption{HN-Cheng}
\label{alg:compl-naive}
\end{algorithm}

\begin{corollary}
    \label{cor:cheng_naive_analysis}
   Given $\module\in\Persfpb{\bbR^d}$ and $\varepsilon>0$, there is a random algorithm to compute an $\varepsilon$-approximation of $s_V$ in time $\mathcal{O}({\thick{\module}^{19}}{\varepsilon^{-11d}})$ up to poly-logarithmic factors.
\end{corollary}
\begin{proof} \underline{correctness}: 
By applying \autoref{alg:compl-naive} to every $\alpha\in\gri$, we obtain a representation (see \autoref{subsec:storing-skyscraper}) of the functor
 \[
 \widehat s^\bullet(\bullet,\bullet)\colon \begin{cases}
\bbR\times (\bbR^n)^\op\times\bbR^n &\to (\bbZ\cup\set{+\infty})^\op\\
     (\theta,x,y)&\mapsto s_{i_\gri^*V}^\theta(\lfloor\alpha\rfloor_\gri,\lfloor\beta\rfloor_\gri). 
 \end{cases}
 \]
We now prove that $ \widehat s^\bullet(\bullet,\bullet)$ is an $2\varepsilon$-approximation of $s_V$.  Let $V'\coloneqq (i_\gri)_!i_\gri^*V$. By \autoref{prop:ext-cont}, $s_{i_\gri^*V}^\theta$ is the restriction of the Skyscraper Invariant of $V'$ to $\gri\times \gri$. Since $V$ and $V'$ are $\varepsilon$-interleaved, we have by { \cite[Theorem 3.7]{fersztand2024harder}}, $d_\text{erosion}(s^\theta_V, s^\theta_{V'})\leq \varepsilon$. Moreover, since $\beta \leq \lfloor\beta+\vec\varepsilon\rfloor_\gri$ and $\alpha-\vec\varepsilon\leq \lfloor \alpha\rfloor_\gri$
   \begin{align*}
        \widehat s_V^\theta(\alpha-\vec\varepsilon,\beta+\vec\varepsilon) & = s^\theta_{V'}(\lfloor \alpha-\vec\varepsilon\rfloor_\gri,\lfloor\beta +\vec\varepsilon\rfloor_\gri) \leq s_{V'}^\theta(\alpha,\beta)\\
        s_{V'}^\theta(\alpha-\vec\varepsilon,\beta+\vec\varepsilon) &\leq s_{V'}^\theta(\lfloor\alpha\rfloor_\gri,\lfloor \beta\rfloor_\gri) = \widehat s_V^\theta(\alpha,\beta).
    \end{align*}
    Hence, $d_\text{erosion}(\widehat s^\theta_V,s^\theta_V)\leq d_\text{erosion}(\widehat s^\theta_V,s_{V'}^\theta)+ d_\text{erosion}(s_{V'}^\theta,s^\theta_V)\leq 2\varepsilon$.
 \\ \hspace*{2.75em} \underline{complexity}: By \autoref{prop:shrunk} with $\ell \coloneqq k p_0q_0$ and $N\coloneqq p_0q_0$, \autoref{line:shrunk_alg_cheng} in \autoref{alg:compl-naive} takes $\Oc(\lvert \gri\rvert (p_0q_0)^{9+\eta/2})=\mathcal{O}(\thick{\module}^{18+\eta}{\varepsilon^{-10d-\eta}})$ for any $\eta>0$. The recursive step adds a factor $\thick{\module}$, whence the desired complexity.
\end{proof}

\subsection{Computing shrunk subspaces}\label{subsec:shrunk}

Fix two integers $N,N'>0$. For $1\leq i\leq N$ and $1\leq j\leq N'$, we denote by $E_{ij}^{NN'}\in\K^{N\times N'}$ the matrix with a 1 at entry $(i,j)$ and 0 everywhere else. 

Given a subspace $U\subset \K^{N'}$ and a matrix $A\in\K^{N\times N'}$, let $AU$ and $A^{-1}U$ denote respectively the image and the inverse image of $U$ by the linear map defined by $A$ in the canonical bases. Let $\calA$ be a sub (vector) space of $\K^{N\times N'}$, we denote by $\calA U$ the subspace $\sum_{A\in \calA} AU$.  The \emph{rank} of $\calA$ is the integer
\[
\rank\calA \coloneqq \max_{A\in\calA}\rank A
\]

\paragraph{Wong sequences}  We now describe a key tool to compute the minimal shrunk subspaces introduced in \autoref{pb:c-shrunk}. 

\begin{definition}
    [Wong sequence]\label{def:wong}
    Given a subspace $\calA\subset \K^{N\times N'}$ and a matrix $A\in \calA$, the \emph{Wong sequence} of $(A,\calA)$ is defined by $W^0=0$ and for $i>0$ by $W^{i} = \calA(A^{-1}(W^{i-1}))$.
\end{definition}

Given a Wong sequence $(W^i)_{i\geq 0}$ of $\calA$, if for $i>0 $, we have $W^{i-1}\subset W^{i}$ (\textit{resp}. $W^{i-1}=W^i$) then $W^{i}\subset W^{i+1}$ (\textit{resp}. $W^{i}=W^{i+1}$). As a consequence, $(W^{i})_i$ stabilizes after at most $\min(N,N')$ steps to the subspace of $\K^N$
\begin{equation}\label{eqn:wong-limit}
    \Wong(A,\calA)\coloneqq \bigcup_{i\geq 0} W^i.
\end{equation}

Wong sequences are computable in polynomial-time and their limits are related to shrunk subspaces by the following lemma

\begin{lemma}
    \label{lmm:wong-shrunk}[{\cite[Lemma 9]{ivanyos2015generalized}, \cite[Lemma 6.3]{franks2023shrunk}}]\label{lmm:wong_cshrunk} Given a subspace $\calA\subset \K^{N\times N'}$ and a matrix $A\in\calA$. If $U^*$ denotes the minimal shrunk subspace of $\calA$, then the following two propositions are equivalent 
    \begin{enumerate}[(i)]
        \item $\Wong(A,\calA)\subset \Img A$
        \item $N'-\rank A=\dim U^*-\dim \calA U^*$
    \end{enumerate}
   and if they are satisfied, then  $A^{-1}(\Wong(A,\calA)) = U^*$.
\end{lemma}

Note that by taking the maximum over $U\subset \K^{N'}$ and $A\in \calA$ of the inequality $\dim U-\dim \calA U\leq \dim U -\dim AU\leq \dim \ker A=N'-\rank A$, one derives that
\begin{equation}
    \label{eqn:rk-and-ncrank}N'-\rank \calA\geq\dim U^*-\dim \calA U^*.
\end{equation}
When (\autoref{eqn:rk-and-ncrank}) is an equality, finding a matrix of maximum rank in $\calA$ is enough to obtain from
 \autoref{lmm:wong-shrunk} a deterministic algorithm for the shrunk problem (\autoref{pb:c-shrunk}). The key idea introduced in \cite{ivanyos2015generalized} is to transform our matrix space $\calA$ until (\autoref{eqn:rk-and-ncrank}) becomes an equality. Two operations are considered: blow-ups and field extensions.

\begin{definition}
    [blow-up]\label{def:blowup}
    Given $p,q>0$, the $(p,q)$-blow-up of a matrix subspace $\calA\subset \K^{N\times N'}$ is the subspace $\calA^\set{p,q} \coloneqq \K^{p\times q}\otimes \calA $ of $\K^{pN\times qN'}$. If $p=q$, we denote by $\calA^{\set{p}}$, the blow-up $\calA^{\set{p,p}}$. 
\end{definition}

The following lemma guarantees that the minimal shrunk subspace of $\calA$ can be recovered from the minimal shrunk subspace of a blow-up of $\calA$. 

\begin{lemma}
    [{\cite[Proposition 5.2]{ivanyos2017non}, \cite[Lemma 6.6]{franks2023shrunk}}]\label{lmm:blowup_cshrunk} 
Given a subspace $\calA\subset \K^{N\times N'}$ and an integer $p>0$, if ${U^*}$ denotes the minimal shrunk subspace of $\calA$, the minimal shrunk subspace of $\calA^{\set{p}}$ is $\widetilde U\coloneqq \K^p\otimes U^*$.  
\end{lemma}

The following lemma allows us to compute the minimal shrunk subspace in an extension of $\K$.

\begin{lemma}[{\cite[Lemma 5.3]{ivanyos2018constructive}}]\label{lmm:shrunk-change-field}
Given matrices $A_1,\dots,A_\ell\in \K^{N\times N'}$ and a field extension $\mathbb L/\K$, let $U^*_{\mathbb L}$ and $U^*_{\K}$ be the minimal shrunk subspace of respectively $\calA_{\mathbb L} \coloneqq\sum_iA_i\mathbb L$ and $\calA_\K\coloneqq\sum_iA_i\K$. We have $\dim U^*_{\mathbb L} - \dim \calA_{\mathbb L} U^*_{\mathbb L} =\dim U^*_{\K} - \dim \calA_\K U^*_{\K}  $.
\end{lemma}
Assume that $\mathbb L$ is an extension of $\K$ of degree $g$ and that $\phi$ is a fixed embedding of $\mathbb L$ into $\K^{g\times g}$. Given a matrix $X\in {\mathbb L}^{N\times N'}$, we denote by $\phi(X)$ the matrix in $\K^{gd\times gN'}$ obtained by replacing each coefficient $x$ of $X$ by the block matrix $\phi(x)$. A consequence of the above lemma is that for $p\geq 1$ and $X_1,\dots,X_\ell\in \K^{p\times p}$, if the pair $(\sum_iX_i\otimes A_i, \calA_{\mathbb L}^{\set p})$ satisfies \autoref{lmm:wong-shrunk}(ii), then $(\sum_i \phi(X_i)\otimes A_i, \calA_{\K}^{\set {gp}})$ also satisfies \autoref{lmm:wong-shrunk}(ii).

\paragraph{Deterministic and random algorithms} Fix a subspace $\calA\subset \K^{N\times N'}$. Lemmas \ref{lmm:blowup_cshrunk} and \ref{lmm:shrunk-change-field} show that we are allowed to blow-up $\calA$ and to take extensions of $\K$. The following theorem controls the size of the blow-up and the degree of the field extension required to guarantee the existence of a matrix satisfying  \autoref{lmm:wong-shrunk}(ii).

\begin{theorem}[{\cite[Theorem 1.5]{ivanyos2018constructive}}] \label{thm:blow-up-ncrank} Given a field $\K$ and a subspace $\calA\subset \K^{N\times N'}$, assume that $\lvert \K\rvert \geq \min(N,N')^{\Omega(1)}$. Let $ p\geq \min(N,N')-1$ and let $\widetilde U$ denote the minimal shrunk subspace of $\calA^\set{p}$, then 
\[
pN'-\rank \calA^{\set p} = \dim \widetilde U- \dim \calA^\set{p} \widetilde U.\qedhere
\]
\end{theorem}

In \cite{ivanyos2017non}, \autoref{thm:blow-up-ncrank} is constructive, in the sense that it provides a deterministic algorithm to compute in polynomial time a matrix of maximum rank in $\calA^{\set p}$. Using  Lemmas \ref{lmm:blowup_cshrunk} and \ref{lmm:shrunk-change-field}, they deduce the following:

\begin{theorem}[{\cite[Theorem 1.5, Corollary 1.7]{ivanyos2018constructive}}]\label{thm:shrunk-poly}
       Given a fixed field $\K$, there exists a deterministic polynomial-time algorithm to solve the shrunk subspace problem over $\K$ (this algorithm is not required to be  polynomial in $\lvert \K\rvert$).
\end{theorem}

The worst-case time complexity bounds of deterministic algorithms for the shrunk subspace problem are polynomials of very high degree. Over $\K=\bbQ$ and with $N=N'$, the authors of \cite{franks2023shrunk} found that the best complexity bound available in the literature was $\Oc( N^{12}(N+\dim\calA))$. As a consequence, they proposed a faster random algorithm. This algorithm relies on the fact that in a large enough blow-up and over a large enough field, a random matrix is of maximal rank with high probability. \autoref{algo:compute-cshrunk} is the adaptation of \cite[Algorithm 5]{franks2022shrunk_arxiv} to the case of finite fields.

\begin{algorithm}
\texttt{Input}: {a finite field $\K$, integers $p,N,N'>0$, and a basis $(A_1,\dots,A_\ell)$ of a subspace $\calA\subset \K^{N\times N'}$}\\
\texttt{Output}: {the minimal shrunk subspace $U^*$ of the matrix space $\calA:=\sum_i A_i\K$}
\; Take a field extension $\mathbb L/\K$ of degree $g \coloneqq \max(1,\lceil \log^2_{\lvert \K\rvert} p\rceil)$ and an embedding $\phi\colon\mathbb L\hookrightarrow \K^{g\times g}$
\; Draw uniformly and independently $X_1,\dots,X_\ell$ in $ {\mathbb L}^{p\times p}$
\; 
Let $A\coloneqq \sum_i \phi(X_i)\otimes A_i$ ; \hspace{5em}{\color{gray}
 \textbackslash\textbackslash \   $\phi(X_i)$ is obtained by replacing  each entry $x$
   of $X_i$ by the $g\times g$ block $\phi(x)$
} \\  
$W\coloneqq \Wong(A,\mathcal{A}^{\set{pg}})$\; \If{ $W\subset \Img A$}{ \texttt{Return} $ \pi_{\K^{N'gp}\twoheadrightarrow \K^{N'}}(A^{-1}W)$ ;
\hfill {\color{gray} \textbackslash\textbackslash projection}\\
} 

\caption{Minimal shrunk subspace for finite fields}
\label{algo:compute-cshrunk}
\end{algorithm}

\begin{proposition}[Restatement of \autoref{prop:shrunk} for finite fields]\label{prop:random-alg-shrunk} 
    Having fixed a finite field $\K$, \autoref{algo:compute-cshrunk} runs in time
    \[\Oc((pg+\ell) (pg)^3(N+N')^4)\]
    where $g$ is the poly-log factor $g\coloneqq 1+ \lceil \log^2p\rceil$.
    If \autoref{algo:compute-cshrunk} returns an output, this output is the minimal shrunk subspace $U^*$ of $\calA$. For $p\geq \Oc(\min(N,N'))$, it returns $U^*$  with probability $1-(\frac1p)^{\Omega(1)}$.     
\end{proposition}
\begin{proof}\ 
 \underline{correctness}: the return statement is reached only when $W=\Wong(A,\calA^{\set {gp}})$ and $W\subset \Img A$. In that case, by Lemma \ref{lmm:wong_cshrunk}, $A^{-1}(W)$ is the minimal shrunk subspace of $\calA^\set{ gp}$, and by \autoref{lmm:blowup_cshrunk} the algorithm is correct. 

 \qquad   \underline{complexity}: The time-complexity of this algorithm is dominated by the computation of the Wong sequence. For simplicity, we restrict ourselves to the case $N=N'$ -- see \autoref{rem:shrunk-cite}.
 
The sequence $\Wong(A,\calA^{\set{pg}})$ stabilizes after at most $pg N$ steps. Step $1\leq s\leq pgN$ involves an inverse image computation of $U^s\coloneqq A^{-1}(W^{s-1})$ in time $\Oc ((pgN)^3)$, and a direct image computation $W^s\coloneqq \calA^{\set{pg}} U^{s}$, which we now analyze. A generating family of $W^s$ is given by 
 \[
 \set{(E^{pg,pg}_{ij}\otimes A_k) U^s\mid i,j\in [pg], k\in [\ell]}.
 \]
 The concatenation of these vectors yields a matrix $M$ composed of $pg\times pg\ell $ blocks of size $N\times \dim U^s$. Each block-row $M_1, \dots M_{pg}$ of $M$ is identical and is composed of the following  blocks $\set{A_k U^s_i\mid k\in [\ell], i\in [pg]}\cup\set {0}$, where $U^s_1\dots U^s_{pg}$ are the blocks-rows of size $N\times \dim U^s$ of $U^s$. It follows that one can compute $W^s$ by reducing the matrix $M_1$ of size $N\times (pg\ell\dim U^s)$ whose blocks can be computed in time $\Oc(N^2\ell \dim U^s )$ and whose reduction can be obtained in time $\Oc(N^2 pg\ell\dim U^s) = \Oc(\ell N(Npg)^2)$. The desired complexity is $\Oc(Npg((Npg)^3 + N^3 (pg)^2\ell )) = \Oc((\ell+pg) (pg)^3N^4)$.

   \qquad \underline{probability of success}: we denote by $r:= N'-\dim U^*+\dim \calA U^*$. Let $p\geq \Oc(\min(N,N'))$, by Lemma \ref{lmm:wong_cshrunk} and \ref{lmm:blowup_cshrunk}, we only need to show that $\rank A=gpr$ with high probability. By \autoref{thm:blow-up-ncrank} and \autoref{lmm:shrunk-change-field}, since the extension $\mathbb L$ has size $\lvert \K\rvert^{g}\geq p^{\Omega(1)}=\min(N,N')^{\Omega(1)}$, there exists $X_1,\dots,X_\ell\in {\mathbb L}^{p\times p}$ such that $\rank \sum_iA_i\otimes X_i =pr$.  More precisely, if we see each entry in $X_1,\dots,X_\ell$ as an indeterminate, there exists a nonzero $pr\times pr$ minor $P$ in $\sum_iA_i\otimes X_i$. By Schwartz-Zippel lemma, the minor does not vanish with probability at least $1-\frac{\deg\ P}{\lvert \mathbb L\rvert}=1-(\frac1p)^{\Omega(1)}$. Observe that $\rank A= g\rank\sum_iA_i\otimes X_i$ so, with probability $1-(\frac1p)^{\Omega(1)}$, we have as desired $\rank A=pgr$. 
\end{proof}

\begin{remark}\label{rem:shrunk-cite}
    When $\K$ is infinite, the random matrices $X_1,\dots ,X_\ell$ in \autoref{algo:compute-cshrunk} can be drawn uniformly from a finite subset of $\K$ of size $p^{\Omega(1)}$. By the same argument as  \autoref{prop:random-alg-shrunk}, this modified version of  \autoref{algo:compute-cshrunk} returns $U^*$ with high probability in time $\Oc((p+\ell) p^3(N+N')^4)$ for any $p\geq \min(N,N')-1$. This argument is proved in detail for $\K=\bbQ$ and $N=N'$ as Theorem 6.7 in \cite{franks2022shrunk_arxiv}. In their notation, $p$, $d$ and $n$ correspond to our $\ell$, $p$ and $(N+N')$. Note that their results for square matrices can be easily adapted to non-square matrices. Indeed, the matrices of the space $\calA\subset\K^{N\times N'}$ can be padded with zeroes to obtain a subspace of $\K^{\max(N,N')\times \max(N,N')}$. One can check that the minimal shrunk subspace $U^*$  of $\calA$ is unchanged if $N\leq N'$ and becomes $U^* \oplus \K^{N-N'}$ if $N>N'$. 
\end{remark}

\subsection{Optimizing Cheng's algorithm}

\paragraph{Optimized Wong sequence computations}
By harnessing the block structure of (\autoref{eqn:calA}), we are able to improve the complexity in \autoref{prop:shrunk} for the subspaces appearing in \autoref{alg:compl-naive}. The rest of this paragraph is devoted to the proof of the following proposition

\begin{proposition}
    \label{prop:cheng_advanced_analysis}
     Given $\module\in\Persfpb{\bbR^d}$ and $\varepsilon>0$, there is a random algorithm to compute an $\varepsilon$-approximation of $s_V$ in time $\mathcal{O}({\thick{\module}^{17}}{\varepsilon^{-9d}})$ up to poly-logarithmic factors.
\end{proposition}

 The following lemma is a generalization of the complexity analysis of \autoref{prop:random-alg-shrunk}. 

\begin{lemma}
    \label{lmm:wong-compute-blow-up}
  Let $\calA_1,\dots,\calA_t$ be matrix subspaces of respectively $\K^{N_1\times N'}$, \dots, $\K^{N_t\times N'}$. Let $p,q>0$ be integers, and let $\phi\colon \set{1,\dots,p}\to\set{1,\dots,t}$ be a function.  Consider the matrix space $\mathcal B\subset \K^{\sum_{i=1}^pN_{\phi(i)}\times qN'}$ generated by
    \[
    \left(E_{ij}^{pq}\otimes \calA_{\phi(i)}\right)_{1\leq i\leq p , 1\leq j\leq q}.
    \]
    Given $A\in \mathcal B$ and bases for $\calA_1,\dots,\calA_t$, the Wong sequence of $(A,\mathcal{ B}) $ can be computed in $\mathcal{O}((p+q+\ell)(p+q)^3NN'(N+N')^2)$ where $N\coloneqq \max_{1\leq i\leq t} N_i$ and $\ell\coloneq\sum_{i=1}^t \dim\calA_i$.
\end{lemma}
\begin{proof}
    The Wong sequence $\Wong(A,\mathcal{B})$ converges after at most $\min(Np,N'q)=\mathcal{O}((p+q)\min(N,N'))$ steps. At step $s\geq 0$, we reduce a matrix of size at most $Np\times (Np+N'q)$ to compute $A^{-1}(W_{s-1})$ at a cost of $\Oc((p+q)^3(N+N')^3)$. 

    Let $U_s$ be the subspace $A^{-1}(W_{s-1}) \subset \K^{N'q}$, we show that the image computation $W_s\coloneqq \mathcal BU_s$ can be done using only $\ell q$ block matrix multiplications. The key observation is that for  $1\leq i\leq p$, and $1\leq j\leq q$, the product $(E^{pq}_{ij}\otimes \calA_{\phi(i)})U_s$ only depends on $(\phi(i),j)$ and can be obtained from the multiplication of a basis of  $\calA_{\phi(i)}$ with the corresponding row block $U_{s,j}$ of $U_s$. As a consequence, the computation of a  generating family of $\mathcal BU_s$ can be done using $\ell q$ block multiplication costing $\Oc(NN'^2q)$ each. 
    
    We have obtained a generating family of $\mathcal BU_s$ represented by a diagonal block matrix composed of $p\times p$ blocks. For $k\in\set{1,\dots,t}$, the diagonal blocks indexed by $i\in \phi^{-1}(k)$ are the generating family $M_i$ we computed for $\calA_{k} (U_{s,1} \dots U_{s,q})$. Note that $M_i$ has $1\times q\dim\calA_k$ blocks of size $N_k\times N'q$. Its reduction thus costs $\Oc(N^2q^2 N'\dim \calA_k)$. The total cost to extract a basis for $\mathcal BU_s$ is $\mathcal{O}(N^2N'q^2\ell)$.  Both the inverse image and the image computation at step $s$ is done in $\Oc((p+q+\ell) (p+q)^2(N+N')^3)$, whence the desired complexity.
\end{proof}

\begin{proof}[Proof of \autoref{prop:cheng_advanced_analysis}]  We update the complexity analysis of \autoref{cor:cheng_naive_analysis} and \autoref{prop:random-alg-shrunk} with our improved complexity for Wong computations. Remember the space $\calA_\alpha\subset\K^{q_0\times p_0}$ and the degrees $\set{\beta_1,\dots,\beta_k}$ from (\autoref{eqn:calA}). By \autoref{prop:random-alg-shrunk},
\[
\autoref{algo:compute-cshrunk}(\K,r\geq \Oc(p_0q_0),N=N'\coloneqq p_0q_0,\calA\coloneqq \calA_\alpha^\set{p_0,q_0})
\]
successfully returns the minimal shrunk subspace of $\calA_\alpha^\set{p_0,q_0}$  with probability  $1- (\frac1{r})^{\Omega(1)}$. The complexity $C(\thick V,\varepsilon)$ of \autoref{line:shrunk_alg_cheng} in \autoref{alg:compl-naive} is given by applying \autoref{lmm:wong-compute-blow-up}  with $\ell\coloneqq k$, $(p,q)\coloneqq r\lceil\log_{\lvert \K\rvert}^2r+1\rceil (p_0,q_0)$ and $N,N'\leq \thick{\module}$. Let $\eta>0$ and note that in this case, we have $\ell\leq q = \Oc(p+q)$ so the Wong sequence complexity is dominated by the inverse image computation. We have
\begin{align*}
C(\thick \module,\varepsilon) &=\Oc((r^{1+\eta/8}(p_0+q_0) \thick \module)^4  ) \\
&= \mathcal{O}(  
(
 \thick \module^{4+\eta/4} k^{2+\eta/8}
)^4 )\\
&= \Oc( \thick{\module}^{16+\eta} k^{8+\eta} )
\\&= \Oc( \thick{\module}^{16+\eta} \varepsilon^{-8d-\eta}).
\end{align*}
The total complexity of \autoref{alg:compl-naive} is $\Oc(\lvert \gri\rvert \thick V)C(\thick V,\varepsilon) =  \Oc( \thick{\module}^{17+\eta} \varepsilon^{-9d-\eta}) $, as desired. 
\end{proof}

\paragraph{Empirical improvements} Further empirical improvements can be obtained by computing $F^{i(p,q)}$ in \autoref{prop:shrunk-discrepency} for small integers $p,q>0$ such that $\frac{p}{q}\approx \frac{p_0}{q_0}$. Indeed, the complexity of computing $F^{i(p,q)}$ is heavily dependent on the size of $pq$. For simplicity, assume that $\frac pq>\frac{p_0}{q_0}$. If $F^{i(p,q)}\neq 0$, then we can replace $U$ by $F^{i(p,q)}$ in \autoref{alg:compl-naive}. Otherwise, we know that there are no subspaces $U'\subset \module_\alpha$ such that $\mu(\langle U'\rangle)>\frac{p}{p+q}$. We keep choosing better approximations $\frac pq$ of $\frac{p_0}{q_0}$ (so with higher values of $p,q$) until either we find a proper submodule in the HN filtration of $\module$ at $\alpha$, or we can guarantee that $\module$ is semistable at $\alpha$.

Finally, given $p,q>0$ and a matrix $A\in \K^{p\times q}\otimes\calA_\alpha$, our implementation uses the block structure of $A$ to partially reduce it. More precisely, since there exists $\Lambda\in \K^{pk\times q}$ such that $A = (\lambda_{ij}V_{\alpha \to \beta_{i\text{ modulo }k}})_{ij} $, one can column-reduce the first $q$ lines of blocks by reducing $\Lambda$ and each $( V_{\alpha \to \beta_i})_{1\leq i\leq k}$. 

\autoref{alg:compl-naive} can be used as the routine $\text{ HN-Filtration}_\alpha  \text{-Factors}$ in \autoref{alg:skeleton}. However, \autoref{sec:experiments} will show that even with our optimisations, Cheng's algorithm remains too slow for practical sizes of data. In the next sections, we will present, in the case of finite fields, algorithms that are more efficient than \autoref{alg:compl-naive}.

\section{Finding Highest-Slope Submodule}\label{sec:hn_exhaustive}

We will explain how to compute the Harder-Narasimhan filtration of $\langle V_\alpha \rangle$ with an exhaustive search. To do so, we must compute slopes. The slope $\mu_\alpha(\module)$ at $\alpha$ is entirely determined by the Hilbert function of $\langle V_\alpha \rangle$ and thus by its Betti numbers.

\begin{proposition}\label{prop:compute_slope}
Let $\module\in\Persfpb{\bbR^d}$, then 

\begin{equation}\label{eq:int_dim} \int\limits_{\bbR^d} \udim \module  
=  \sum\limits_{i = 0}^{d} (-1)^{d+i} \sum\limits_{\gamma \in b_i(\module)} \prod\limits_{j = 1}^{d} \gamma_j. \end{equation}
\end{proposition}

\begin{proof}
Consider a minimal free resolution $0 \to F_d \to \dots \to F_0 \to \module \to 0$ which will be of maximal length $d$ by Hilbert's Syzygy theorem. In our bounded case, it will even be exactly of length $d$ by duality of free and injective resolutions (first proved by Miller in \cite{Miller2001} and rediscovered in \cite{BLL23} for the computation of Multiparameter Persistent Cohomology). Each $F_i$ has a basis with degrees at $b_i(\module)$. By exactness, for each $\alpha \in \bbR^d$ we have

\begin{equation}\label{eq:dim_betti_sum} \dim_\K V_\alpha = \sum\limits_{i=0}^d (-1)^i \dim_\K \left(F_i\right)_\alpha = \sum\limits_{i=0}^d (-1)^i \sum_{\gamma \in b_i(\module)} \mathds{1}_{\{\gamma \leq \alpha \}}. \end{equation}

Let $B \in \bbR^d$ be large enough that $\supp(\module) \subset\prod_{j \in [d]}(-\infty,B_j]$, then

\begin{equation}\label{eq:int_poly} \int\limits_{\bbR^d} \udim \module  
=  \int\limits_{\prod_{j \in [d]}(-\infty,B_j]} \udim \module  
= \sum\limits_{i = 0}^{d} (-1)^{i} \sum\limits_{\gamma \in b_i(\module)} \prod\limits_{j = 1}^{d} \left( B_j - \gamma_j \right). \end{equation}

As a function of $B$, (\autoref{eq:int_poly}) is a polynomial and constant, so we can set $B=0$ which results in (\autoref{eq:int_dim}).
\end{proof}

  \paragraph{Iteration over subspaces.} Instead of considering $\langle V_\alpha\rangle$, we assume without loss of generality, that $\module$ is uniquely generated at $0$. Recall that $\thick{\module}= \max \udim \module$, and choose any basis $\K^{\thick{\module}} \iso V_0$. We want to find the sub vector space $U \subset \K^{\thick{\module}}$ for which $\langle U \rangle \subset \module$ has the highest slope. An exhaustive search would involve computing $\mu(\langle U \rangle)$ with \autoref{alg:submodule} for \emph{all} sub vector spaces $U \subset V_0$. This is clearly possible if we work over a finite field. Since all software for constructing multiparameter Persistent Homology groups so far uses finite fields, this is not a restriction in practice. Still, this procedure works also in the infinite case, because the finite presentation of $\module$ allows only a finite number of submodules with non-isomorphic Hilbert function. We invite the reader to work this out.
  Fix a prime power $q$ and set $\K = \bbF_q$.
Denote by $\mathcal{P}V_0$ the set of 1-dimensional subspaces of $V_0$ and define 
$ W_{\text{max}} \in \argmax\limits_{ W \in \mathcal{P}V_0} \mu( \langle W \rangle ) $

\begin{proposition}\label{prop:criterion}
Let $1<k \leq \thick{\module}$ and define $V_k \coloneqq \{ W \in \mathcal{P}V_0 \mid \mu( \langle W \rangle ) \geq \mu( \langle W_{\text{max}} \rangle )/k \}$.
If for every $k$-subspace $U \subset V_0$, it holds that $\mathcal{P}U \not\subset V_k$, then $\langle W_{\text{max}}\rangle \subset \module$ has a higher slope than any other submodule of thickness $k$. 
\end{proposition}

\begin{proof}
Consider any $U \subset V_0$ of dimension $k$. By assumption, there is a $1$-dimensional subspace $T \in \mathcal{P} U \subset \mathcal{P} V_0$ which is not in $V_k$. Therefore, $\mu( \langle T \rangle ) < \mu( \langle W_{\text{max}} \rangle )/k$. 
\[ \mu(\langle U\rangle ) = 
\bigslant{k}{\int_{\bbR^d} \udim_{\K} \langle U \rangle}
< \bigslant{k}{\int_{\bbR^d} \udim_{\K} \langle T \rangle} = k \mu(\langle T \rangle) < \mu( \langle W_{\text{max}} \rangle).\qedhere\]
\end{proof}

Every $k$-dimensional vector space over $\bbF_q$ contains exactly $(q^k-1)/(q-1)$ subspaces of dimension $1$, so the premise of \autoref{prop:criterion} is trivially satisfied if $|V_k| < (q^k-1)/(q-1)$. $V_0$ has  $ \sim q^{\thick{\module}^2/4}$ subspaces, but only $\sim q^{\thick{\module}}$ of these are $1$-dimensional.

\begin{algorithm}[H]
\caption{Exhaustive search for the highest-slope submodule}
\label{alg:brute_force}
\DontPrintSemicolon
\KwIn{A graded matrix $M \in \bbF_q^{B_0 \times B_1}$ minimally presenting $\module \in \Persfpb{\bbR^d}$ generated at $0$.}
\KwOut{A subspace $S \subset V_0$ such that $\langle S \rangle $ is the submodule of maximal slope.}
$\texttt{slopes}\colon ( \{ \texttt{Subsets of } V_0\} \to \bbR) \gets  \{ \}$\;
\For{$W \in \mathcal{P}V_0$}{
    Pick any $v \neq 0 \in W$, $N_0 \gets  \autoref{alg:submodule}(M, v)$\;
    Compute a resolution $ N_i \in \K^{B_{i-1} \times B_{i}}$ (with LWKR or Schreyer's algorithm)\;
    $b_i(\module) \gets B_{i}$\;
    $\texttt{slopes}.\text{append}\left( v \mapsto \sum\limits_{i = 0}^{d} (-1)^{d+i} \sum\limits_{\gamma \in b_i(\module)} \prod\limits_{j = 1}^{d} \gamma_j \right)$ \;
}
$v_{\text{max}} \gets \argmax \texttt{slopes} $\;
$V_k \gets \{ w \in \mathcal{P}V_0 \mid \texttt{slopes}(w) \geq \texttt{slopes}(v_{\text{max}})/k \}$\;

\For{$1\leq k \leq \thick{\module}$}{
    \If{$|V_k| \geq (q^k-1)/(q-1)$}{
        \For{$W \in \textbf{Gr}_k V_0$}{
            Pick a basis $\mathcal{B} \subset W$, $N_0 \gets  \autoref{alg:submodule}(M, \mathcal{B})$\;
            Compute a resolution $ N_i \in \K^{B_{i-1} \times B_{i}}$ (with LWKR or Schreyer's algorithm)\;
            $b_i(\module) \gets B_{i}$\;
            $\texttt{slopes}.\text{append}\left( \mathcal{B} \mapsto \sum\limits_{i = 0}^{d} (-1)^{d+i} \sum\limits_{\gamma \in b_i(\module)} \prod\limits_{j = 1}^{d} \gamma_j \right)$ \;
        }
    }
}
\Return $\argmax \left( \texttt{slopes} \right)$\;
\end{algorithm}

\begin{proposition}\label{prop:brute_force_time}
Let $T_{\mathsc{ker}}(m , n, d)$ be the time to compute the kernel of an $m \times n$ $\bbR^d$-graded matrix. \autoref{alg:brute_force} returns the correct result in $\Oc(d \cdot T_{\mathsc{ker}}(|B_0| , |B_1|, d) \cdot q^{\thick{\module}^2/4 + \Oc(\thick{\module})})$ time.
\end{proposition}

\begin{proof}
Correctness follows from \autoref{prop:compute_slope} and \autoref{prop:criterion}. \autoref{alg:submodule} computes a kernel for each subspaces of $\bbF_q^{\thick{V}}$ and then we need another $d-1$ kernel computations to compute the whole resolution. There are $\sum_i \binom{\thick{V}}{i}_q \sim q^{(\thick{V}/2)^2 + \Oc(\thick{V})}$.
\end{proof}

\paragraph{Computing the HN filtration.}\label{par:quotient}
Let $\alpha \in \bbR^d$ again be arbitrary and $V$ uniquely generated at $\alpha$.
After \autoref{alg:brute_force} has found a highest slope submodule at $\alpha$ it returns a basis $\mathcal{B}$ of the sub vector space $W \subset V_\alpha$ which spans it. 
 We view $\mathcal{B}$ as a graded matrix or $A$-linear map $A^{\dim_\K W}[-\alpha] \to A^{\dim_\K V_\alpha}[-\alpha]$.
 If $M_\alpha \in \K^{G \times R}$ is a presentation of $\langle V_\alpha \rangle$ with exactly $\dim_\K V_\alpha $ generators, then we can compute a presentation of the quotient $ V_\alpha /\langle W \rangle$ with exactly $\dim_\K V_\alpha - \dim_\K W$ generators using the following procedure:

 \begin{itemize}
 \item Bring $\mathcal{B} \in \K^{\dim_\K V_\alpha \cdot \alpha \times \dim_\K W \cdot \alpha}$ in column echelon form
 \item Let  $P \subset [\dim_\K W]$ be the index sets of the pivots of $\mathcal{B}$
 \item Form the graded matrix $[ \mathcal{B} \ M_\alpha ] \in \K^{G \times ( \dim_K W \cdot \alpha \cup R)}$
 \item Column-reduce it using only operations from $\mathcal{B}$ to $M_\alpha$ until all entries of $M_\alpha$ with row index in $P$ are zero, producing $M_\alpha'$.
 \item Delete the rows of $M'_\alpha$ indexed by $P$ and return $M'_\alpha$.
 \end{itemize}
This is essentially \cite[Proposition 4.14]{djk_arxiv} and one step of the minimisation procedure. We then fully minimise $M'_\alpha$ by computing its kernel. This is necessary to be able to compute its Hilbert function correctly later in the algorithm.
By recursively calling \autoref{alg:brute_force} on $M_\alpha'$ and repeating the quotient calculation we can compute the whole HN filtration at $\alpha$. If we start with a module $V$, then this can repeat at most $\thick{\module}$ times, because the quotient will have dimension at least 1 less than $V$ at $\alpha$ - and in practice \emph{exactly} 1 most of the time (\autoref{obs:factor_dimension}).

\paragraph{Runtime of the assembled algorithm.}
Let $\module$ be a bounded persistence module, minimally presented by $M \in \K^{G \times R}$.
Denote by $n = |G| + |R|$ the input size.
Recall that the thickness $\thick{V}$ is the maximal pointwise dimension and denote by $\mathfrak{t}_l{(\module)}$ the maximal thickness of an indecomposable submodule of $\module$. Let $k \coloneqq \max_{\alpha \in \bbR^d} \mathfrak{t}_l{( \langle \module_\alpha \rangle )}$ be the maximal thickness of a uniquely generated, indecomposable submodule, let
 $l$ be the maximal number of relations of the same degree, and let $T_{\text{ker}}(n,k)$ the time to compute the kernel of a graded matrix of size $k \times n$, and $\omega$ be the matrix multiplication constant.

\begin{proposition}\label{prop:runtime_analysis}
Using \autoref{alg:brute_force} and \mathsc{aida},
 \autoref{alg:skeleton} computes an $\varepsilon$-approximation of the Skyscraper Invariant and its runtime is in
 \begin{align*} 
 \Oc \left( \frac{{\diam}(\supp V)}{\varepsilon^d}  \left( T_{\mathsc{ker}}(n,n,d) + T_{\mathsc{aida}}(n,l) 
 + \thick{V}  \cdot T_{\mathsc{ker}}(k, n, d) \cdot q^{k^2/4 + \Oc(k)} \right) \right)
 \end{align*}

 where $T_{\mathsc{aida}}(n,l) = n^{\omega+1} \mathfrak{t}{(\module)}^\omega + n^5 + n^{\omega+1}l^{\omega-1}(l + \thick{\module}) q^{l^2/4 + \Oc(l)} $.
\end{proposition}

\begin{proof}
We want to fill in the variables corresponding to the subroutines in \autoref{prop:skeleton_time}. From the preceding paragraph, it follows that \autoref{alg:skeleton} will call at most $\dim_\K V_\alpha$ times \autoref{prop:brute_force_time} for each $\alpha \in \bbR^d$. The run-time for the decomposition with \textsc{AIDA} \cite{djk_arxiv} follows from an unpublished result of the second author using fast homomorphism computation with the \textsc{Persistence-Algebra} library\footnote{https://github.com/JanJend/AIDA, https://github.com/JanJend/Persistence-Algebra}.
\end{proof}

\section{Exact Computation of the Skyscraper Invariant}\label{sec:wall_and_chamber}
For many computations on persistence modules it is often enough to perform them only at those parameter values that share 
each of their coordinates with a generator or relation (for example the barcode template for the fibred barcode \cite{lesnick_wright, lesnickwright25}).

\paragraph{The slope in relation to $\alpha$.}

 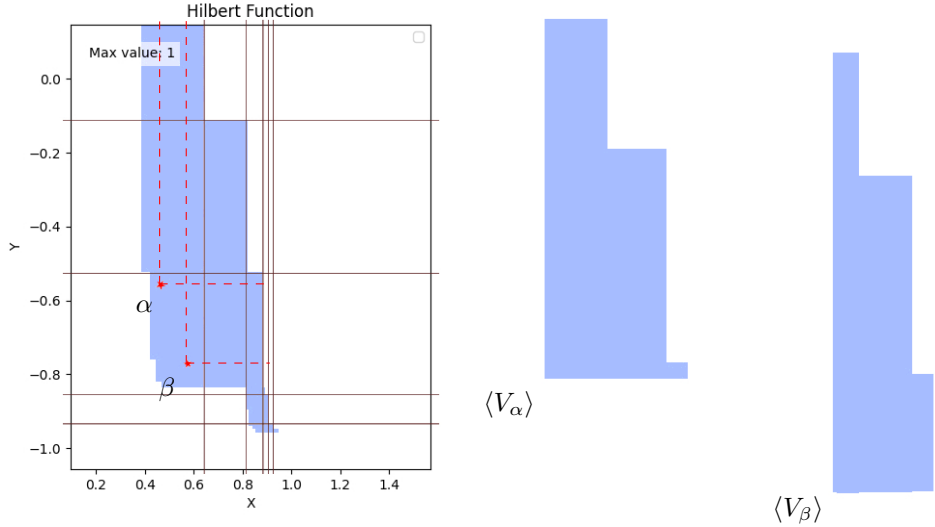
\begin{figure}[H]
    \includegraphics[width=0.88\textwidth]{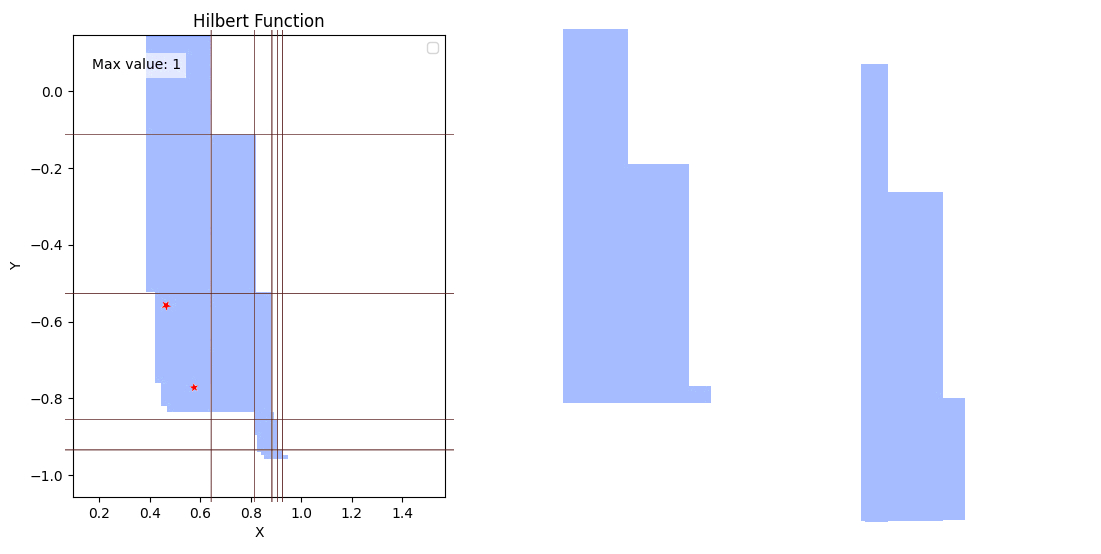}

\begin{tikzpicture}[overlay]
    \node at (1.9,3.3) {$\alpha$};
    \node at (2.2,2.2) {$\beta$};

    \node at (6.7 , 2.0) {$\langle V_\alpha \rangle$};
    \node at (10.5, 0.6) {$\langle V_\beta \rangle$};
    \draw[red, dashed] (2.1,3.6) -- (2.1, 7.1);
    \draw[red, dashed] (2.1,3.6) -- (3.5, 3.6);
    \draw[red, dashed] (2.45,2.55) -- (2.45,7.1);
    \draw[red, dashed] (2.45,2.55) -- (3.55,2.55);
\end{tikzpicture}
\caption{The direct summand corresponding to the small circle in \autoref{fig:barcode} with an overlaid grid. }
\label{fig:grid}
\end{figure}

Consider any cell $C$ in the grid on the left in \autoref{fig:grid}, which is slightly coarser than the induced grid. 
For all $\alpha \in C$, the submodules $\langle V_\alpha \rangle$ will look qualitatively similar: They can be deformed into one another. In consequence, we can find presentations for each with the same underlying matrix and slightly changed degrees. We will use this fact to avoid recomputing decompositions and HN filtrations for every $\alpha$ on the $\varepsilon$-grid.

 \paragraph{Shifting the Presentation Matrix.}
 Consider a module $V$, $\alpha \in \bbR^d$, and let $C$ be the cell in $\grid(V)$ with $\alpha \in C$. For $J \subset [d]$, denote by $F_J \subset \bbR^d$ the hyperplane spanned by the coordinate axes indicated by $J$. There are subsets $b^\alpha_i(\langle V_\alpha\rangle) \subset b_i(\langle V_\alpha\rangle)$ of the graded Betti numbers which lie on one of the hyperplanes $F_J+\alpha$ for $|J| < d$. For all other Betti numbers, none of their coordinates can lie in the projections to the $i$-th coordinate of $C$, $\pi_i(C)$, by definition of $\grid(V)$. For $\beta \in C$, to get the Betti numbers of the module $\langle V_\beta \rangle$ we just shift the center of the hyperplanes $F_J+\alpha$ for $\alpha$ to $\beta$ (\autoref{fig:alpha_change}) and drag all the graded Betti numbers on the hyperplanes along. That is, we change every coordinate $\alpha_i$, $i \in [d]$ to $\beta_i$. 

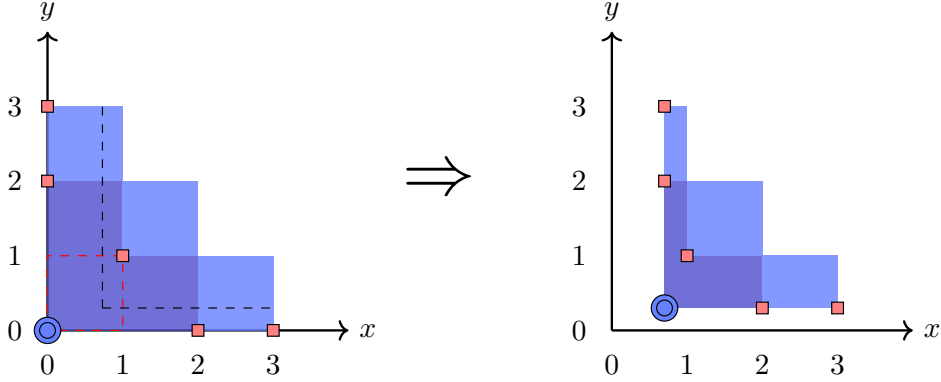
\begin{figure}[H]
{\definecolor{lightred}{rgb}{1, 0.5, 0.5}
\definecolor{lightblue}{rgb}{0.4, 0.5, 1}
\definecolor{darkblue}{rgb}{0.3, 0.3, 0.8}
\newcommand{\bbracks}[1]{\textcolor{lightblue}{\{} #1 \textcolor{lightblue}{\}} }
\newcommand{\gbracks}[1]{\textcolor{lightred}{(} #1 \textcolor{lightred}{)} }
\newcommand{\lrcell}{\cellcolor{lightred}}
\newcommand{\pblcell}{\cellcolor{lightblue}}

\begin{tikzpicture}[scale=0.9]
    % First diagram
    \begin{scope}
    \draw[->, thick] (0,0) -- (4,0) node[right] {$x$};
    \draw[->, thick] (0,0) -- (0,4) node[above] {$y$};
    \foreach \x in {0,1,2,3} {
        \node at (\x, -0.2) [below] {\x};
    }
    \foreach \y in {0,1,2,3} {
        \node at (-0.2, \y) [left] {\y};
    }
    \fill[darkblue, opacity=0.8] (0,0) rectangle (2,2);
    \fill[white] (1,1) rectangle (2,2);
    \fill[lightblue, opacity=0.8] (0,2) rectangle (1,3);
    \fill[lightblue, opacity=0.8] (1,1) rectangle (2,2);
    \fill[lightblue, opacity=0.8] (2,0) rectangle (3,1);
    \draw[dashed, red] (0,0) rectangle (1,1);
    \draw[dashed, black] (0.73,0.3) -- (3,0.3);
    \draw[dashed, black] (0.73,0.3) -- (0.73,3.0);
    \fill[lightblue, draw=black] (0,0) circle (5pt);
    \draw[black] (0,0) circle (3pt);
    \foreach \x/\y in {2/0, 0/2, 3/0, 1/1, 0/3} {
        \node[rectangle, draw=black, fill=lightred, inner sep=2pt] at (\x,\y) {};
    }
    \end{scope}
    
    \node at (5.2, 2) [font=\Huge] {$\Rightarrow$};
    % Second diagram (shifted right)
    \begin{scope}[xshift=7.5cm]
    
    \draw[->, thick] (0,0) -- (4,0) node[right] {$x$};
    \draw[->, thick] (0,0) -- (0,4) node[above] {$y$};
    \foreach \x in {0,1,2,3} {
        \node at (\x, -0.2) [below] {\x};
    }
    \foreach \y in {0,1,2,3} {
        \node at (-0.2, \y) [left] {\y};
    }
    \fill[darkblue, opacity=0.8] (0.7,0.3) rectangle (2,2);
    \fill[white] (1,1) rectangle (2,2);
    \fill[lightblue, opacity=0.8] (0.7,2) rectangle (1,3);
    \fill[lightblue, opacity=0.8] (1,1) rectangle (2,2);
    \fill[lightblue, opacity=0.8] (2,0.3) rectangle (3,1);
    \fill[lightblue, draw=black] (0.7,0.3) circle (5pt);
    \draw[black] (0.7,0.3) circle (3pt);
    \foreach \x/\y in {2/0.3, 0.7/2, 3/0.3, 1/1, 0.7/3} {
        \node[rectangle, draw=black, fill=lightred, inner sep=2pt] at (\x,\y) {};
    }
    \end{scope}
\end{tikzpicture}}
\caption{The uniquely generated module from \autoref{ex:stable_module} after cutting it off.}
\label{fig:alpha_change}
\end{figure}

We devote the next subsection to proving that this really constructs $\langle V_\beta \rangle$, and that this transformation preserves decompositions.

\subsection{Algebraic Theory of Induced submodules}\label{subsec:alg-details-induced}

To describe what happens formally when looking at $\langle V_\beta \rangle$ for varying $\beta$ in a grid cell, we need to have notions of cutting and extending modules along inclusions of posets. 
Developing this theory has intrinsic value because these techniques are used in other places (for example \cite[Section 4]{djk_arxiv}).

    Let $\beta \in \bbR^d$ and let $ \Persfp{\bbR^d}_{\beta_0 = \beta}$ denote the category of persistence modules that are uniquely generated at $\beta$. The assignment $V\in \Persfp{\bbR^d}\mapsto  \langle V_\beta\rangle\in\Persfp{\bbR^d}_{\beta_0 = \beta}$ defines an additive functor $\langle -_\beta\rangle$ which fits into the adjunction
    \begin{equation}
        \label{eq:adjunction_langle_V_beta_rangle}
     (\Persfp{\bbR^d}_{\beta_0 = \beta} \into \Persfp{\bbR^d}) \quad \dashv  \quad  \langle -_\beta\rangle.
    \end{equation}
    In particular, $\langle -_\beta\rangle$ is left exact.

  \begin{definition}\label{def:tranfer-functor}
For $\alpha, \beta\in\bbR^d$, we define the \emph{transfer functor} 
\[t_{\beta \to \alpha} =
(\iota_\alpha)_!(\iota_\alpha)^*(\iota_\beta)_*(\iota_\beta)^* \colon \Persfp{\bbR^d} \to \Persfp{\bbR^d}.\qedhere\]
\end{definition}

\begin{proposition}\label{prop:sur_restriction}
Let $\alpha\leq\beta$ in $\bbR^d$ and let $V\in\Persfp{\bbR^d}$ such that $V_{\alpha\to \beta}$ is surjective. We have 
\[
 \langle V_\beta \rangle \iso \langle \langle V_\alpha \rangle_\beta \rangle \iso  (\iota_{\beta})_! (\iota_{\beta})^* \langle V_\alpha\rangle \iso  t_{\alpha\to\beta} \langle V_\alpha\rangle.\qedhere
\]
\end{proposition}
\begin{proof}
   Let $V\in\Persfp{\bbR^d}$ such that $V_{\alpha\to \beta}$ is surjective, and let $\gamma\geq\beta$. Since $V_{\alpha\to \gamma}$ factors as $V_\alpha \to V_\beta \to V_\gamma$ it always holds that $ \langle V_\alpha\rangle_\gamma \subset \langle V_\beta \rangle_\gamma$. The first map in the factorisation is surjective, so $ \langle V_\alpha\rangle _\gamma = \langle V_\beta \rangle_\gamma$. In other words, $\langle \langle V_\alpha \rangle_\beta \rangle = \langle V_\beta \rangle$. By \autoref{ex:kan}, the equalities of subspaces assemble to the isomorphism $(\iota_{\beta})_! (\iota_{\beta})^*\langle V_\alpha\rangle\iso (\iota_{\beta})_! (\iota_{\beta})^* \langle V_\beta \rangle$. Denoting by $\iota$ the inclusion of posets $\langle\beta\rangle\hookrightarrow\langle \alpha\rangle$, we have $t_{\alpha\to\beta} = (\iota_\beta)_!( \iota_\alpha \iota)^* (\iota_\alpha)_*\iota_\alpha^*  = (\iota_\beta)_!\iota^*c^*_{\iota_\alpha}\iota_\alpha^* = (\iota_{\beta})_! (\iota_{\beta})^*$. Finally, 
   \[
  t_{\alpha\to\beta} \langle V_\alpha\rangle \iso  (\iota_{\beta})_! (\iota_{\beta})^*  \langle V_\alpha\rangle \iso (\iota_{\beta})_! (\iota_{\beta})^* \langle V_\beta \rangle \iso \langle V_\beta \rangle.\qedhere\] 
\end{proof}

\begin{proposition}\label{prop:grid_similar}
Let $\module \in \Persfp{\bbR^d}$ be a persistence module, $\alpha \in \bbR^d$, and $\beta \geq \alpha$ arbitrary such that $\alpha$ and ${\beta}$ are in the same cell of $\grid(V)$. It holds that 
\begin{enumerate}[(i)]
\item If $(A[R_i], d_i)_{i \in [d]}$ is a minimal resolution of $\langle V_\alpha \rangle$, then $\left( (\iota_{\beta})_! (\iota_{\beta})^* A[R_i], (\iota_{\beta})_! (\iota_{\beta})^* d_i \right)_{i \in [d]}$ is a minimal resolution of $\langle V_\beta\rangle$. For $i\in [d]$, $\left( (\iota_{\beta})_! (\iota_{\beta})^* A[R_i], (\iota_{\beta})_! (\iota_{\beta})^* d_i \right)$ is formed by replacing, for every $j\in [d]$ and every $\gamma \in R_i$, every coordinate $\gamma_j$ by $\beta_j$ if $\gamma_j = \alpha_j$. In particular, $b_i(\langle V_\beta \rangle) = b_i(\langle V_\alpha \rangle) \vee \beta$.

\item There is an isomorphism $t_{\beta \to \alpha}\langle V_\beta \rangle \iso \langle V_\alpha \rangle$

\item If $\varphi \colon \langle V_\alpha\rangle \iso \bigoplus V_j$ is an indecomposable decomposition, then applying $\langle -_\beta \rangle$ induces an indecomposable decomposition $\varphi_{| \langle V_{\beta} \rangle} \colon \langle V_{\beta} \rangle \iso \bigoplus \langle (V_j)_{\beta} \rangle $ 

\end{enumerate}
\end{proposition}

The consequences of \autoref{prop:grid_similar} are as follows: $(i)$ lets us compute integrals of $\udim\langle V_\beta\rangle$ as a function of integrals of $\udim\langle V_\alpha \rangle$, $(ii)$ lets us compare Harder-Narasimhan filtrations, and $(iii)$ means that we do not need to recompute decompositions if $\beta$ changes.
\begin{proof}
\begin{enumerate}[(i)]
\item \autoref{cor:res_res} shows that  $\left( (\iota_{\beta})_! (\iota_{\beta})^* A[R_i], (\iota_{\beta})_! (\iota_{\beta})^* d_i \right)_{i \in [d]}$ is a resolution of  $ (\iota_{\beta})_! (\iota_{\beta})^* \langle V_\alpha \rangle$, which by \autoref{lem:same-pos-iso} and  \autoref{prop:sur_restriction} is isomorphic to   $\langle V_\beta \rangle$.  

We now prove the minimality of $\left( (\iota_{\beta})_! (\iota_{\beta})^* A[R_i], (\iota_{\beta})_! (\iota_{\beta})^* d_i \right)_{i \in [d]}$. Let $M \colon A[S] \to A[R]$ be any map in the resolution of $\langle V_\alpha \rangle$. Since $\alpha$ and $\beta$ lie in the same cell of $\grid(V)$, the $i$-th coordinate of every $s \in S$ and every $r \in R$ is either $\alpha_i$ or larger than $\beta_i$. The application of the functor is equivalent to applying $\vee \beta$ by \autoref{lmm:restriction} and so this will turn any coordinate of the form $\alpha_i$ to $\beta_i$. Therefore, the map $\vee \beta \colon \bbR^d \to \bbR^d$ given by this procedure is injective when restricted to $S \cup R$ and so must preserve minimality.

\item Denote by $\iota\colon \langle \beta \rangle \into \langle \alpha \rangle$ the inclusion of posets. 
Again, by \autoref{lem:same-pos-iso} and \autoref{prop:sur_restriction} there is an isomorphism $\iota_\beta^* \langle V_\beta \rangle \iso \iota_\beta^*(\iota_\beta)_!\iota_\beta^* \langle V_\alpha \rangle \iso  \iota_\beta^* \langle V_\alpha \rangle \iso \iota^* \iota_\alpha^* \langle V_\alpha \rangle$. Denote by $u_*$ the unit of the adjunction $\iota^* \dashv \iota_*$. For any $\delta \geq \alpha$, the map of vector spaces
\[ \langle V_\alpha \rangle _\delta = \iota_\alpha^* \langle V_\alpha \rangle _\delta  \overset{u_*}\iso \iota_* \iota^* \iota_\alpha^* \langle V_\alpha \rangle_\delta \iso  \iota_*  \iota_\beta^* \langle V_\beta \rangle_\delta \overset{\mathclap{\ref{ex:kan}}} \iso \langle V_\alpha\rangle_{\delta \vee \beta}\]
is given by the structure map $\langle V_\alpha\rangle_{\delta \to \delta \vee \beta}$. $\delta$ and $ \delta \vee \beta$ lie again in the same cell of $\grid (V)$. Since they also belong to $\langle \alpha\rangle$, they lie in the same grid cell of $\grid(\langle V_\alpha\rangle)$. Thus, $\langle V_\alpha\rangle_{\delta\to \delta\vee\beta}$ is an isomorphism by \autoref{lem:same-pos-iso} and $(u_*)$ is an isomorphism at $(\iota_\alpha)^* \langle V_\alpha \rangle$.

\begin{align*}
(\iota_\alpha)_! \iota_\alpha^* (\iota_\beta)_* \iota_\beta^* \underline{ \langle V_\beta \rangle}
& \overset{ \mathclap{\ref{prop:sur_restriction}}}{\iso} 
(\iota_\alpha)_! \iota_\alpha^* (\iota_\beta)_* \iota_\beta^* \langle V_\alpha \rangle 
\iso 
(\iota_\alpha)_! \iota_\alpha^* (\iota_\alpha)_* \iota_* \iota^* \iota_\alpha^* \langle V_\alpha \rangle \\
& \overset{ \mathclap{\ref{rem:Kan}}}{\iso} (\iota_\alpha)_! \iota_* \iota^* \iota_\alpha^* \langle V_\alpha \rangle 
\overset{u_*^{-1}}{\iso} (\iota_\alpha)_! \iota_\alpha^* \langle V_\alpha \rangle
\iso \langle V_\alpha \rangle
\end{align*}

\item Since $\langle -_\beta \rangle$ is an additive functor, so $\langle \varphi_\beta \rangle \colon
 \langle \langle V_\alpha \rangle_\beta \rangle \overset{ \mathclap{\ref{prop:sur_restriction}}}{=} \langle V_\beta \rangle \iso \bigoplus \langle (V_j)_\beta \rangle$ is an isomorphism.
 For any $j$, the module $V_j$ is again generated at $\alpha$ and $\beta$ is in the same position as $\alpha$ wrt. $\grid(V_j)$. Therefore,
 
\begin{align*} \End(\langle (V_j)_\beta \rangle  ) 
\overset{ \mathclap{\ref{prop:sur_restriction}}}{\iso} & \Hom \left( (\iota_\beta)_!(\iota_\beta)^* V_j, (\iota_\beta)_! (\iota_\beta)^* V_j \right) 
 \overset{ \mathclap{\ref{ex:kan}}}{\iso}  \Hom \left( (\iota_\beta)^* V_j, (\iota_\beta)^* V_j\right)  \\
  \iso & \Hom \left(  V_j, (\iota_\beta)_*(\iota_\beta)^* V_j\right)  
\overset{ \mathclap{(\ref{eq:adjunction_langle_V_beta_rangle})}}{\iso} \Hom \left(  V_j, \langle {(\iota_\beta)_*(\iota_\beta)^* V_j}_\alpha \rangle \right)  \\
\iso & \Hom \left(  V_j, (\iota_\alpha)_! (\iota_\alpha)^* (\iota_\beta)_*(\iota_\beta)^* V_j \right)  
 \overset{(ii)}{\iso} \Hom \left( V_j, V_j\right)  
\end{align*}
It follows that $\langle (V_j)_\beta \rangle$ is indecomposable, too.\qedhere
\end{enumerate}
\end{proof}

\begin{remark}\label{rmk:generalize_transfer_fun_partition}   Given $V\in\Persfp{\bbR^d}$, the transfer functors define the following relation $\sim_{t,V}$ on $\bbR^d$:
\[
\alpha \sim_{t,V} \beta \quad \Longleftrightarrow \quad (\alpha\leq \beta \text{ and }t_{\beta \to \alpha}\langle V_\beta \rangle  \iso \langle V_\alpha \rangle).
\]
Extending this relation by symmetry and transitivity yields an equivalence relation. Combining \autoref{prop:grid_similar}(ii) and \autoref{prop:sur_restriction}, we see that the partition $\bbR^d/\!\sim_{t,V}$ is coarser  than the one induced by $\grid (V)$. A possible generalisation of \autoref{prop:grid_similar}(i,iii) is to replace the assumption $\lfloor\alpha\rfloor_{\grid(V)} = \lfloor \beta\rfloor_{\grid(V)}$ with $\alpha\sim_{t,V}\beta$. This would allow \autoref{alg:full_hnf} to iterate at \autoref{line:loop_grid_Vi} over the coarser partition $\bbR^d/\!\sim_{t,V_i}$ instead of the one induced by $\grid (V_i)$. For instance, the degrees $\alpha,\beta\in\bbR^2$ in \autoref{fig:grid} satisfy $\alpha\mathrel{\sim_{t,V}} (\alpha\vee\beta)\mathrel{\sim_{t,V} } \beta$ but do not lie in the same rectangular cell of $\grid(V)$.
\end{remark}

\subsection{Slope Polynomials}
For $\alpha \in \bbR^d$, we write $\supp \alpha \coloneqq \{i \in [d] \mid \alpha_i \neq 0\}$.

\begin{lemma}\label{lem:slope_poly}
Let $\module$ be generated at $0$.
For $\alpha$ in the cell of $\grid(\module)$ containing $0$, the function
 $\alpha \mapsto \mu \left( \langle V_\alpha \rangle \right)^{-1}$ is a multilinear polynomial of degree $d$ which we denote by $p_V$ and
\begin{equation}\label{eqn:pvalpha}
    p_V(\alpha)  = \frac{1}{\thick{V}} \sum\limits_{J \subset [d]} (-1)^{|J|} \int\limits_{F_{[d] \setminus J}} \udim \iota_ {F_{[d] \setminus J}}^* V\alpha ^J = \sum\limits_{J \subset [d]} (-1)^{|J|} \mu \left( \iota_ {F_{[d] \setminus J}}^* V\right)^{-1} \alpha ^J.
\end{equation} 
In particular, the leading coefficient is $(-1)^d$ and the constant term is $\mu(\module)^{-1}$.
\end{lemma}

We prove a more general version that can be used for a central charge where $\bbR^d$ is equipped with a measure coming from an $L^1$ function.

\begin{proposition}\label{prop:poly_slope} 
Let $\module$ be generated at $0$, $\beta$ in the first cell of $\grid(\module)$, and $f \in L^1(\bbR^d)$, then
\begin{equation}\label{eq:int_dim_f}
 \int\limits_{\bbR^d} \udim (\langle V_\beta \rangle) f  = \sum\limits_{i = 1}^{d} (-1)^{i} \sum\limits_{\gamma \in b_i(\module)} 
\int\limits_{\prod\limits_{j \not \in \supp(\gamma)}[\beta_j, \infty)} 
\int\limits_{\prod\limits_{j  \in \supp(\gamma)}[\gamma_j, \infty)}  f.
 \end{equation}
 
 If $\supp(\module)$ is bounded and $f=\mathds{1}_{\supp(\module)}$, this is a multi-linear polynomial in $\beta$.
\end{proposition}

\begin{proof}
By \autoref{prop:grid_similar} (i), the Betti-numbers of $\langle V_\beta \rangle$ are those of $\module$ by replacing every $0$ in the $j$-th coordinate with $\beta_j$:
\begin{align*} 
\int_{\bbR^d} \udim  \langle V_{\beta} \rangle f\ 
&\overset{\mathclap{(\ref{eq:dim_betti_sum})}}{=} \quad 
\int_{\bbR^d} \left(
\sum_{i=0}^d (-1)^i \sum_{\gamma \in b_i(\langle \module_{\beta} \rangle)} \mathds{1}_{\langle \gamma \rangle}\right)f \\
&= \quad\sum_{i=0}^d (-1)^i \sum_{\gamma \in b_i(\langle \module_{\beta} \rangle)}  \int_{\langle \gamma\rangle} f\\
 & \overset{\mathclap{\ref{prop:grid_similar}(i)}}{=}\quad  \sum\limits_{i=0}^d (-1)^i \sum_{\gamma \in b_i(\module) \vee \beta} \int_{ \langle \gamma\rangle } f \\
&= \quad\sum_{i=0}^d (-1)^i \sum_{\gamma \in b_i(\module)} \int_{\langle \gamma\vee \beta\rangle} f \\
&\overset{\mathclap{\ref{prop:grid_similar}(i)}}{=} \quad \sum\limits_{i = 1}^{d} (-1)^{i} \sum\limits_{\gamma \in b_i(\module)} 
\int\limits_{\prod\limits_{j \not \in \supp(\gamma)}[\beta_j, \infty)} 
\int\limits_{\prod\limits_{j  \in \supp(\gamma)}[\gamma_j, \infty)}  f.
 \end{align*}

If $\supp(\module)$ is bounded and $f= \mathds{1}_{\supp(\module)} \in L^1(\bbR^d)$ then by \autoref{prop:compute_slope} 

\begin{equation*} 
  \int\limits_{\bbR^d} \udim (\langle V_\beta \rangle) = \sum\limits_{i = 0}^{d} (-1)^{d+i} \sum\limits_{\gamma \in b_i(\module)} \prod\limits_{j \not\in \supp(\gamma)} \beta_j \prod\limits_{j \in \supp(\gamma)} \gamma_j, 
 \end{equation*}
 which is a multilinear polynomial in $\beta_1, \dots, \beta_d$.
\end{proof}

\begin{proposition}\label{prop:variation} Let $\module$ be generated at $0$ and $f \in L^1(\bbR^d)$. Then $\intdim \langle V_\alpha \rangle f$ is differentiable wrt. $\alpha$ on the interior of the first grid cell $C$ and
\[ \lim\limits_{\varepsilon \to^+ 0}
   \left( \frac{d^{|J|}}{d\alpha^J} \right)_{\alpha = \vec \varepsilon} \intdim \langle V_\alpha \rangle f = (-1)^{|J|}  \int\limits_{F_{[d] \setminus J}} \udim \iota_ {F_{[d] \setminus J}}^* V \]
\end{proposition}

 The upshot is that although the inverse slope is not a polynomial in this general case, in a small neighbourhood of $0$ we can still approximate it by the slope polynomial (cf \textbf{Open questions} in \autoref{sec:conclusion}), this is especially true if $f$ is a step function.

\begin{proof}

 $\udim { \module }$ is supported on the positive orthant $\bbR_+^d$.  Let $\alpha \in \bbR^d_+$, then for every $I \subset [d]$, we get a subset $\prod_{i \in I}[0,\alpha_i] \times \prod_{i \not \in I }[0, \infty) \subset \bbR_+^d$ which can be used to describe the set $\langle \alpha\rangle $ via inclusion-exclusion:

For any positive $L^1$ function $f$ we have (\autoref{fig:proof}) 

\begin{equation}\label{eq:ex_in} \int_{ \langle \alpha\rangle} f = \sum\limits_{I \subset [d]} (-1)^{|I|}  \int\limits_{\prod_{i \in I}[0,\alpha_i]} \int\limits_{\prod_{i \not \in I } [0, \infty)} f. 
\end{equation}

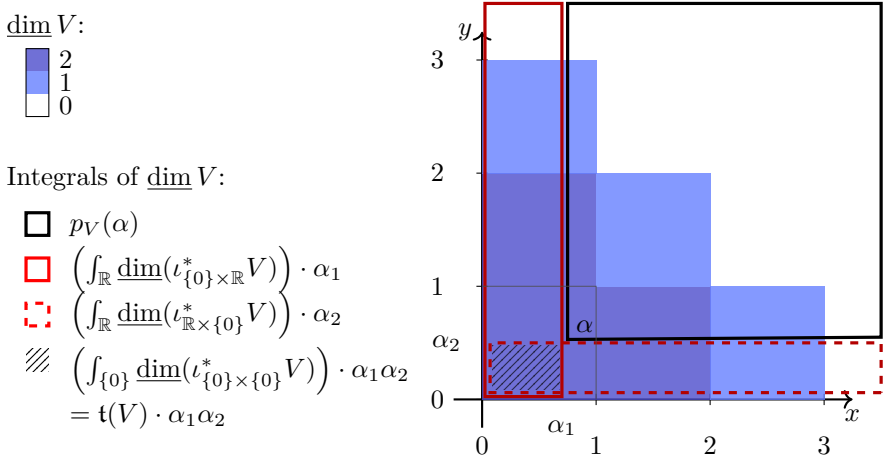
\begin{figure}[H]
\raisebox{-\height}{\definecolor{lightred}{rgb}{1, 0.5, 0.5}
\definecolor{darkred}{rgb}{0.8, 0.3, 0.4}
\definecolor{darkerred}{rgb}{0.7,0,0}
\definecolor{lightblue}{rgb}{0.4, 0.5, 1}
\definecolor{darkblue}{rgb}{0.3, 0.3, 0.8}
\definecolor{mygray}{rgb}{0.3,0.3,0.3}

\newcommand{\bbracks}[1]{\textcolor{lightblue}{\{} #1 \textcolor{lightblue}{\}} }
\newcommand{\gbracks}[1]{\textcolor{lightred}{(} #1 \textcolor{lightred}{)} }
\newcommand{\lrcell}{\cellcolor{lightred}}
\newcommand{\pblcell}{\cellcolor{lightblue}}

\begin{tikzpicture}[scale=1.5]
    % Axes
    \draw[->, thick] (-0.25,0) -- (3.25,0) node[below] {$x$};
    \draw[->, thick] (0,-.25) -- (0,3.25) node[ left] {$y$};
    % Axis labels
    \foreach \x in {0,1,2,3} {
        \node at (\x, -0.25) [below] {\x};
        \draw (\x,-0.05)--(\x,0.05);
    }
    \foreach \y in {0,1,2,3} {
        \node at (-0.25, \y) [left] {\y};
        \draw (-0.05,\y)--(0.05,\y);
    }
    \pgfmathsetmacro{\alphax}{0.7}
    \pgfmathsetmacro{\alphay}{0.5}

    % Coordinates
    \node at (-0.1, \alphay) [left] {$\alpha_2$};
    \node at (\alphax, -0.1) [below] {$\alpha_1$};

    % Darker blue region: [0,2)^2 minus [1,2)^2
    \fill[darkblue, opacity=0.8] (0,0) rectangle (2,2);
    \fill[white] (1,1) rectangle (2,2);

    % Lighter blue regions:
    \fill[lightblue, opacity=0.8] (0,2) rectangle (1,3);
    \fill[lightblue, opacity=0.8] (1,1) rectangle (2,2);
    \fill[lightblue, opacity=0.8] (2,0) rectangle (3,1);

    % Gray
    \fill[pattern=north east lines, opacity=0.8] (0.09,0.09) rectangle (\alphax-0.02,\alphay-0.02);

    \draw[darkerred, very thick](\alphax,3.5) -- (\alphax,0.025) -- (0.025,0.025) -- (0.025,3.5)--cycle;
    
    \draw[darkerred,dashed,very thick] (3.5,0.06) -- (0.07,0.06) -- (0.07,\alphay) --  (3.5,\alphay)--cycle;

    \node at (\alphax+0.2,\alphay+0.15) [black] {$\alpha$};
    \draw[black,very thick] (3.5,\alphay+0.05) -- (\alphax+0.05,\alphay+ 0.03) -- (\alphax+0.05,3.5) -- (3.5,3.5) -- cycle;

    %cell
    \draw[mygray, opacity=0.8] (0,1)--(1,1)--(1,0);

\coordinate (legend-blue) at (-4,2.5);

\begin{scope}[shift={(legend-blue)}]

    \fill[lightblue,opacity=0.8] (0,0.2) rectangle (0.2,0.4);
    \fill[darkblue,opacity=0.8]  (0,0.4) rectangle (0.2,0.6);
    \draw (0,0) rectangle (0.2,0.6);
    
    \node[anchor=west] at (-0.25,0.8) {$\udim V$:};
    \node at (0.35,0.1){0};
     \node at (0.35,0.3){1};
      \node at (0.35,0.5){2};

\end{scope}

\coordinate (legend-rect) at (-4,0.25);

\begin{scope}[shift={(legend-rect)}]

\node[anchor=west] at (-0.25,1.7) {Integrals of $\udim V$:};

        \fill[pattern=north east lines, very thick] (0,0) rectangle (0.2,0.2);
            \draw[very thick,red,dashed] (0,0.4) rectangle (0.2,0.6);
                \draw[very thick,red] (0,0.8) rectangle (0.2,1);
                    \draw[very thick] (0,1.2) rectangle (0.2,1.4);

    \node[ anchor=west] at (0.3,1.3){$p_V(\alpha)$};
    \node[ anchor=west] at (0.3,0.9){$ \left( \int_\bbR  \udim (\iota^*_{\{0\} \times \mathbb{R}}V) \right) \cdot \alpha_1  $};
   \node[ anchor=west] at (0.3,0.5){$ \left( \int_\bbR \udim ( \iota^*_{  \mathbb{R} \times \{0\} }V) \right) \cdot \alpha_2 $};
  \node[ anchor=north west,align=left] at (0.3,0.3)
 {$ \left( \int_{\{0\}} \udim (\iota^*_{\{0\} \times \{0\}}V)  \right) \cdot  \alpha_1  \alpha_2$ \\[2pt] $
    = \mathfrak{t}(V) \cdot  \alpha_1  \alpha_2  $};
\end{scope}

\end{tikzpicture}}
\caption{Continuing \autoref{ex:discrete_module}, (\autoref{eqn:pvalpha}) is obtained  for $\alpha\in[0,1)^2$ by exclusion-inclusion of the red and black rectangles. On the red rectangles, $\module$ is constant in one direction.}
\label{fig:proof}
\end{figure}

We use the fundamental theorem of calculus.
 \begin{equation*}
\begin{aligned}
   \left( \frac{d^{|J|}}{d\alpha^J} \right)_{\!\!\alpha}\  \int\limits_{\bbR^d} \udim \langle V_\alpha \rangle f \:
  & \overset{ \mathclap{\ref{prop:sur_restriction}}}{=} \quad 
    \left( \frac{d^{|J|}}{d\alpha^J} \right)_{\!\!\alpha}\  \ 
    \int\limits_{ \prod\limits_{j \in [d]} [\alpha_j, \infty)} \udim \module f  \\
  & \overset{\mathclap{(\ref{eq:ex_in})}}{=}  \quad 
   \sum\limits_{I \subset [d]} (-1)^{|I|} 
   \left( \frac{d^{|J|}}{d\alpha^J} \right)_{\!\!\alpha}\ 
   \int\limits_{\prod\limits_{i \in I}[0,\alpha_i]}
   \int\limits_{F_{[d] \setminus I}} \udim \module f 
\\
   &\hspace{-5em}\left( \text{If } I \cap J \neq \emptyset \text{ and say } j \in J \setminus I, \text{ then the integral is constant wrt. } \alpha_j \right)
\\
& =  \quad
\sum_{I \supset J} (-1)^{|I|} 
\left( \frac{d^{|J|}}{d\alpha^J} \right)_{\!\!\alpha}\ 
\int\limits_{\prod_{i \in I}[0,\alpha_i]}
\int\limits_{F_{[d] \setminus I}} \udim \module f
   \\
  &=\quad 
   \sum\limits_{I \supset J} (-1)^{|I|} 
   \int\limits_{\prod\limits_{i \in I \setminus J}[0,\alpha_i]}
   \int\limits_{F_{[d] \setminus I}} 
   \ \left( \udim \module f \right)_{| F_{[d] \setminus J} \times \{\alpha_j\}_{j \in J}}  \\
   &=\quad
   \sum\limits_{I \subset \left( [d] \setminus J\right)} (-1)^{|I|+|J|} 
   \int\limits_{\prod\limits_{i \in I}[0,\alpha_i]}
   \int\limits_{F_{([d]\setminus J) \setminus I}} 
   \ \left( \udim \module f \right)_{| F_{[d] \setminus J} \times \{\alpha_j\}_{j \in J}}  \\
  &  \overset{\mathclap{(\ref{eq:ex_in})}}{=} \quad  (-1)^{|J|} 
   \int\limits_{\prod\limits_{i \in [d] \setminus J}[0,\alpha_j]}
   \left(  \udim \module f \right)_{| F_{[d] \setminus J} \times \{\alpha_j\}_{j \in J}}.
\end{aligned}
\end{equation*}
Hence,
\[
\lim\limits_{\vec \varepsilon \to^+ 0}
   \left( \frac{d^{|J|}}{d\alpha^J} \right)_{\!\alpha = \vec \varepsilon}\ \int\limits_{\bbR^d} \udim \langle V_\alpha \rangle f = 
 \quad  (-1)^{|J|}  \int\limits_{F_{[d] \setminus J}} \left( \udim V f \right)_{| F_{[d] \setminus J} }. \qedhere
\]
\end{proof}

\begin{proof}[Proof of \autoref{lem:slope_poly}]
See \autoref{fig:alpha_change} for an illustration in the case $d=2$. By \autoref{prop:poly_slope}, the inverse slope is a  degree-$d$ multilinear polynomial $p_V$. The coefficient of $\alpha^J$ is therefore given by $\left( \frac{d^{|J|}}{d\alpha^J} \right)_{\alpha = 0} p_V(\alpha)$. Since $p_V(\alpha)$ equals the inverse slope only on the positive orthant, we can pass to the inverse slope once we replace $0$ with any sequence in $\bbR^d_+$ that converges to $0$. Using \autoref{prop:variation}, we have

\begin{align*}
\left( \frac{d^{|J|}}{d\alpha^J} \right)_{\alpha = 0} p_V(\alpha) \ 
&= \quad \lim\limits_{\varepsilon \to^+0} \left( \frac{d^{|J|}}{d\alpha^J} \right)_{\alpha = \vec \varepsilon} p_V(\alpha) \\
& \overset{\mathclap{\text{definition}}}{=} \quad \lim\limits_{\varepsilon \to^+0} \left( \frac{d^{|J|}}{d\alpha^J} \right)_{\alpha = \vec \varepsilon} \mu \left( \langle V_\alpha \rangle \right)^{-1} \\
&= \quad \lim\limits_{\vec \varepsilon \to^+ 0}
   \left( \frac{d^{|J|}}{d\alpha^J} \right)_{\alpha = \vec \varepsilon} \frac{ \intdim \langle V_\alpha \rangle }{\thick{V}}\\
&=  \quad\frac{(-1)^{|J|}}{\thick{V}}  \int\limits_{F_{[d] \setminus J}} \udim (\iota_{F_{[d] \setminus J}}^* V )  \\
&= \quad (-1)^{|J|} \mu \left( \iota_ {F_{[d] \setminus J}}^* V\right)^{-1}.\qedhere
\end{align*}
\end{proof}

\begin{observation}\label{obs:lower_envelope}
To find the submodule of highest slope for all $\alpha$ in a grid cell $C$ at the same time, compute the lower envelope of the slope polynomials of all non-isomorphic submodules. Since they have the same leading term $(-1)^d \alpha^{[d]}$ we can omit it for the computation of this diagram. We call this the \emph{truncated} slope polynomial and  denoted by $\tilde p_V$. 
\end{observation}

\subsection{Computing Arrangements for the Highest Slope Submodule.}

 For the important case of $d=2$, \autoref{obs:lower_envelope} tells us that the regions with equivalent highest slope submodules form the minimisation diagram of the lower envelope of a set of \emph{planes} in $\bbR^3$. This subdivision is therefore convex-polygonal and it can be computed efficiently as the upper convex hull of its dual point set \cite[14.2]{deBerg2008}.

\begin{example}
We first perform the calculation by hand on the stable module $\module$ from \autoref{ex:stable_module}, as indicated by \autoref{fig:proof}.
Recall that there were 4 sub vector spaces of $\module_0 = \K^2$ to consider, regardless of the choice of field.
Using the formula \autoref{lem:slope_poly} for the slope polynomials, we only need to compute the slope of the 1-parameter modules given by the restrictions to the $x$- and $y$-axis. We draw the barcodes of the restrictions:

\begin{align*}
  U_1 \coloneqq \left\{ \begin{pmatrix} 1 \\ 0\end{pmatrix} \begin{pmatrix} 0 \\ 1\end{pmatrix} \right \} \colon &
  \quad \iota_{F_{\{1\}}}^* \langle U_1 \rangle \simeq \begin{tikzpicture}[scale=0.6]
  \draw[thick] (0,0.3) -- (2,0.3);
  \draw[thick] (0,0)   -- (3,0);
\end{tikzpicture} & \quad  &
\iota_{F_{\{2\}}}^* \langle U_1 \rangle \simeq \begin{tikzpicture}[scale=0.6]
  \draw[thick] (0,0.3) -- (2,0.3);
  \draw[thick] (0,0)   -- (3,0);
\end{tikzpicture}
\\
U_2 \coloneqq \left\{ \begin{pmatrix} 1 \\ -1\end{pmatrix} \right \} \colon & \quad
\iota_{F_{\{1\}}}^* \langle U_2 \rangle \simeq \begin{tikzpicture}[scale=0.6]
  \draw[thick] (0,0) -- (3,0);
\end{tikzpicture} & \quad &
\iota_{F_{\{2\}}}^* \langle U_2 \rangle \simeq  \begin{tikzpicture}[scale=0.6]
  \draw[thick] (0,0) -- (3,0);
\end{tikzpicture}
\\
U_3 \coloneqq \left\{ \begin{pmatrix} 0 \\ 1\end{pmatrix} \right \} \colon & \quad
\iota_{F_{\{1\}}}^* \langle U_3 \rangle \simeq \begin{tikzpicture}[scale=0.6]
  \draw[thick] (0,0) -- (2,0);
\end{tikzpicture} & \quad &
\iota_{F_{\{2\}}}^* \langle U_3 \rangle \simeq \begin{tikzpicture}[scale=0.6]
  \draw[thick] (0,0) -- (3,0);
\end{tikzpicture} 
\\
 U_4 \coloneqq \left\{ \begin{pmatrix} 1 \\ 0\end{pmatrix} \right \} \colon & \quad
 \iota_{F_{\{1\}}}^* \langle U_4 \rangle \simeq  \begin{tikzpicture}[scale=0.6]
  \draw[thick] (0,0) -- (3,0);
\end{tikzpicture} & \quad &
\iota_{F_{\{2\}}}^* \langle U_4 \rangle \simeq \begin{tikzpicture}[scale=0.6]
  \draw[thick] (0,0) -- (2,0);
\end{tikzpicture}
\end{align*}
The sum of the lengths of the bars gives the coefficients of the polynomials and we draw the minimisation diagram. It will already determine the complete HN filtration because the module has thickness $2$.
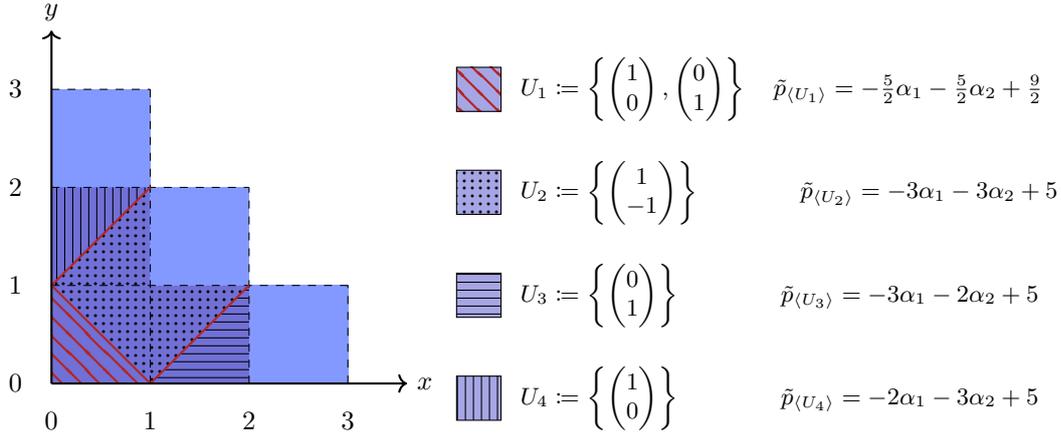
\begin{figure}[H]
\definecolor{lightred}{rgb}{1, 0.5, 0.5}
\definecolor{lightblue}{rgb}{0.4, 0.5, 1}
\definecolor{darkblue}{rgb}{0.3, 0.3, 0.8}
\definecolor{darkred}{rgb}{0.7, 0.15, 0.15}

\begin{tikzpicture}[scale=1.3]
    % ---- Base colored regions ----
    \fill[darkblue, opacity=0.8] (0,0) rectangle (2,2);
    \fill[white] (1,1) rectangle (2,2);
    \fill[lightblue, opacity=0.8] (0,2) rectangle (1,3);
    \fill[lightblue, opacity=0.8] (1,1) rectangle (2,2);
    \fill[lightblue, opacity=0.8] (2,0) rectangle (3,1);

    % ---- Pattern overlays ----
    % (A) Red thick NW-hatch (\): corner triangle at origin
    \fill[pattern=thick nw lines, pattern color=darkred]
        (0,0) -- (1,0) -- (0,1) -- cycle;

    % (B) Black dots: only darkblue part of diamond (exclude lightblue triangle (1,1)-(2,1)-(1,2))
    \fill[pattern=dots, pattern color=black]
        (0,1) -- (1,0) -- (2,1) -- (1,1) -- (1,2) -- cycle;

    % (C) Black horizontal lines: bottom-right triangle, basis {(0,1)}
    \fill[pattern=horizontal lines, pattern color=black]
        (1,0) -- (2,0) -- (2,1) -- cycle;

    % (D) Black vertical lines: upper-left triangle, basis {(1,0)}
    \fill[pattern=vertical lines, pattern color=black]
        (0,1) -- (0,2) -- (1,2) -- cycle;

    % ---- Dashed grid lines ----
    \draw[dashed, black] (0,1) -- (3,1);
    \draw[dashed, black] (0,2) -- (2,2);
    \draw[dashed, black] (0,3) -- (1,3);
    \draw[dashed, black] (1,0) -- (1,3);
    \draw[dashed, black] (2,0) -- (2,2);
    \draw[dashed, black] (3,0) -- (3,1);

    % ---- Diagonal lines in darkred ----
    \draw[darkred, thick] (1,0) -- (0,1);
    \draw[darkred, thick] (1,0) -- (2,1);
    \draw[darkred, thick] (0,1) -- (1,2);

    % ---- Axes ----
    \draw[->, thick] (0,0) -- (3.6,0) node[right] {$x$};
    \draw[->, thick] (0,0) -- (0,3.6) node[above] {$y$};
    \foreach \x in {0,1,2,3} { \node at (\x,-0.2) [below] {\x}; }
    \foreach \y in {0,1,2,3} { \node at (-0.2,\y) [left]  {\y}; }

    % ---- Legend ----
    % (A) red thick NW-hatch: basis {e1, e2}
    \fill[darkblue, opacity=0.5] (4.1,2.775) rectangle (4.55,3.225);
    \fill[pattern=thick nw lines, pattern color=darkred]
        (4.1,2.775) rectangle (4.55,3.225);
    \draw (4.1,2.775) rectangle (4.55,3.225);
    \node[right, font=\small] at (4.65,3.0)
        {$U_1 \coloneqq \left\{\begin{pmatrix}1\\0\end{pmatrix},\begin{pmatrix}0\\1\end{pmatrix}\right\} \quad  \tilde p_{\langle U_1 \rangle} = - \frac{5}{2} \alpha_1 - \frac{5}{2}  \alpha_2 + \frac{9}{2}  $};

    % (B) dots: basis {(1,1)}
    \fill[darkblue, opacity=0.5] (4.1,1.725) rectangle (4.55,2.175);
    \fill[pattern=dots, pattern color=black]
        (4.1,1.725) rectangle (4.55,2.175);
    \draw (4.1,1.725) rectangle (4.55,2.175);
    \node[right, font=\small] at (4.65,1.95)
        {$U_2 \coloneqq\left\{\begin{pmatrix}1\\-1\end{pmatrix}\right\} \quad \quad \quad  \quad \tilde p_{\langle U_2 \rangle} = - 3 \alpha_1 - 3 \alpha_2 + 5
        $};

    % (C) horizontal lines: basis {(0,1)}
    \fill[darkblue, opacity=0.5] (4.1,0.675) rectangle (4.55,1.125);
    \fill[pattern=horizontal lines, pattern color=black]
        (4.1,0.675) rectangle (4.55,1.125);
    \draw (4.1,0.675) rectangle (4.55,1.125);
    \node[right, font=\small] at (4.65,0.9)
        {$U_3 \coloneqq \left\{\begin{pmatrix}0\\1\end{pmatrix}\right\} \quad \quad \quad  \quad\tilde p_{\langle U_3 \rangle} = - 3 \alpha_1 - 2 \alpha_2 + 5$};

    % (D) vertical lines: basis {(1,0)}
    \fill[darkblue, opacity=0.5] (4.1,-0.375) rectangle (4.55,0.075);
    \fill[pattern=vertical lines, pattern color=black]
        (4.1,-0.375) rectangle (4.55,0.075);
    \draw (4.1,-0.375) rectangle (4.55,0.075);
    \node[right, font=\small] at (4.65,-0.15)
        {$U_4 \coloneqq \left\{\begin{pmatrix}1\\0\end{pmatrix}\right\} \quad \quad \quad  \quad\tilde p_{\langle U_4 \rangle} = - 2 \alpha_1 - 3 \alpha_2 + 5$};

\end{tikzpicture}
\caption{Left: The truncated slope polynomials and their minimisation diagram overlaid on the dimension-$2$ part of the Hilbert function of $\module$.}
\label{fig:slope_subdivision}
\end{figure}

This example reveals a small but important feature: At the intersection of $k\geq 2$ regions -- i.e. on lines and points -- corresponding to modules $U_1,\dots U_k$, the \emph{largest} submodule of highest slope should be $U\coloneqq U_1+\dots+U_k$. In our algorithms, we did not specify what to do if a point lands exactly on a boundary between regions. A brief calculation shows that it does not matter if we choose $U_1$ instead of $U$. Indeed, \autoref{lmm:see-saw} ensures that $U/U_1$ is the first module in the HN filtration of $V/U_1$ and has the same slope as $U$.
\end{example}

We write the preceding procedure into an algorithm over $\bbR^2$
and remark that it works just as well for higher $d$ by replacing the appropriate formulas, but computing the minimisation diagrams becomes much more time intensive and floating point arithmetic will necessarily produce larger errors when the degree of the polynomials increases.

\begin{algorithm}[H]
\caption{Highest Slope Submodule in a Grid Cell.}
\label{alg:all_max_slope}
\DontPrintSemicolon
\KwIn{A graded matrix $M$ presenting $V$ uniquely generated at $\alpha$; $C \subset \bbR^2$}
\KwOut{A convex polygonal subdivision $\{S_i \}_{i \in I}$ of $C$. For each $i \in I$ a sub vector space $W \subset V_\alpha$ and the slope polynomial $p_W$.}
 \texttt{slope-polynomials} $ \in \left(  \{\text{Finite Subsets of } V_\alpha\} \times \bbR^3 \right)[ \ ] \gets  \text{ empty array }$\;
\For{sub vector space $W \subset V_\alpha$}{
    $N \gets $ \autoref{alg:submodule}$(M, B)$, for some basis $B \subset W$\;
    $K \in \K^{R \times S} \gets \ker N $ (\textsc{mpfree}), $b_0(\langle W \rangle) \gets (\deg(B)), \ b_1(\langle W \rangle) \gets R, \ b_2(\langle W \rangle) \gets S$ \;
    $\tilde p_B \gets \frac{1}{|B|} \left( \intdim (\iota^*_{ \{0\} \times \bbR}V)x_1 +\intdim(\iota_{ \bbR \times \{0\}}^*V) x_2 + \intdim \module\right) $ with \autoref{prop:compute_slope}\;
    $\texttt{slope-polynomials}.\text{append}\left(B, \tilde p_B \right)$\;
}
\texttt{subdivision} $\{S_i\}_{i \in I} \gets $ 
\textsc{LowerEnvelope}( \texttt{slope-polynomials}) \;
$\forall i \in I \colon S_i \gets S_i \cap C $; $J \coloneqq \{i \in I \mid S_i \neq \emptyset \}$\;

\Return $\{(S_j,\texttt{slope-polynomials}_{j} )\}_{j \in J}$\;
\end{algorithm}

In the last line, $\texttt{slope-polynomials}_{j}$
is the pair $(B, \tilde p_{B})$ such that $ \tilde p_{B}$ is minimal on $S_j$.

\begin{example}
The thickness-$6$ module from \autoref{tab:combined_exp} produces the following subdivision for its first cell.

\begin{figure}[H]
\begin{minipage}[t]{0.4\textwidth}
    \centering
    \tikz[remember picture, baseline=(leftimg.south)]{
        \node[inner sep=0pt] (leftimg)
            {\includegraphics[width=1.15\textwidth]{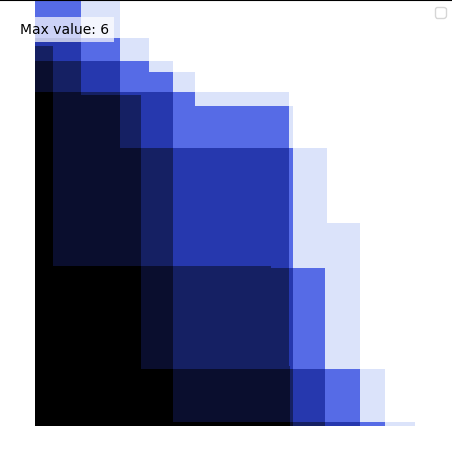}};
    }
\end{minipage}
\hfill
\begin{minipage}[t]{0.51\textwidth}
    \centering
    \tikz[remember picture, baseline=(rightimg.south)]{
        \node[inner sep=0pt] (rightimg)
            {\includegraphics[width=\textwidth]{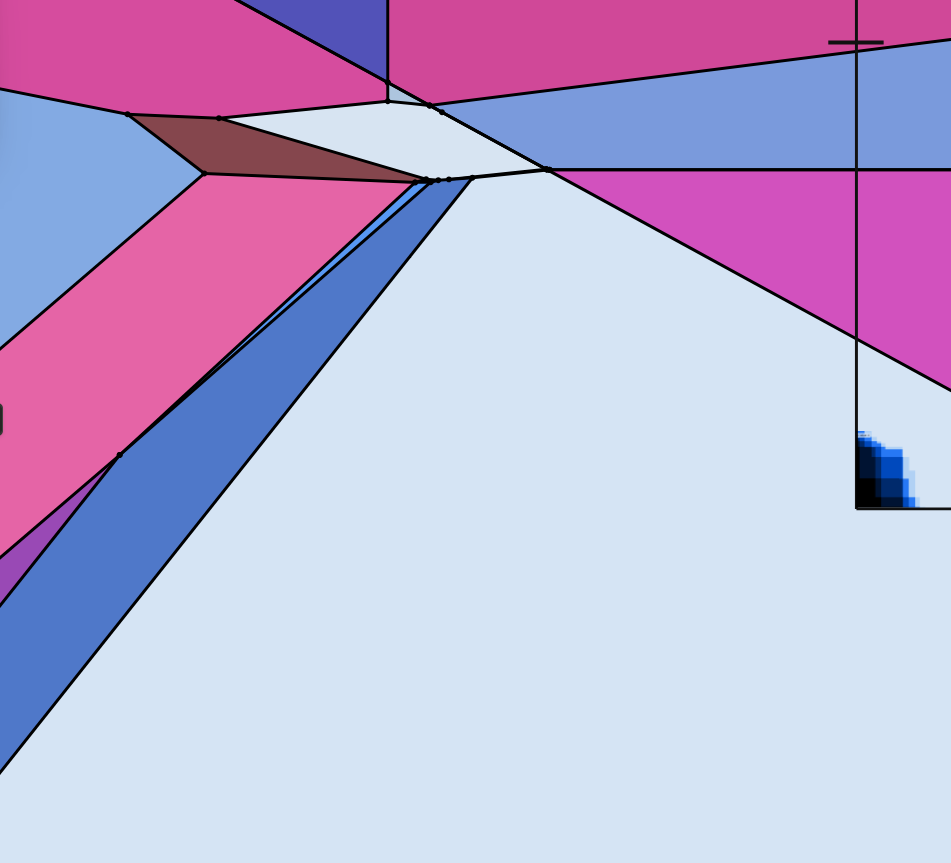}};
    }
\end{minipage}
\begin{tikzpicture}[remember picture, overlay]
    \draw[dashed, thick, opacity=0.4, dash pattern=on 8pt off 4pt] ([xshift= 0.6cm]leftimg.north west) -- ([yshift= 0.1cm, xshift=-0.7cm]rightimg.east);
    \draw[dashed, thick, opacity=0.4, dash pattern=on 8pt off 4pt] ([xshift= -0.6cm, yshift=0.6cm]leftimg.south east) -- ([yshift=-0.75cm,  xshift=-0.35cm]rightimg.east);
\end{tikzpicture}
\caption{The minimisation diagram overlaid on the Hilbert function}
\label{fig:torus_indecomp}
\end{figure}
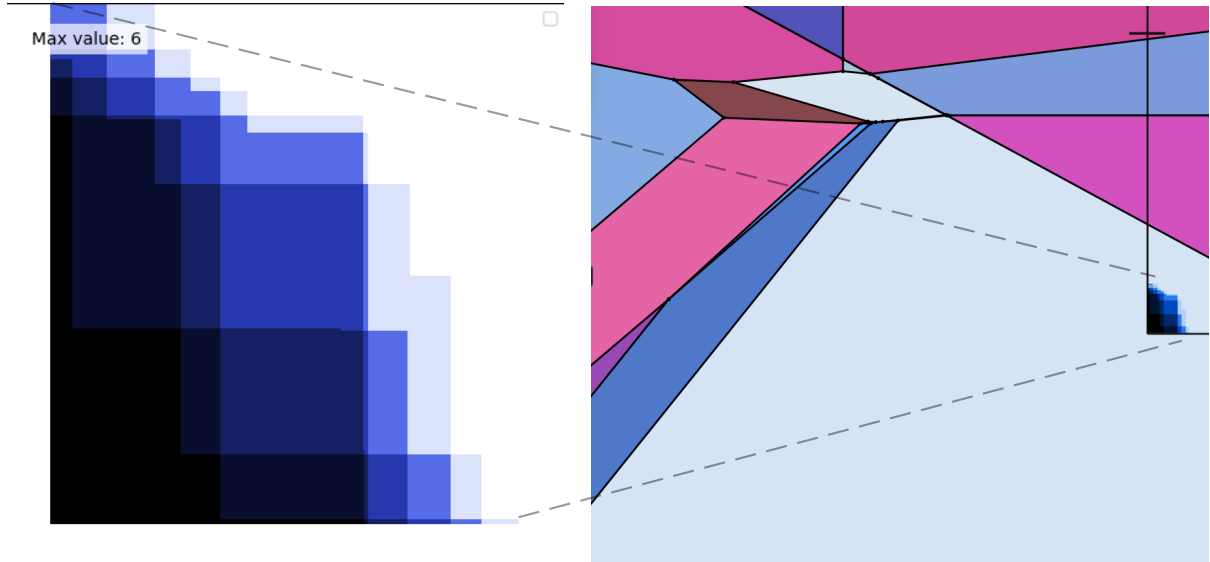

\end{example}

\begin{remark}\label{rem:subdivision_size}
We ran \autoref{alg:all_max_slope} for all modules used in \autoref{tab:combined_exp}. In each case, as in \autoref{fig:torus_indecomp}, the first grid cell lies entirely in one face of the subdivision. This means that the subdivision will be just $\{C\}$ together with one slope polynomial.
\end{remark}

\begin{proposition}\label{prop:max_slope_proof}
\autoref{alg:all_max_slope} runs in expected time $\Oc(\thick{\module}^3 2^{\thick{\module}^2/4 + \Oc(\thick{\module})})$.
Let $\{(S_j, B_j, p_j)\}_{j \in J}$ be its output. If $C$ is contained in the grid cell of $\langle V_\alpha \rangle$ with lower left corner $\alpha$, then for any $j$ and $\beta \in S_j$, $\langle \langle B_j \rangle_\beta \rangle \subset \langle V_\beta \rangle$ is the submodule of highest slope at $\beta$ and its slope is $p_j(\beta)^{-1}$.
\end{proposition}

\begin{proof}
 The runtime is the same as for \autoref{alg:brute_force}, except that kernel computation is in $\thick{\module}^3$ for 2-parameter modules and we compute a lower envelope. That can be done in expected time $n\log n$ yielding
\[ \Oc( 2^{\thick{\module}^2/4 + \thick{\module}}\log(2^{\thick{\module}^2/4 + \Oc(\thick{\module})}) = \Oc( \thick{\module}^2 2^{\thick{\module}^2/4 + \Oc(\thick{\module})} ) \qedhere.\]
Let $\beta \in S_j$. By \autoref{lem:same-pos-iso} the map $V_\alpha \to V_\beta = \langle V_\alpha \rangle_\beta$ is an isomorphism and $\langle \langle V_\alpha \rangle_\beta \rangle = \langle V_\beta \rangle$ by \autoref{prop:sur_restriction}. For every sub vector space $U \subset V_\beta$ there is a sub vector space $W \coloneqq V_{\alpha \to \beta}^{-1}(U) \subset V_\alpha$ such that $(\iota_\beta)_! \iota_\beta^* \langle W \rangle \iso \langle \langle W \rangle_\beta \rangle$ is equal to $\langle U \rangle$ as a submodule and this $W$ is chosen with some basis $B$ in the outer loop. Then \autoref{lem:slope_poly} verifies that $\mu(U)^{-1} = \mu(\langle \langle W \rangle_\beta \rangle)^{-1} = p_B(\beta)^{-1}$.
Let $B_j$ be the basis associated with the face $S_j$.  Since $p_j = p_{B_j}$ is minimal here, $p_j(\beta)$ is the minimum of all the slope polynomials at $\beta$, so the slope $\mu(\langle \langle B_j \rangle_\beta \rangle)$ is the maximum of all values $\mu(U)$.
\end{proof}

\subsection{Wall-and-Chamber Structure.}

Over a finite poset, the notion of $S$-equivalence for a module $V$ defines an equivalence relation on the space of \emph{all} stability conditions or central charges (\autoref{subsec:HN-formal}). This equivalence partitions the space of stability conditions into a finite number of codimension 0 and codimension 1 subspaces called \emph{chambers} and \emph{walls} \cite{HILLE2002205}. Here, we work over $\bbR^d$ and are only interested in the HN filtration of uniquely generated modules along \emph{Skyscraper} stability conditions, which are concentrated at a single $\alpha$. One can extend the wall-and-chamber structure to the space of Skyscraper stability conditions after pulling back the modules and stability conditions to appropriate finite sub-posets \cite[Appendix A]{fersztand2024harder}.

\begin{definition}\label{def:S_equ}Let $V \in \Persfpb{\bbR^d}$ and let $\alpha \leq \beta$ be in the same cell of $\grid(V)$. We say that $\alpha$ and $\beta$ are \emph{HN equivalent}, $\alpha \sim_{\text{HN}(V)}\beta$, if the Harder-Narasimhan filtrations $\langle V_\alpha \rangle^\bullet$ of $\langle V_\alpha \rangle$ and $\langle V_\beta\rangle^\bullet$ of $\langle V_\beta \rangle$ have the same length $\ell$ and for all $1\leq i\leq \ell$, we have $\langle V_\alpha\rangle^i \iso t_{\beta \to \alpha} \langle V_\beta\rangle^i$. We extend $\sim_{\text{HN}(V)}$ by symmetry and transitivity into an equivalence relation. 
\end{definition}

\autoref{alg:all_max_slope} does not yet return the whole HN filtration. For that we need to call it recursively on every face and so we will get a new subdivision of this face.

\begin{definition}
 A \emph{nested subdivision} of length $d$ of $C \subset \bbR^d$ is a rooted tree $T=(T, r \in T, \ \rho\ \colon T \setminus \{r\} \to T )$ of depth $d$ and a tuple of connected regions $ S_t \subset \bbR^d$ indexed by $T$ such that $S_r = C$ and
\begin{equation}
\forall t \in \Img(\rho) \colon \bigcup\limits^{\bullet}_{c \in \rho^{-1}(t)} S_c = S_t.
\end{equation}
\end{definition}

For every $\alpha \in C$ there is a unique leaf $l \in T$ with $\alpha \in S_l$. If the regions $S_v$ are convex polygonal, then point location can be done in $\log(|T|)$ time after appropriate preprocessing \cite[Theorem 6.3]{deBerg2008}.

We then want to endow each face $t \in T$ in the nested subdivision with a basis for the sub vector space $W_t$ of $V_\alpha$ which induces the filtration factor $F$ associated with it. 
Viewing $T$ as a poset with all arrows pointing to the root this is just a functor $T^{op} \to \mathcal{P}(V_\alpha)$ to the poset of subsets with inclusions. 

\begin{definition}
A pair
$\left((T, r, \rho), (S_t, W_t,  p_t)_{t \in T} \right)$ where
\begin{itemize}
    \item $(S_t)_{t \in T}$ is a nested subdivision,
    \item $(W_t)$ is a functor $T \to \mathcal{P}(V_\alpha)$ and
    \item $(p_t)_{t\in T}$ is a tuple of functions $S_t \to \bbR$ 
\end{itemize}
is called \emph{slope subdivision} of $C$ for $V$ uniquely generated at $\alpha$ if:

For an arbitrary $\beta \in C$, and $t_0 \ t_1 \dots t_\ell$  the unique directed path in $T$ from the root to a leaf such that $\beta \in S_{t_i}$ for $i = 0,\dots, \ell$, the Harder-Narasimhan filtration of $\langle V_\alpha \rangle$ \emph{at $\beta$} is 
  \[ 0 \into \langle W_{t_1} \rangle \into  \langle W_{t_2} \rangle \into \dots \into \langle W_{t_\ell} \rangle \into \langle V_\beta \rangle \into \langle V_\alpha \rangle \]
  and for each $i$, $\mu( W_{t_i}/W_{t_{i-1}}) =  p_{t_i}(\beta) + (-1)^{d} \prod_{i \in [d]} \beta_i$.
\end{definition}

We can now compute the whole HN filtration recursively.

\begin{algorithm}[H]
\caption{Harder-Narasimhan filtration for a grid cell.}
\label{alg:exact_hnf}
\DontPrintSemicolon
\KwIn{A graded matrix $M \in \K^{G \times R}$ presenting $V$ uniquely generated at $\alpha$; $C \subset \bbR^2$}
\KwOut{Rooted tree $(T, r, \rho)$. A nested convex polygonal subdivision $\{S_t \}_{t \in T}$ of $C$. A functor $W_\bullet \colon T^{op} \to \vect_K^{inj}$ and slope polynomials $\left( \tilde p_{W_t} \right)_{t \in T}$.}
Initialise $r, T \gets \{r\} , \rho$; $S = \{S_r\} \gets \{C\}$; $\left(W_t \right)_{t \in T} = (W_r) \gets (\emptyset)$; $(\tilde p_t)_{t \in T} = (\tilde p_r) \gets (0)$\;
$\{(S_i, B_i, \tilde p_{V_i})\}_{i \in I} \gets$ \autoref{alg:all_max_slope}($M, C$) \;
\For{$i \in I$}{

$M_i \gets \text{presentation of } \langle V_\alpha \rangle / \langle B_i \rangle$ (\autoref{par:quotient})\;
$\left((T_i, r_i, \rho_i), (S_t, W_t, \tilde p_t)_{t \in T_i} \right) \gets$ \autoref{alg:exact_hnf}$(M_i, S_i)$\;
$T \gets T \cup T_i$; $\rho(r_i) \gets r$\;
$(S_t,W_t, \rho_t)_{t \in T} 
\gets (S_t,W_t, \rho_t)_{t \in T} \cup (S_t, W_t \cup B_i,\tilde p_t)_{t \in T_i}$
}
\Return $\left((T, r, \rho), (S_t, W_t,\tilde  p_t)_{t \in T} \right) $\;
\end{algorithm}

\begin{proposition}\label{prop:hnf_cell_correct}
  If $C$ is contained in the grid cell of $V$ which contains $\alpha$, then the output \\
  
$\left((T, r, \rho), (S_t, W_t, \tilde p_t)_{t \in T} \right)$  of \autoref{alg:exact_hnf}$(M, C)$ is a
  \emph{slope subdivision} of $C$ for $V$.
\end{proposition}

\begin{proof}
Assume that the algorithm works for uniquely generated modules $V$ with $\thick{V} \leq n$. For $\thick{V}=1$ there is nothing to do.

Let $V$ be with $\thick{V}= n+1$, then by \autoref{prop:max_slope_proof} the first step finds all possible modules of highest slope. In particular $\mu( W_{t_1}/W_{t_{0}}) = \mu(W_{t_1}/0) = p_{t_1}(\beta) = \tilde p_{t_1}(\beta) + \beta_1 \beta_2$. The quotients must all be of thickness $\leq n$, where we know that \autoref{alg:exact_hnf} returns the correct result. To pull back the filtrations of the quotients it is enough to enlarge the basis for each sub vector space by the basis of the highest slope submodule, which is done in the last step.
\end{proof}

\begin{proof}[Proof of \autoref{thm:main}]
A slope subdivision is exactly the Wall-and-Chamber decomposition of the space of stability conditions restricted to the central charges indexed by $\alpha \in \bbR^d$ which we use here. \autoref{thm:main} is a direct consequence of
\autoref{prop:hnf_cell_correct} using the definition of equivalence in \autoref{def:S_equ} and \autoref{prop:grid_similar} (ii).
\end{proof}

\paragraph{Exact Computation of the Skyscraper Invariant.}
We assemble an algorithm to compute the Skyscraper Invariant up to arbitrary precision.

\begin{algorithm}[H]
\caption{Exact Skyscraper Invariant}
\label{alg:full_hnf}
\DontPrintSemicolon
\KwIn{A graded matrix $M \in \K^{G \times R}$ presenting $\module \in \Persfpb{\bbR^2}$}
\KwOut{A set of convex polygonal subdivisions $\{S_{i, \alpha, j,} \}_{I \times G_I \times J_{I, G_I} }$ of $\supp \module \subset \bbR^2$ and for each $F \in S_{i,j,l}$ an ordered list $HNF_l$ of semistable modules equipped with degree-$2$ multilinear polynomials, which compute their inverse slopes}
Decompose $\module$ into $\bigoplus_{i \in I} V_i$ with \texttt{AIDA} \;
\For{ $i \in I$}{
    \For{$\alpha \in \grid(V_i)$}{\label{line:loop_grid_Vi}
        Compute $\langle V_{i,\alpha} \rangle$ with \autoref{alg:submodule} \;
        Decompose $\langle V_{i,\alpha} \rangle \iso  \bigoplus_{j \in J_{i, \alpha}} V_{i, \alpha, j}$ with \texttt{AIDA} \;
        \For{$j \in J$}{
            Each time when calling \autoref{alg:all_max_slope} in the next step, for each basis $B_t$ also add a set of intervals $\Int_t \subset \mathcal{P}(\bbR^2)$ to the slope subdivision such that
            $\sum_{\mathcal{I} \in \Int_t} \mathds{1}_{\mathcal{I}} = \udim \langle B_t\rangle$ \;
            $\left((T_j, r_j, \rho_j), (S_t, W_t, \Int_t,\tilde p_t)_{t \in T_j} \right) \gets$ \autoref{alg:exact_hnf} $( V_{i, \alpha, j}, C_\alpha)$ \;
            
        }
    }
}
\Return{ $\left((T_j, r_j, \rho_j), (S_t, \Int_t, \tilde p_t)_{t \in T_j} \right)_{j \in  \bigcup_{i \in I, \alpha \in \grid(V_i)}, J_{i,\alpha}}$}
\end{algorithm}

\paragraph{Query time.}
We need to explain how to get the values of the Skyscraper Invariant $s_V^\theta(\beta, \gamma)$ from the output $\left((T_j, r_j, \rho_j), (S_t, \Int_t, \tilde p_t)_{t \in T_j} \right)_{j \in  \bigcup_{i \in I, \alpha \in \grid(V)}, J_{i,\alpha}}$.

Locate for each $i$ the next lower $\alpha \in \bbR^d$ such that for any $j \in J_{i,\alpha}$, $\beta \in S_{r_j}$. This can be done in $|I|\log(n)$ time, although in practice we will only want to sweep over $I$ and each $\grid(V_i)$ once and get the values  $s_V^\theta(\beta, \gamma)$ for all $\beta$ e.g. on an $\varepsilon$-grid at once.

For each $j \in J_{i, \alpha}$ locate the path $t_0 \dots t_\ell$ from the root to a leaf in $T_j$ such that $\beta \in S_{t_\ell}$ and consider the sets of intervals $\Int_{t_0} \dots \Int_{t_\ell}$  together with the numbers $p_{t_0}(\beta)  \dots  p_{t_\ell}(\beta)+ \beta_1 \beta_2$. 
Form a list containing all of these for all $i$ and $j$, then sort it (in $\thick{\module}\log \thick{\module}$ time). Traverse this list until the slope would become lower than $\theta$. For each $\Int_\bullet$ and $\mathcal{I} \in \Int_\bullet$ count how often $\gamma \in \mathcal{I}$ - this is the value $s_V^\theta(\beta, \gamma)$.

\begin{proposition}
The Skyscraper Invariant can be correctly computed from the output of \autoref{alg:full_hnf} with the procedure above and, using the same variables as in \autoref{prop:runtime_analysis}, its runtime is in
\[ \Oc \left( n^2\left( n^{1+\omega} \mathfrak{t}{(\module)}^\omega + n^5 + n^{\omega+1}\ell^{\omega-1}(\ell + \thick{\module}) q^{\ell^2/4 + \Oc(\ell)} + \thick{V}  \cdot kn^2 \cdot q^{k^2/4 + \Oc(k)} \right) \right) \]
if we assume that the subdivisions contain an average number of faces.
\end{proposition}

\begin{proof}
The runtime analysis is almost the same as for \autoref{alg:skeleton}. The only difference is that we iterate over $\bigcup_i\grid(\module_i)$ whose total size is at most $\Oc(n^2)$. Each time \autoref{alg:exact_hnf} is called, we compute as many new filtrations as there are faces in the minimisation diagram computed by \autoref{alg:all_max_slope}. This number is on average $\log 2^{k^2/4 + \Oc(k)} \sim \Oc(k^2)$ and is negligible compared to the iteration over all subspaces. Since we deal with $2$-parameter modules we use \textsc{mpfree} for kernel computation which is essentially a column reduction of a $k \times n$ matrix.

Correctness follows from the correctness of \autoref{alg:exact_hnf} (\autoref{prop:hnf_cell_correct}) and additivity of Harder-Narasimhan filtrations.
\end{proof}

\subsection{Parallel Grid Scanning}
In practice, where one wants to compute filtered vectorisations of the Skyscraper Invariant, it is not a problem to restrict the module to an $\varepsilon$-grid $\varepsilon \bbZ^2$ and use \autoref{alg:skeleton}. The real power of \autoref{thm:main} and \autoref{alg:all_max_slope} lies in the fact that it is a recipe to avoid recomputation: For any direct summand $V_i$ of the input module, the induced grid $\grid(V_i)$ is coarser than $\varepsilon \bbZ^2$ in many places, especially when $\varepsilon$ is small. Since, as described in \autoref{rem:subdivision_size}, the slope subdivision of a cell typically consists of a single region, we mostly need to compute a single HN filtration per cell of the induced grid and not one for every point on $\varepsilon \bbZ^2$. We therefore overlay the two grids and use \autoref{alg:all_max_slope} only for those cells in each $\grid(V_i)$ that intersect $\varepsilon \bbZ^2$. We then store the Hilbert functions of all factors in an array. This enables the computation of the Skyscraper Invariant at practically arbitrary resolution. 

\begin{definition}
Let $\mathcal{G} \subset \bbR^d$ be a grid and $\mathcal{H} \subset \bbR^d$ a poset. We define the (grid) restriction
\[ \mathcal{G}_{|\mathcal{H}} \coloneqq \mathcal{G}_{|\mathcal{H}} \coloneqq  (\lfloor -\rfloor_\gri)^{-1}(\mathcal H). \]
\end{definition}
If $\mathcal{H}$ is a grid, then $\mathcal{G}_{|\mathcal{H}}$ is also a grid.

\begin{example}\label{ex:grid_comparison}
We consider modules generated from the point sets LargeHypoxicRegion and LargeHypoxicRegion2 from \cite{Vipond21} which we will use later in \autoref{sec:experiments} to showcase the runtimes of the algorithms on non-synthetic large data (in \autoref{tab:large_data}). We computed the first homology group of the density-alpha bifiltration of LargeHypoxicRegion and consider its two largest indecomposable components, $V$ and $W$, by number of generators. We set $\varepsilon = 0.5 \%$, plot the induced grid, the $\varepsilon$-grid, and the restrictions $\grid(V)_{|\varepsilon \bbZ^2} \cap \supp(V)$ and $\grid(W)_{|\varepsilon \bbZ^2} \cap \supp(W)$ over a small patch of the parameter space. 

\begin{figure}[H]
\begin{minipage}{0.48\textwidth}
\centering
    \includegraphics[width = \textwidth]{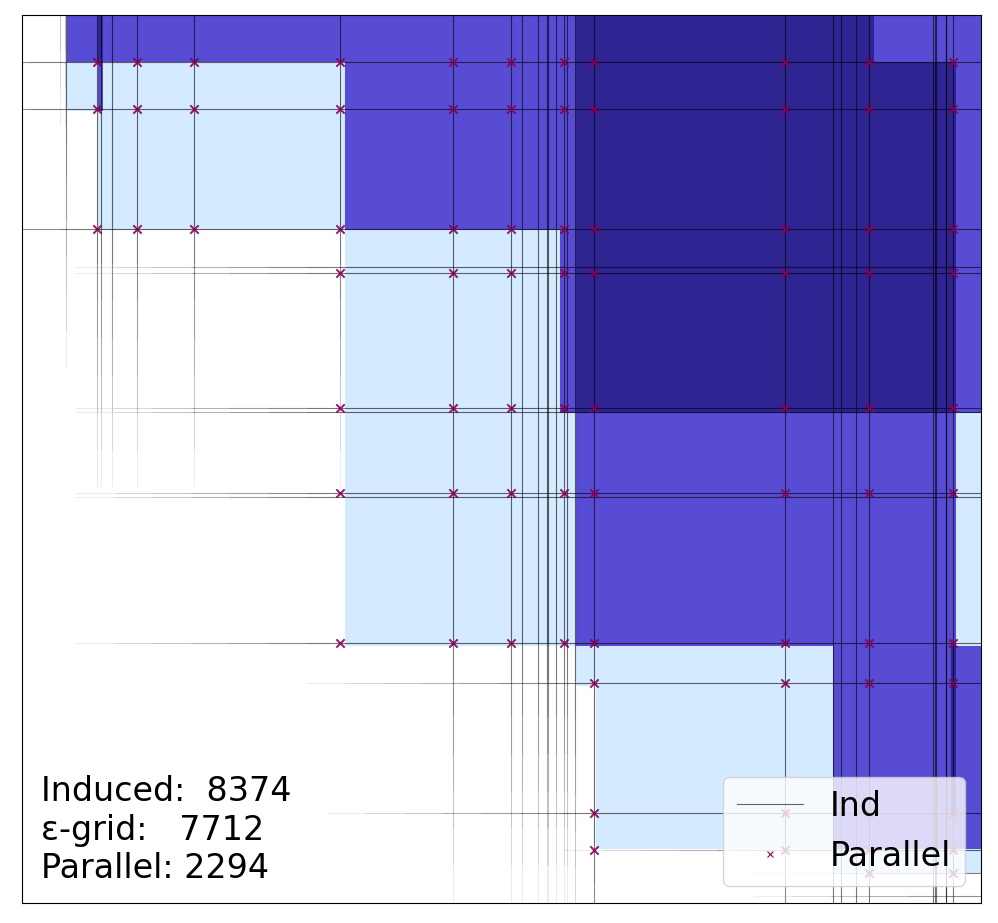}
\end{minipage}
\hfill
\begin{minipage}{0.48\textwidth}
\centering
    \includegraphics[width = \textwidth]{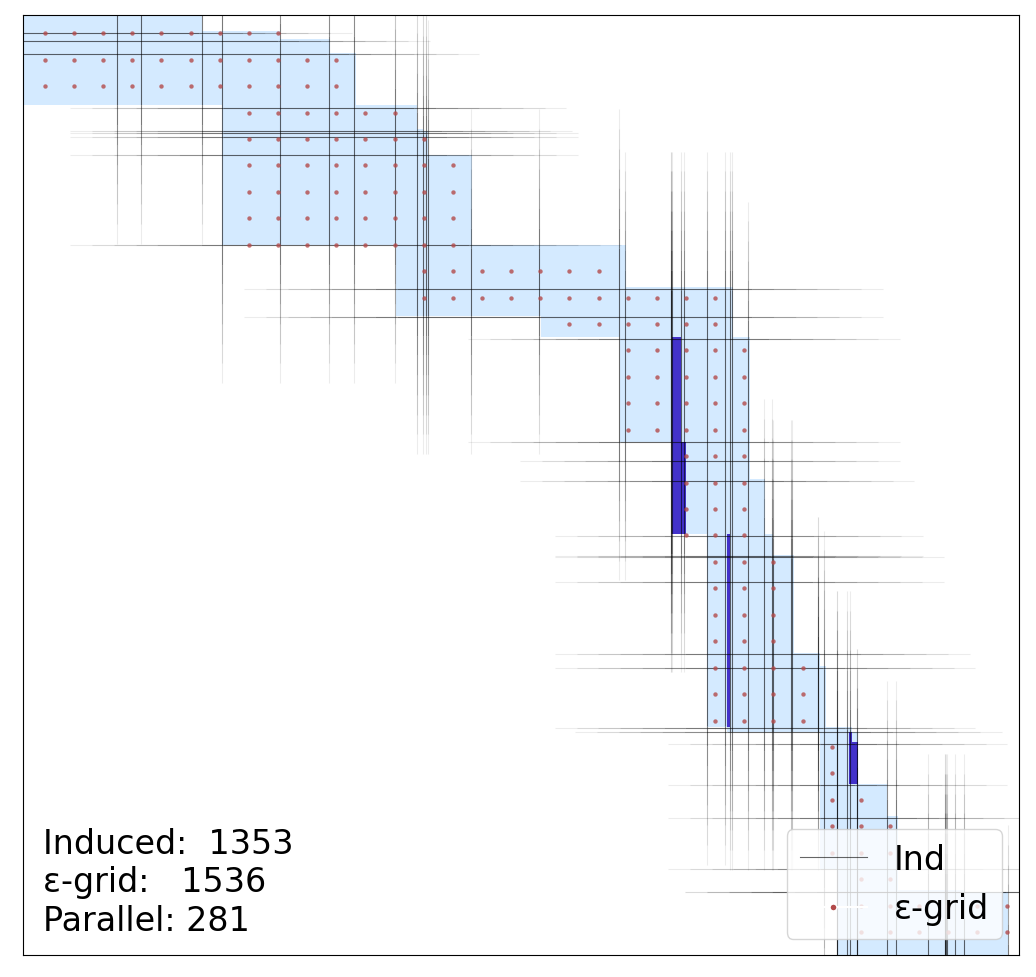}
\end{minipage}
\caption{Enlarged detail of the Hilbert functions of the modules $\module$ and $W$. Left: the induced grid $\grid(V)$ in black lines overlaid on the support and its restriction $\grid(V)_{|\varepsilon \bbZ^2} \cap \supp(V)$ marked with purple crosses for the parallel grid scan. Right: the induced grid $\grid(W)$ indicated by black lines and the $\varepsilon$-grid $\varepsilon \bbZ^2 \cap \supp(W)$ indicated by red dots.\\
Bottom left corners: number of points in the overlaid grids.}
\label{fig:induced_restricted}
\end{figure}
Observe how there are, in some places on the right, many points of the $\varepsilon$-grid in one cell of the induced grid. Conversely, there are sometimes clusters of close lines that will disappear in the restriction on the left.
By counting the number of grid points on the support, we get a precise count of how often the main loops in \autoref{alg:skeleton}, \autoref{alg:full_hnf}, and \autoref{alg:parallel} will be executed. As we can see, for these examples, the $\varepsilon$-grid and the induced grid are of very similar sizes, but the restricted grid is much smaller in both cases.
\end{example}

Since we first decompose any input module $\module$ into its components $\bigoplus \module_i$, the induced grids $\grid(V_i)$ and their restrictions $\grid(V_i)_{|\varepsilon \bbZ^2}$ become particularly small, as depicted in  \autoref{fig:induced_restricted}. 
It is now cheaper to compute all decompositions and slope subdivisions only for the points on each $\grid(V_i)_{|\varepsilon \bbZ^2}$ and then collect the filtration factors and their associated slopes for each point on $\varepsilon \bbZ^2$ where we really need them.
\begin{algorithm}[H]
\caption{Parallel Grid Scan}
\label{alg:parallel}
\DontPrintSemicolon
\KwIn{$\varepsilon \in \bbR$; a graded matrix $M \in \K^{G \times R}$ presenting $V \in \Persfpb{\bbR^2}$.} 
\KwOut{$\forall \beta \in \varepsilon \bbZ^2 \cap \supp(\module)$, a list of pairs $\udim W, \mu\in \text{fun}(\bbR^2, \bbN) \times \bbR$}
Decompose $\module$ into $\bigoplus_{i \in I} V_i$\;
\For{ $i \in I$}{
    $\gri_i \gets \grid(V_i)_{|\varepsilon \bbZ^2} \cap \supp(V_i)$ \;
    \For{$\alpha \in \gri_i$}{
        Compute $\langle V_{i,\alpha} \rangle$ with \autoref{alg:submodule} \;
        Decompose $\langle V_{i,\alpha} \rangle \iso  \bigoplus_{j \in J_{i, \alpha}} V_{i, \alpha, j}$\;
        \For{$j \in J$}{
            Each time when calling \autoref{alg:all_max_slope} in the next step, for each basis $B_t$ also add a set of intervals $\Int_t \subset \mathcal{P}(\bbR^2)$ to the slope subdivision such that
            $\sum_{\mathcal{I} \in \Int_t} \mathds{1}_{\mathcal{I}} = \udim \langle B_t\rangle$ \;
            $\left((T_j, r_j, \rho_j), (S_t, W_t, \Int_t, \tilde p_t)_{t \in T_j} \right) \gets$ \autoref{alg:exact_hnf} $( V_{i, \alpha, j}, C_\alpha)$ \;
        }
    }
}

% Step 3: iterate over epsilon-grid points
\For{$\beta \in \varepsilon \bbZ^2 \cap \supp(\module)$}{
\texttt{Initialise} a list of elements in $\text{fun}(\bbR^2, \bbN) \times \bbR$, $L_\beta \gets \{ \}$ \;
\For{$i \in I$}{
Locate the largest lower bound $\alpha_i \in \gri_i$ with $\alpha_i \leq \beta$ \;
\For{$j \in J_{i, \alpha_i}$}{
$t_1 \dots t_l \gets$ path of $\beta$ in the nested subdivision $\left((T_j, r_j, \rho_j), (S_t)_{t \in T_j} \right)$ \;
\For{$k = 1 \dots l$}{
$L_\beta$.append($(\Int_{t_k})_{|\langle \beta \rangle}, \tilde p_{t_k}(\beta) + \beta_1 \beta_2$) \;
}
}
}
\Return{$L_\beta$}
}
\end{algorithm}

In practice, we do not want to compute all of these slope subdivisions first, because then we would need to store all of them, which can be memory intensive. Instead, we put a co-lexicographic order on the grid $\varepsilon \bbZ \cap \supp(V)$ and the induced grid. We scan the $\varepsilon$-grid along this order and maintain, for each $i$, a pointer $p_i$ to the nearest lower element in each $\gri_i \coloneqq \grid(V_i)_{|\varepsilon \bbZ^2}$.
Whenever one of the pointers $p_i$ is moved in $x$-direction, we compute a decomposition and slope subdivision with \autoref{alg:exact_hnf}. Whenever it is moved in $y$-direction we delete all previously computed decompositions and slope subdivisions belonging to $V_i$.
By scanning the two grids in \emph{parallel} we can be sure to only store, for each $i$, at most $ |\pi_x (\grid(V_i)_{|\varepsilon \bbZ^2})|$ many decomposed modules and corresponding slope subdivisions.

\section{Experiments}\label{sec:experiments}
All experiments are performed on an i7-1255U Processor without any parallelisation. A timeout was set to $10$ min and each run is averaged over $5$ instances whenever possible. Our software and experiments are built using the \textsc{Persistence-Algebra} C++ library\footnote{https://github.com/JanJend/Skyscraper-Invariant, https://github.com/JanJend/Persistence-Algebra}.

\paragraph{Harder-Narasimhan filtrations.}
To generate indecomposable uniquely generated submodules of non-zero thickness we used synthetically generated noisy samples of a torus embedded in $\bbR^3$ (cf. \autoref{obs:factor_dimension}) of thickness up to $7$ and random presentation matrices.
We test the exhaustive search \autoref{alg:brute_force} and its filter \autoref{prop:criterion}. For thickness $9$ modules, \autoref{alg:brute_force} always ran for more than 10 minutes.

\begin{table}[h]
\centering
\caption{Runtime (ms) of \autoref{alg:brute_force}. Columns indexed by dimension/thickness}
\label{tab:combined_exp}
\begin{tabular}{ll|ccccccccc}
\hline
\textbf{Dataset} & \textbf{Variant} & \textbf{2} & \textbf{3} & \textbf{4} & \textbf{5} & \textbf{6} & \textbf{7} & \textbf{8}\\
\hline
Torus & with filter & 2.2 & 2.8 & 5.5 & 19.5 & 202 & 2703 & -  \\
Torus & w/o filter & 2.8 & 2.8 & 5.3 & 18.8 & 206 & 2732 & -  \\
Random & with filter & 8.1 & 4.8 & 6.4 & 28.9 & 454 & 5888 & 100031 \\
Random & w/o filter & 3.6 & 2.3 & 5.1 & 27.1 & 265 & 3624 & 61499  \\
\hline
\end{tabular}
\end{table}

We conclude that the filter (\autoref{prop:criterion}) for \autoref{alg:brute_force} is barely helpful.

The brute force approach of \autoref{alg:brute_force} has a doubly exponential dependence on the thickness but comparatively no dependence on the size of the presentation. Therefore, any restriction to a grid also does not really change the runtime, which is also why we have not included these in \autoref{tab:combined_exp}. Cheng's algorithm, on the other hand, needs us to restrict the module to a regular grid and is highly sensitive to the size of this grid. We will now test just how much time this takes in practice.

\paragraph{Runtime of Cheng's algorithm}
We have restricted the torus generated modules to equidistant grids and used Cheng's algorithm (\autoref{alg:compl-naive}) over both $\bbF_2$ and $\bbQ$. Since we can only generate modules from bifiltrations over $\bbF_2$, we interpreted them as $\bbQ$-modules when possible.

\begin{table}[H]
\centering
\setlength{\tabcolsep}{3pt}
\caption{Runtimes of \autoref{alg:compl-naive}(ms). An asterisk means that some, and a dash that all instances timed out. Columns are indexed by dimension/thickness, and rows by grid size}
\label{tab:cheng_exp}
\begin{tabular}{c|cccccc}
\hline
\textbf{$\bbQ \  / \  \bbF_2$} & \textbf{2} & \textbf{3} & \textbf{4} & \textbf{5} & \textbf{6} & \textbf{7} \\
\hline
2x2 & 1.9/1.8 & 1.8/3.0 & 1.8/4.3 & 5.0/5.0 & 4.4/10.3 & 11.6/33.4 \\
3x3 & 8.7/18.5 & 15.2/27.4 & 21.8/43.2 & 76.6/157k* & 40.8/3.7k & 53.0/301 \\
4x4 & 113/2.0k & 1.1k/117k* & 412/223* & 720/33.6k & 499/66.6k & 177/38.4k* \\
5x5 & 1.8k/154k* & 11.8k/283* & 36.8k/279k* & 11.3k/- & 5.2k/254k* & 2.7k/- \\
6x6 & 61.5k/- & 205k/- & 323k/- & 39.4k/- & 66.6k/- & 16.1k/- \\
\hline
\end{tabular}
\end{table}

We observe a great variability of runtimes for Cheng's algorithm depending on whether the heuristic introduced at the end of \autoref{sec:cheng} successfully computes and certifies the correctness of all the HN filtrations. The random algorithm from \cite{franks2023shrunk} performs computations on square matrices with $\thick{\module}^2\varepsilon^{-d}r$ rows for a large enough $r>0$. For $\K=\bbQ$ (but not for $\K=\bbF_2$), we have observed in our experiments that $r=1$ is large enough. The runtime of Cheng’s algorithm, when the empirical optimisations described in the final paragraph of \autoref{sec:cheng} are omitted, is significantly worse than that of our implementation, particularly over $\bbF_2$.

\begin{table}[H]
\centering
\setlength{\tabcolsep}{3pt}
\caption{Runtimes of \autoref{alg:compl-naive}(ms) with and without optimisations. A dash means that some instances timed out. Columns are indexed by dimension/thickness, and rows by grid size.}
\label{tab:runtime_methods}
\begin{tabular}{c|ccc}
    \hline
    \multicolumn{4}{c}{\textbf{without optimisations}} \\
    \hline
    \textbf{$\bbQ \ / \ \bbF_2$} & \textbf{2} & \textbf{3} & \textbf{4} \\
    \hline
    2x2 & 5/10k & 44/8.4k & 62/40k \\
    3x3 & 33/- & 563/- & 474/- \\
    4x4 & 357/- & 5.1k/- & 5k/- \\
    \hline
\end{tabular}
\hspace{2em}
\begin{tabular}{c|ccc}
    \hline
    \multicolumn{4}{c}{\textbf{with optimisations}} \\
    \hline
    \textbf{$\bbQ \ / \ \bbF_2$} & \textbf{2} & \textbf{3} & \textbf{4} \\
    \hline
    2x2 & 3/4 & 30/132 & 13/62 \\
    3x3 & 14/25 & 133/4k & 118/1.4k \\
    4x4 & 131/19k & 3.7k/- & 172/24k \\
    \hline
\end{tabular}
\end{table}

\paragraph{Comparison.}  \autoref{tab:cheng_exp} shows the strong dependence  of \autoref{alg:compl-naive} on the size of the grid to which the uniquely generated module has to be restricted. \autoref{alg:brute_force} does not require this restriction but, as visible in \autoref{tab:combined_exp}, it is highly dependent on the thickness of the input. 

The asymptotic runtimes of the algorithms suggest that for practical grids $(> 50 \times 50)$, Cheng's algorithm would outperform the exhaustive search not before $\thick{\module} \geq 20$. At that point, the number of arithmetic operations is too large for any computer to handle. 

\paragraph{Runtime on practical data.}
 We used benchmark data from \cite{benchmark_repo, Vipond21} to compare the parallel grid scan \autoref{alg:parallel} with \autoref{alg:skeleton}. 

\begin{table}[H]
\centering
\caption{Runtime (s) over a $50\times 50$ grid on small point sets: density-alpha \cite{function_delaunay}, multicover bifiltrations \cite{sheehy12multicover, multicover23}  (\autoref{subsec:bifiltr}) from sampled spheres and of random points. Number of points in name.}
\label{tab:small_data}
\begin{tabular}{c|ccccccc}
\hline
\textbf{Data} & Circle & $\Alpha^\gamma$7.5k & $\Alpha^\gamma$15k & $\Alpha^\gamma$30k & $\multicover$46 & $\multicover$96 & $\multicover$175 \\
\hline
n & 100 & 140 & 219 & 333 & 3313 & 9142.8 & 19,618.6 \\
\autoref{alg:parallel} & 0.03 & 0.16 & 0.42 & 1.25 & 0.08 & 0.41 & 1.45 \\
\autoref{alg:skeleton} & 1.87 & 0.47 & 0.41 & 0.76 & 12.68 & 32.97 & 73.54  \\
\hline
\end{tabular}
\end{table}

\autoref{alg:parallel} should not show the same quadratic dependency on the grid size of \autoref{alg:skeleton}. \autoref{tab:large_data} shows this and also that our algorithm computes the Skyscraper Invariant in a practical timeframe for relevant data sizes.

\begin{table}[H]
\centering
\caption{Runtime (s) on 40k points, uniformly sampled, and 4k locations in hypoxic2 (\emph{grid-size}).}
\label{tab:large_data}
\begin{tabular}{c|cccccc}
\hline
\textbf{Data} & uni $200$, $\bH_1$& uni $200$, $\bH_0$ & Hyp $20$ & Hyp $40$  & Hyp $80$  & Hyp $160$ \\
\hline
n & 113042  & 93480  & 9443  & 
    9443  &  9443  & 9443  \\
\autoref{alg:parallel} & 299.7 & 422.0 & 54.8 & 93.4 & 167.5 & 311.3  \\
\hline
\end{tabular}
\end{table}

\section{Applications and Open Problems}\label{sec:conclusion}

\paragraph{Conclusion.}
Our algorithms, together with the filtered Landscapes give researchers a new invariant for topological data analysis that is robust to noise, finer than the rank invariant, efficiently computable, and directly interpretable from a visual representation.

To develop fast algorithms, we studied the wall-and-chambers structure on skyscraper stability conditions for persistence modules and obtained a full description in \autoref{thm:main}. 

Our experiments with Cheng's algorithm empirically show that higher blow-up sizes are necessary for finite fields compared to infinite ones, confirming the intuition behind the results of \cite{ivanyos2018constructive}.  

\subsection{Filtered Landscapes} 
The Multiparameter Landscape \cite{vipond2018multiparameter} is directly defined by the ranks of the diagonal maps $V_{\alpha} \to V_{\alpha + (t,t)}$. Therefore, we can filter the Landscape by replacing the rank with the Skyscraper Invariant. As an example, we will use the IHC stained sample LargeHypoxicRegion2 of a tumour\footnote{https://github.com/MultiparameterTDAHistology/SpatialPatterningOfImmuneCells}, for which the authors of \cite{Vipond21} analyse the location of immune cells using Multiparameter Landscapes. 

\begin{figure}[h]
    \begin{minipage}[c]{0.48\textwidth}
        \centering
        \includegraphics[width= \textwidth]{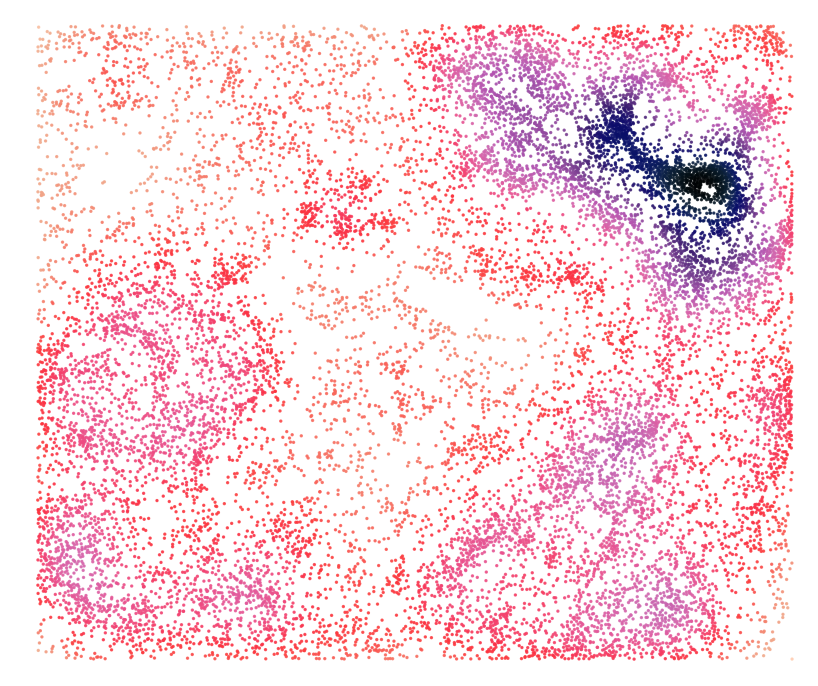}
    \end{minipage}
    \begin{minipage}[c]{0.49\textwidth}
        \centering
        \includegraphics[width= \textwidth]{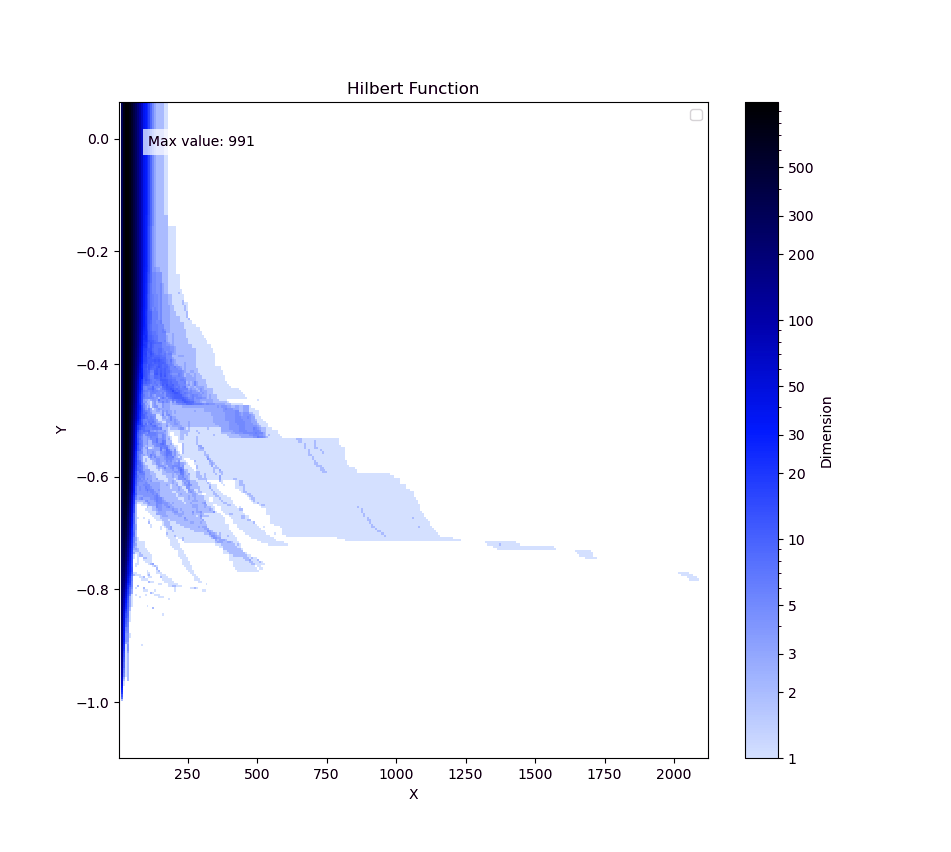}
    \end{minipage}
    \caption{Immune cells in tumor tissue with density estimate and Hilbert function of $H_1$. }
    \label{fig:hypoxic}
\end{figure}

Using \textsc{multipers} \cite{multipers}, we computed its density-alpha bifiltration with \textsc{function-delaunay} \cite{function_delaunay}, a minimal presentation of the first persistent homology group thereof with \textsc{mpfree} \cite{mpfree}, and restricted it to a $300\times 300$ grid. 
 Computing the Skyscraper Invariant on this $35360 \times 25553$ matrix took around 5 minutes. 

\begin{figure}[H]
\resizebox{\textwidth}{!}{\input{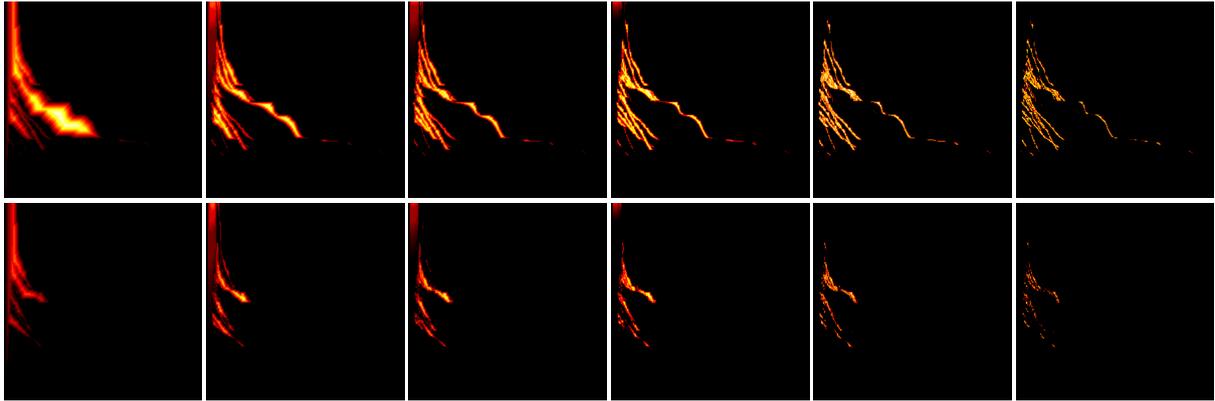}}
\caption{The filtered Landscape for $\bH_1$ of \autoref{fig:hypoxic} with $\theta = 0, 450, 900, \dots, 7200$ and $k=1,2$.}
\label{fig:landscape}
\end{figure}

Since the filtered Landscapes are stable \cite[Corollary 3.12]{fersztand2024harder}, their differences are, too, revealing elements of \emph{long} lifetime. Observe, e.g. the three strands in the fourth picture.

\begin{figure}[H]
    \centering
    \begin{minipage}[c]{0.24\textwidth}
        \centering
        \includegraphics[scale=0.32]{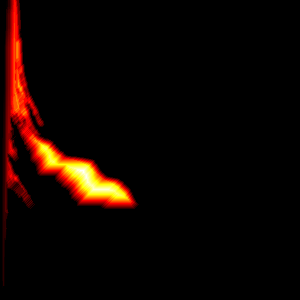}
    \end{minipage}
    \hfill
    \begin{minipage}[c]{0.24\textwidth}
        \centering
        \includegraphics[scale=0.32]{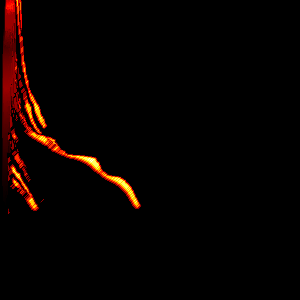}
    \end{minipage}
    \hfill
    \begin{minipage}[c]{0.24\textwidth}
        \centering
        \includegraphics[scale=0.32]{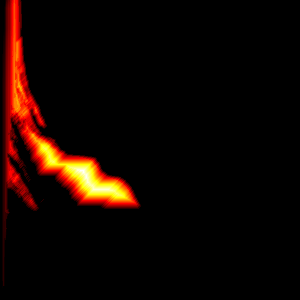}
    \end{minipage}
    \hfill
    \begin{minipage}[c]{0.24\textwidth}
        \centering
        \includegraphics[scale=0.32]{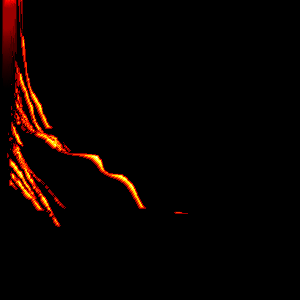}
    \end{minipage}
    \caption{The differences $L_0-L_1$, $L_1-L_2$, $L_0-L_2$, and $L_2-L_4$ from the first row.}
    \label{fig:diff}
\end{figure}

\subsection{Open Problems.}
\paragraph{Improving the runtime.}
\autoref{fig:grid} and \autoref{fig:induced_restricted} make it clear that, in practice, we are still computing the HN filtration too often. In \autoref{rmk:generalize_transfer_fun_partition}, we outline a strategy to prove that it is enough to compute the slope subdivision only for the equivalence classes of $\sim_{t, V}$. These are coarser than the grid by \autoref{prop:grid_similar} (ii), but we would probably need a condition that is easy to compute and that guarantees that two points in $\grid(V)$ are related by $\sim_{t, V}$ to experience a significant speed up. 

\paragraph{Differentiability.}
  To incorporate invariants of MPM into Machine Learning pipelines, the authors of \cite{diff} define a framework, which needs to be extended to include the Skyscraper Invariant. 
Instead of Landscapes, one could also filter other vectorizations of the rank invariant.
There are qualitatively different HN filtrations of MPM via changing the definition of the slope \cite{fersztand2024harder}. Can our methods be adapted to this general setting? 

\paragraph{Other Stability Conditions.}
 \autoref{rem:other_stab_cond} points towards a possible generalisation to another stability condition. Consider a positive locally integrable function $f\colon \bbR^d\to\bbR$. One can replace the real part of the stability condition $\int_{\bbR^d}\udim V$ with $\int_{\bbR^d}f \cdot\udim V$. For each $J \subset [d]$, we can compute the slopes $\mu( \langle V_\alpha \rangle_{|F_J} )$ and construct the slope polynomial with these coefficients. Using \autoref{prop:variation}, this polynomial will approximate the slope. If $f$ is Lipschitz or has other bounds on its growth, then we can bound the error of this approximation. If for some parameter value $\alpha$, one slope polynomial is not just minimal, but even minimal up to this error, we have again found the unique submodule of highest slope for the stability condition at $\alpha$.  

\paragraph{Möbius Inversion.}
In \cite{mccleary2022edit}, McCleary and Patel Möbius-inverted the rank invariant (\cite{Rota1964}) to represent it as a signed set of rectangles. 
For each $\theta \in \bbR$ the same can be done with $s^\theta\colon (\bbR^d)^{op} \times \bbR^d\to \bbZ^\op$, yielding a path in the space of signed $d$-cells in $\bbR^d$, $\Hom( \text{Int}(\bbR^d),\bbZ)$. What does it look like? 

For a specific $\theta$, $s^\theta$ is not always represented by a persistence module, but it sometimes is. An example was given in \autoref{ex:s_rank_representable} and we can see that these modules are not finitely presented anymore, but they are finitely encoded over a semi-algebraic set in the sense of Miller's work \cite{Miller25}. 
If $s^\theta$ can be represented by a module, we conjecture that this is always the case. In general, this implies that the Möbius inversion will not be finite. The question is therefore: Is there a signed semi-algebraic barcode or persistence diagram \cite{Kim2021, botnan2021signed} that represents each $s^\theta$ and can this be used to construct an efficient data structure to store the Skyscraper Invariant?

\printbibliography

\appendix

\end{document}